\documentclass[12pt, twoside]{article}
\usepackage{amsmath,amsthm,amssymb}
\usepackage{times}
\usepackage{enumerate}

\def\titlerunning#1{\gdef\titrun{#1}}
\makeatletter
\def\author#1{\gdef\autrun{\def\and{\unskip, }#1}\gdef\@author{#1}}
\def\address#1{{\def\and{\\\hspace*{18pt}}\renewcommand{\thefootnote}{}%
\footnote {#1}}%
\markboth{\autrun}{\titrun}}
\makeatother
\def\email#1{e-mail: #1}
\def\subjclass#1{{\renewcommand{\thefootnote}{}%
\footnote{\emph{Mathematics Subject Classification (2010):} #1}}}
\def\keywords#1{\par\medskip
\noindent\textbf{Keywords.} #1}

\newtheorem{thm}{Theorem}[section]
\newtheorem{cor}[thm]{Corollary}
\newtheorem{lem}[thm]{Lemma}

\theoremstyle{definition}
\newtheorem{defin}[thm]{Definition}

\newtheorem{exa}[thm]{Example}

\numberwithin{equation}{section}

\usepackage[initials]{amsrefs}
\usepackage{amsfonts, mathrsfs}
\usepackage{amscd}
\usepackage{fullpage}
\usepackage[labelformat=empty]{subfig}
\usepackage{booktabs} 
\usepackage{faktor}
\usepackage{xypic}
\usepackage[dvips]{graphicx,epsfig}
\usepackage{enumerate}
\usepackage{cleveref}
\usepackage[pdftex]{color}
\usepackage{multirow}

\allowdisplaybreaks[3]

\renewcommand{\Im}{\mathfrak{Im}}

\DeclareMathOperator{\tr}{Tr}

\newcommand{\de}{{\partial}}

\newcommand{\rd}{\mathrm{d}}
\newcommand{\ri}{\mathrm{i}}
\newcommand{\re}{\mathrm{e}}

\newcommand{\bbN}{\mathbb{N}}
\newcommand{\bbK}{\mathbb{K}}

\newcommand{\bbZ}{\mathbb{Z}}

\newcommand{\bbC}{\mathbb{C}}
\newcommand{\bbP}{\mathbb{P}}

\def\bary{\begin{array}} 
\def\eary{\end{array}} 
\def\ben{\begin{enumerate}} 
\def\een{\end{enumerate}}
\def\bit{\begin{itemize}} 
\def\eit{\end{itemize}}
\def\nn{\nonumber} 

\newcommand{\cO}{\mathcal{O}}
\newcommand{\cT}{\mathcal{T}}

\newcommand{\cP}{\mathcal{P}}

\newcommand{\DD}{\mathcal{D}}
\newcommand{\LL}{\mathcal{L}}
\newcommand{\cS}{\mathcal{S}}

\newcommand{\cU}{\mathcal{U}}
\newcommand{\cA}{\mathcal{A}}
\newcommand{\HH}{\mathcal{H}}
\newcommand{\cB}{\mathcal{B}}
\newcommand{\cF}{\mathcal{F}}
\newcommand{\cI}{\mathcal{I}}

\newcommand{\cX}{\mathcal{X}}

\newcommand{\cM}{\mathcal M}

\newcommand{\eq}[1]{\begin{equation}#1\end{equation}}
\newcommand{\ea}[1]{\begin{align}#1\end{align}}

\def\beq{\begin{equation}}                     %
\def\eeq{\end{equation}}                       %
\def\bea{\begin{eqnarray}}                     
\def\eea{\end{eqnarray}}
\def\bary{\begin{array}} 
\def\eary{\end{array}} 
\def\ben{\begin{enumerate}} 
\def\een{\end{enumerate}}
\def\bit{\begin{itemize}} 
\def\eit{\end{itemize}}
\def\nn{\nonumber} 
\def\de {\partial}

\def\a{\alpha}
\def\b{\beta}
\def\g{\gamma}

\def\Res{\mathrm{Res}}

\theoremstyle{plain}

\newtheorem{prop}[thm]{Proposition}
\newtheorem{conj}[thm]{Conjecture}

\newtheorem*{conj*}{Conjecture}
\newtheorem*{cor*}{Corollary}

\theoremstyle{definition}

\newtheorem{rmk}{Remark}[section]
\newtheorem{examplen}{Example}[section]

\crefname{equation}{Eq.}{Eqs.}
\crefname{eqnarray}{Eq.}{Eqs.}
\crefname{defin}{Definition}{Definitions}
\crefname{conj}{Conjecture}{Conjectures}
\crefname{lem}{Lemma}{Lemmas}
\crefname{thm}{Theorem}{Theorems}
\crefname{rmk}{Remark}{Remarks}
\crefname{prop}{Proposition}{Propositions}
\crefname{section}{Section}{Sections}
\crefname{appendix}{Appendix}{Appendices}
\crefname{cor}{Corollary}{Corollaries}
\crefname{figure}{Figure}{Figures}

\newcommand{\Li}{\operatorname{Li}}

\newcommand{\GIT}[1]{/\!\!/_{\kern-.2em #1 \kern0.1em}}

\newcommand{\sgn}{\mathrm{sgn}}
\renewcommand{\l}{\left}
\renewcommand{\r}{\right}
\newcommand{\bra}{\left\langle}
\newcommand{\ket}{\right\rangle}

\newcommand{\ev}{\operatorname{ev}}

\def\bes{\begin{subequations}}
\def\ees{\end{subequations}}

\begin{document}

\hyphenation{two-di-men-sio-nal}


\baselineskip=17pt


\titlerunning{Rational reductions of the 2D-Toda hierarchy}
\title{Rational reductions of the 2D-Toda hierarchy and mirror symmetry}

\author{Andrea Brini
\and
Guido Carlet
\and
Stefano Romano
\and 
Paolo Rossi
}

\date{}

\maketitle

\address{A. Brini: Institut Montp\'ellierain Alexander Grothendieck, UMR 5149 CNRS,
Universit\'e de Montpellier, B\^atiment 9, Place Eug\`ene Bataillon,
Montpellier Cedex 5, France;
\email{andrea.brini@univ-montp2.fr}
\and
G. Carlet: Korteweg--de Vries Institute for Mathematics, University of Amsterdam,
Postbus 94248, 1090 GE Amsterdam, The Netherlands;
\email{guido.carlet@uva.nl}
\and
S. Romano: IRMP, Universit\'e Catholique de Louvain, Chemin du Cyclotron 2, 1348 Louvain-la-Neuve, Belgium;
\email{stefano.romano83@gmail.com}
\and
P. Rossi: IMB UMR5584, CNRS, Universit\'e Bourgogne Franche-Comt\'e, F-21000 Dijon, France;
\email{paolo.rossi@u-bourgogne.fr}
}

\subjclass{81T45 (primary), 81T30, 57M27, 17B37, 14N35}

\begin{abstract}

We introduce and study a two-parameter family of symmetry reductions of the
2D Toda lattice hierarchy, which are characterized by a rational factorization of the Lax
operator into a product of an upper diagonal and the inverse of a lower
diagonal formal difference
operator. They subsume and generalize several classical $1+1$ integrable
hierarchies, such as the bigraded Toda hierarchy, the Ablowitz--Ladik hierarchy and E.~Frenkel's $q$-deformed
Gelfand--Dickey hierarchy. We establish their characterization in terms of block T\"oplitz matrices for the associated
factorization problem, and study their Hamiltonian structure. At the dispersionless level, we
show how the Takasaki--Takebe classical limit gives rise to a family of non-conformal 
Frobenius manifolds with flat identity. We use this to generalize the relation of
the Ablowitz--Ladik hierarchy to Gromov--Witten theory by proving an analogous mirror theorem for the general rational
reduction: in particular, we show that the dual-type Frobenius manifolds we obtain
are isomorphic to the equivariant quantum cohomology of a family of toric
Calabi--Yau threefolds obtained from minimal resolutions of the local orbifold line.

\keywords{Rational reductions, Gromov--Witten, integrable hierarchies, mirror symmetry, 2D-Toda, Ablowitz--Ladik.}
\end{abstract}

\section{Introduction}

The two-dimensional Toda equation,
\eq{
(\de^2_x-\de^2_t) x_n = \re^{x_{n+1}}- 2 \re^{x_n}+\re^{x_{n-1}}, \quad n \in \bbZ,
\label{eq:todaeq}
}
is among the archetypical examples in classical field theory of integrable
non-linear dynamical systems in two space dimensions. Besides its intrinsic
interest in the theory of integrable systems
\cite{mikhailov1979integrability,MR810623,MR2753798,Takasaki:1991zs}, the hierarchy of
commuting flows of \cref{eq:todaeq} - the so-called {\it 2D-Toda hierarchy} -  has provided a
unifying framework for a variety of problems in various branches of
Mathematics and Mathematical Physics, ranging from the combinatorics of matrix
integrals \cite{Gerasimov:1990is,MR1794352} to enumerative geometry \cite{MR2199226,pauljohnsonthesis} and applications to
Classical and Quantum Physics
\cite{MineevWeinstein:2000ie,Eguchi:1994in,Nakatsu:2007dk}.\\

 The purpose of this paper is to construct and study an infinite family of symmetry reductions of the
two-dimensional Toda hierarchy, which we dub the the {\it rational reductions of
  2D-Toda} (henceforth, RR2T). Their defining feature is the following
factorization property of the 2D-Toda Lax
operators:
\eq{
L_1^a=AB^{-1}, \quad L_2^b=BA^{-1},
}
where $A$ and $B$ are respectively a degree $a\geq 1$
upper diagonal and a degree $b\geq 1$ lower diagonal
difference operator; this property is preserved by the Toda flows. It turns
out that the resulting hierarchies enjoy remarkable properties both from the
point of view of the theory of integrable systems, as well as from the vantage
of their applications to the topology of moduli spaces of stable maps.
\\

\subsection{Main results}

The RR2T, which are the natural counterpart in the
2D-Toda world of the ``constrained reductions'' of the KP hierarchy of
\cite{Bonora:1993fa, Aratyn:1993ry}, are distinguished in a number of ways. First
off, the embedding into the Toda
hierarchy recovers and ties together a host of known classical
integrable hierarchies in 1+1 dimensions: notable examples include the
Ablowitz--Ladik system \cite{MR0377223, Brini:2011ff}, the bi-graded Toda
hierarchy \cite{MR2246697}, and
the $q$-deformed version of the Gelfand--Dickey hierarchy
\cite{MR1383952}. Moreover, rational reductions have a natural
characterization in the associated factorization problem, where they
correspond to the block T\"oplitz condition on the moment matrix; in the
semi-infinite case this naturally generalizes the
ordinary T\"oplitz condition arising in the theory of unitary matrix models.
Thirdly, the analysis of the
relation of the Hamiltonian structure on the reduced system to the (second) Poisson
structure of the parent 2+1 hierarchy reveals that the reduction itself is
remarkable in that it is a {\it purely kinematical phenomenon}, whose ultimate cause is completely independent of the
particular form of the Hamiltonians: the submanifold in field space where the
Lax operator factorizes comes along with an infinite-dimensional degeneration of the
Poisson tensor, whose pointwise kernel contains the conormal fibers to
the factorization locus. Fourthly,
the semi-classical Lax--Sato formalism for the dispersionless limit of the
hierarchy gives rise to a host of (old and new) solutions of
WDVV in the form of a family of semi-simple, non-conformal Frobenius dual-type
structures on a genus zero double
Hurwitz space\footnote{We
  borrow terminology from \cite{2012arXiv1210.2312R}.} having covariantly
constant identity. For $b\leq 1$, they
are {\it bona fide} dual in the sense of Dubrovin \cite{MR2070050} of
conformal Frobenius manifolds of charge $d=1$, with possibly non-flat unit. The double Hurwitz space picture entails, on one hand, the existence of a bi-Hamiltonian
structure of Dubrovin--Novikov type at the dispersionless level for several subcases,
as well as a tri-Hamiltonian structure as in \cite{pavlov2003tri,
  2012arXiv1210.2312R} for $a=b$; on the other, it furnishes for all
$(a,b)$ a one-dimensional $B$-model-type Landau--Ginzburg description for the dual-type Frobenius structure. Generalizing a result of \cite{Brini:2011ff}, we show that the
resulting non-conformal Frobenius manifolds are isomorphic to the
$(\bbC^*)^2$-equivariant orbifold cohomology of the local $\bbP^1$-orbifolds
with two stacky points of order $a$ and $b$ \cite{johnson2008notes}, or equivalently \cite{MR2529944}, of
the $(\bbC^*)^2$-equivariant cohomology of one of their toric minimal
resolutions (the {\it toric trees}).
This establishes a (novel) version of equivariant mirror symmetry for these
targets via one-dimensional logarithmic Landau--Ginzburg models, which has various
applications to the study of wall-crossings in Gromov--Witten theory as
anticipated in \cite{Brini:2013zsa}, and it leads us to conjecture that the full
descendent Gromov--Witten potential for these targets is a tau function of the
RR2T, a statement that we verify in genus less than or equal to one. \\

The paper is organized as follows. In \cref{sec:todagen}, after reviewing the
Lax formalism for the 2D-Toda hierarchy, we first construct the RR2T in the
bi-infinite case, study the reduction of the 2D-Toda flows, and discuss
various examples. We then illustrate their relation to biorthogonal ensembles
on the unit circle and the factorization problem of block T\"oplitz matrices,
and discuss the Hamiltonian structure of the hierarchy. \cref{sec:dless} is
devoted to the study of the dispersionless limit of the flows. We analyze
the Takasaki--Takebe limit of the equations in the framework of Frobenius
structures on double Hurwitz spaces and determine explicitly the dual-type
structures that arise, as well as the extra flat structures that occur in
special cases. Finally, \cref{sec:GW} is devoted to the relation with
Gromov--Witten theory. We prove an equivariant mirror theorem for toric trees,
and outline the range of its implications. First of all, we 
verify up to genus one that the full descendent Gromov--Witten potential is a
tau function of the RR2T, upon establishing a Miura equivalence between the
dispersive expansion of the RR2T to quadratic order and the analogue of
the Dubrovin--Zhang quasi-Miura formalism applied to the local theory of the
orbifold line. Moreover, we discuss in detail
the properties of the A-model Dubrovin connection in the light of its
connection with RR2T, prove that its flat sections are multi-variate
hypergeometric functions of type $F_D$, and discuss its implications for the
Crepant Resolution Conjecture at higher genus.

\subsection{Relation to other work}

Several instances of RR2T have made a more or less covert appearance in
the literature. In a prescient work \cite{MR1091265}, Gibbons and
Kupershmidt\footnote{Building on earlier work of Bruschi--Ragnisco \cite{MR979202}; see also \cite{MR1993935,Kharchev:1996kh}.} 
constructed a Lax formalism for a relativistic generalization of the one-dimensional Toda
hierarchy which would correspond in our language to the RR2T of bidegree $(a,1)$, where the
dependent variable in the denominator has been frozen to a parameter equal to
the speed of light. More recent examples include the Ablowitz--Ladik hierarchy treated
by the authors \cite{Brini:2011ff}, corresponding to the case $(a,b)=(1,1)$, and the somewhat degenerate example of the lattice
analogue of KdV \cite{MR0377223}, to which RR2T boils down for
$b=0$. Dual-type structures for the dispersionless limit of the RR2T have
been computed in the special case of the bigraded Toda
hierarchy \cite{MR2287835} and the RR2T of bidegree $(a,a)$ (see also
\cite{MR1804297, Takasaki:2012ba}). Closer to the discussion of \cref{sec:GW} is a very recent preprint
of Takasaki \cite{Takasaki:2013lqa}, where the (full-dispersive) RR2T of bidegree $(b,b)$ with
suitable initial data is considered in connection with the partition function
of the melting crystal model \cite{Okounkov:2003sp} for the so-called
``generalized conifolds'' deformed by shift symmetries
\cite{Takasaki:2013gja}. As the generalized conifolds correspond precisely to the toric
Calabi--Yau threefolds of \cref{sec:GW} for $a=b$, it would be intriguing to
bridge Takasaki's approach with our own, and in particular to intepret the
2D-Toda evolution in the crystal model 
as suitable gravitational deformations of our prepotentials. We will leave
this open for future work.

\section{Rational reductions of 2D-Toda}

\label{sec:todagen}

\subsection{The 2D-Toda hierarchy}
\label{sec:2dtoda}
Denote by $\cA=\{(a_{ij}\in\bbC)_{i,j\in\bbZ}\}$ the vector space of
doubly-infinite matrices with complex coefficients. Equivalently, this is the space
of formal
difference operators $\sum_{r\in\bbZ} a_r \Lambda^r$ where $a_r$ for every $r$
is an element
of the space $\mathscr{F}$ of 
$\bbC$-valued functions on $\bbZ$, and the shift
operator $\Lambda$ acts on $f\in\mathscr{F}$ by $\Lambda^k f(n) = f(n+k)$. For 
$\Delta = \sum_{r\in\bbZ} a_r \Lambda^r \in \cA$, the $\bbC$-linear projections
\ea{
\Delta_+ = & \sum_{r \in \bbZ^+} a_r \Lambda^r , \\
\Delta_- = & \sum_{r \in \bbZ_0^-} a_r \Lambda^r . 
}
define a canonical decomposition $\cA=\cA_+\oplus \cA_-$, 
corresponding to the projections of $\Delta$ to its upper/strictly lower triangular
part. We will denote by $\Delta^T$ its transpose
\eq{
\Delta^T =  \sum_{r \in \bbZ}  \Lambda^{-r} a_r
}
and, whenever defined, we denote its positive/negative order $\mathrm{ord}_{\pm}
\Delta$ as the degree of its projections to $\cA_{\pm}$ as formal difference operators,
\eq{
\mathrm{ord}_\pm \Delta = \deg_{\Lambda^{\pm 1}} (\Delta)_\pm.
}
\\

Armed with these definitions, we construct an infinite dimensional dynamical
system over an affine subspace of $\cA\oplus\cA$, as follows. The {\it 2-dimensional Toda lattice} \cite{MR810623} is the
system of commuting flows $(\de_{s_r^{(1)}}, \de_{s_r^{(2)}},r>0)$ given by the Lax equations
\eq{
\label{eq:2dtlax}
\partial_{s_r^{(1)}} L_i =  [(L_1^r)_+ , L_i ] , \quad 
\partial_{s_r^{(2)}} L_i =  [-(L_2^r)_- , L_i ] , \quad i=1,2,
}
where the 2D-Toda Lax operators are the formal difference operators
\eq{
\label{eq:2dtlaxop}
L_1 = \Lambda + \sum_{j\geq0} u_j^{(1)}  \Lambda^{-j}, \quad 
L_2 = \sum_{j\geq -1} u_{j}^{(2)}  \Lambda^{j}.
}
with $u_j^{(k)} \in \mathscr{F}$ for all $j\in \bbN\cup \{-1\}$, $k=1,2$.
Commutativity of these flows follows from the (simplified form of) the zero-curvature equations
\eq{
\label{eq:commlemm}
\partial_{s_q^{(j)}} L_i^r - \partial_{s_r^{(i)}} L_j^q + [ (L_i^r)_+, (L_j^q)_+ ] - [ (L_i^r)_-, (L_j^q)_-] =0,
} 
which in turn is equivalent to a compatibility condition for the
Zakharov--Shabat spectral problem
\ea{
L_1 \Psi_1 = & w \Psi_1, & L_2^T \Psi_2 = & w \Psi_2, &
\de_{s^{(1)}_q}\Psi_1 = & \l(L_1^q\r)_+ \Psi_1, \nn \\
\de_{s^{(1)}_q}\Psi_2 = & -\l(L_1^q\r)^T_+ \Psi_2^*, &
\de_{s^{(2)}_q}\Psi_1 = & \l(L_2^q\r)^T_- \Psi_1, &
\de_{s^{(2)}_q}\Psi_2 = & -\l(L_2^q\r)^T_- \Psi_2^*.
}
for wave vectors $\Psi_i\in \bbC((w))\otimes \mathscr{F}$, $i=1,2$ \cite{MR810623}. 

An equivalent formulation of the 2D-Toda hierarchy can be given in terms of Sato equations 
\eq{ \label{eq:sato}
\partial_{s_r^{(i)}} S_1 = - (L_i^r)_- S_1, \qquad 
\partial_{s_r^{(i)}} S_2 = - (L_i^r)_- S_2 ,
}
for the dressing operators
\eq{ \label{eq:dressop}
S_1 = 1 + p_1^{(1)} \Lambda^{-1} + \dots, \quad
S_2 = p_0^{(2)} + p_1^{(2)} \Lambda + \dots .
}
The Lax operators are expressed in terms of the dressing operators by
\eq{ \label{eq:lax} 
L_1 = S_1 \Lambda S_1^{-1}, \quad 
L_2 = S_2 \Lambda^{-1} S_2^{-1} ,
}
and the commutativity of the flows $\partial_r^{(i)}$ on $S_i$ again follows from \cref{eq:commlemm}.\\

Under suitable assumptions the initial value problem for the 2D-Toda equation can be solved in terms of a factorization problem~\cite{MR0810626}. Let $\mu\in\cA$ be a matrix depending on the times $s_r^{(i)}$ according to 
\ea{
\frac{\partial \mu}{\partial s_r^{(1)}} = & \Lambda^r \mu  ,\\
\frac{\partial \mu}{\partial s_r^{(2)}} =  & \mu \Lambda^{-r} ,
}
or, equivalently, 
\eq{ \label{eq:mues}
\mu = \exp \big(\sum_{r\geq1} s_r^{(1)} \Lambda^r \big) \mu_0 \exp \big( \sum_{r\geq1} s_r^{(2)} \Lambda^{-r} \big).
}
Assume the factorization 
\eq{ \label{eq:factor}
\mu = S_1^{-1} S_2
}
exists and uniquely determines $S_1$ and $S_2$ as in
\cref{eq:dressop}. Deriving this expression w.r.t. $s_r^{(i)}$ and projecting
it onto $\cA_\pm$ we get that that $S_1$, $S_2$ satisfy the Sato equations
(\cref{eq:sato}), hence the associated Lax operators of \cref{eq:lax} solve
\cref{eq:2dtlax}. In the semi-infinite case the factorization problem can be
directly solved using bi-orthogonal polynomials, as we will show in \cref{sec:biorth}. 

\subsection{The rational reductions}
\label{sec:ratred}
Consider now the difference operators
\ea{
A =& \Lambda^{a} + \alpha_{a-1}  \Lambda^{a-1} + \cdots + \alpha_0 \in
\cA_+, \label{eq:Aop}\\
B =&  1 + \beta_1 \Lambda^{-1} + \cdots + \beta_b \Lambda^{-b} \in 1+\cA_- \label{eq:Bop}
}
for $a,b>0$. We define two factorization maps $L_i : \cA_+\oplus \cA_- \to
\cA$ by
%
\eq{
L_1^{a} = A B^{-1}, \quad L_2^{b} = B A^{-1};
\label{eq:Lbigrad}
}
notice that they give Lax operators in the form of \cref{eq:2dtlaxop}. It is convenient to define also the dual
operators $\widehat{L}_1$, $\widehat{L}_2$ by
\eq{
\widehat{L}_1^{a} = B^{-1} A, \quad \widehat{L}_2^{b} = A^{-1} B.
\label{eq:Lbigrad2}
}
\begin{thm}
\label{prop:red}
For $i=1,2$, $r>0$, the equations
\ea{
&\partial_{s_r^{(i)}} A = (L_i^r)_+ A - A (\widehat{L}_i^r )_+, \label{eq:Aeq} \\
&\partial_{s_r^{(i)}} B = (L_i^r)_+ B - B (\widehat{L}_i^r)_+ \label{eq:Beq}
}
define commutative flows on $A$, $B$ that induce the 2D-Toda Lax equations,~\cref{eq:2dtlax}.
\end{thm}
\begin{proof}
We first check that these flows are well-defined. From
\eq{
A^{-1} L_1 A = ((A^{-1} L_1 A)^a)^{1/a} = (A^{-1} L_1^a A)^{1/a} = (B^{-1} A)^{1/a} = \widehat{L}_1 .
}
we obtain
\eq{
 L_i^r A = A \widehat{L}_i^r,
}
and similarly
\eq{
L_i^r B = B \widehat{L}_i^r. 
}
With the aid of these identities we can rewrite~\cref{eq:Aeq,eq:Beq} as 
\ea{
&\partial_{s_r^{(i)}} A = -(L_i^r)_- A + A (\widehat{L}_i^r )_-, \label{eq:Aeqm} \\
&\partial_{s_r^{(i)}} B = -(L_i^r)_- B + B (\widehat{L}_i^r)_- .\label{eq:Beqm}
}
The r.h.s in both \cref{eq:Aeq,eq:Aeqm} is a difference operator in $\cA_+$ of
order $\mathrm{ord}_+=a-1$, hence 
the flow given by~\cref{eq:Aeq} is well-defined on operators of the
form of \cref{eq:Aop}. Similarly we see that \cref{eq:Beq} gives a well-defined
flow on operators of the form of \cref{eq:Bop}. In general, if $\partial_t A = W A - A \, \widehat{W}$ and $\partial_t B = W B - B \, \widehat{W}$ for some difference operators $W$, $\widehat{W}$, then 
\eq{
\partial_t L_i = [ W, L_i], \quad
\partial_t \widehat{L}_i = [ \widehat{W} , \widehat{L}_i ] .
}
Hence from~\cref{eq:Aeq,eq:Beq,eq:Aeqm,eq:Beqm} it
follows that the operators $L_i$ satisfy the 2D-Toda Lax
equations, \cref{eq:2dtlax}. To prove commutativity, observe that if $\partial_{t_i} A = W^i A - A \, \widehat{W}^i$ for some difference operators $W^i$, $\widehat{W}^i$, $i=1,2$, then
\eq{
\partial_{t_1} \partial_{t_2} A -\partial_{t_2} \partial_{t_1} A = ( W^1_{t_2} - W^2_{t_1} + [W^1,W^2]) A 
-A(\widehat{W}^1_{t_2} -\widehat{W}^2_{t_1} +[\widehat{W}^1,\widehat{W}^2]) .
}
Applying this formula to the flows defined by \cref{eq:Aeq,eq:Beq} we see that the right-hand side vanishes because of~\cref{eq:commlemm}, hence the flows commute. 
\end{proof}

\begin{rmk}
\label{rmk:dual}
Notice that the dual Lax operators also satisfy Lax equations
(\cref{eq:2dtlax}) with $\widehat{L}_i$ instead of $L_i$,
\eq{
 \label{eq:dual2dtlax}
\partial_{s_r^{(1)}} \widehat{L}_i = [(\widehat{L}_1^r)_+ , \widehat{L}_i ] , \quad 
\partial_{s_r^{(2)}} \widehat{L}_i = [-(\widehat{L}_2^r)_- , \widehat{L}_i ] , \quad i=1,2 . 
}
\end{rmk}
\begin{rmk}
The inverses of $A$ and $B$ appearing in Eqs.~\eqref{eq:Lbigrad} and~\eqref{eq:Lbigrad2} are defined as the following upper (resp. lower) diagonal matrices
\begin{equation}
A^{-1} = \sum_{k\geq 0} (1- \alpha_0^{-1} A)^k \alpha_0^{-1} , \quad 
B^{-1} = \sum_{k\geq0} (1-B)^k .
\end{equation}
\end{rmk}

The pairs of matrices of the rational form given by Eq.~\eqref{eq:Lbigrad} form a submanifold of the 2D-Toda phase space of pairs of Lax operators, \cref{eq:2dtlaxop}. 
The previous Theorem shows that, on such submanifold, the 2D-Toda flows coincide with the push-forward under the factorization map, \cref{eq:Lbigrad}, of the vector fields defined by \cref{eq:Aeq,eq:Beq} on the space of pairs $\{(A,B)
\in \cA_+\oplus  \cA_-\}$, where $A$ and $B$ are of the form given by \cref{eq:Aop,eq:Beq}. 
This clearly implies that the submanifold of rational 2D-Toda Lax operators given by \cref{eq:Lbigrad} is invariant under the 2D-Toda flows.

\begin{defin}
A {\rm rational reduction} of the 2D-Toda hierarchy {\rm(RR2T)} of bi-degree $(a,b)$ is the hierarchy of flows induced by the 2D-Toda flows on the invariant subset of matrices of the form ~\eqref{eq:Lbigrad}. 
\end{defin}

We may more generally consider Lax operators of the form
\ea{
L_1 =  \l(\Lambda^{m} A B^{-1} \r)^{1/(a+m)}, \quad L_2 =  \l(B A^{-1}
\Lambda^{-m}\r)^{1/(b+m)}.
\label{eq:equivlax}
}
The same analysis of \cref{prop:red} carries through to this case as well. Notice that in this case the flows in Eq.~\eqref{eq:Aeqm} should be defined in terms of the operator $\hat{A} := \Lambda^m A$, rather than $A$. 

\begin{defin}
Let $(L_1,L_2)$ be as in \cref{eq:equivlax}. The associated reduction of the
2D-Toda lattice hierarchy will be called the m-{\rm generalized RR2T} of bidegree $(a,b)$.
\label{rmk:equivlax}
\end{defin}

\begin{rmk}
We can partially lift the condition that $a,b>0$ by considering the case when
$a=0$ (resp.~$b=0$) as the degenerate situation in which only one half of the
flows given by $\de_{s_r^{(2)}}$ (resp.~$\de_{s_r^{(1)}}$) is defined by \cref{eq:2dtlax,eq:Lbigrad}. All of
the above then carries through to this setting.\\
\end{rmk}

As it turns out, \cref{prop:red} gives rise to a variety of new reductions of the 2D-Toda
hierarchy, incorporating at the same time several known infinite-dimensional
lattice integrable systems.

\setcounter{examplen}{0}

\begin{examplen}[The Ablowitz--Ladik hierarchy]
\label{exa:al}
The Ablowitz--Ladik system~\cite{MR0377223} is a discretization of the
complexified non-linear Schr\"odinger equation given by the second order system
\ea{
\label{eq:AL1}
\ri \dot{x}_n= & -\frac{1}{2}\l(1-x_n y_n\r)\l(x_{n+1}+  x_{n-1}\r)+x_n, \\
\ri \dot{y}_n= & \frac{1}{2}\l(1-x_n y_n\r)\l( y_{n+1}+ y_{n-1}\r)-y_n,
\label{eq:AL2}
}
for $n \in \bbZ$. This system is Hamiltonian, and it possesses an infinite number of local conserved
currents in involution \cite{MR0377223}. As shown in  \cite{Brini:2011ff}, after
work of Adler--van~Moerbeke \cite{MR1794352} and Cafasso \cite{MR2534519} in
the semi-infinite case, its integrability is bequeathed from a rational
embedding into the 2D-Toda hierarchy. Explicitly, introduce lattice
variables $\a, \b\in\mathscr{F}$ through
\ea{
\label{eq:ALa}
\a_n = & -\frac{y_n}{y_{n+1}}, \\
\b_n = & \frac{(1-x_n y_n) y_{n-1}}{y_n}.
\label{eq:Alb}
}
Then \cite{Brini:2011ff} the Ablowitz--Ladik hierarchy is the pull-back under \cref{eq:ALa,eq:Alb}
of the rational reduction of the 2D-Toda flows of bidegree $(a,b)=(1,1)$. 
\end{examplen}

\begin{examplen}[The $q$-deformed Gelfand--Dickey hierarchy]
\label{exa:qdef}
Denote by $D_q$ the scaling ($q$-difference) operator on the real line,
$D_qf(x)=f(x q)$, and  write $Q_\pm$ for the projection of a $q$-difference operator $Q$ onto
its $q$-differential/strictly $q$-pseudo-differential part. Lax
equations in the form
\ea{
\de_{t_m} \mathfrak{L}= \l[\mathfrak{L},\l(\mathfrak{L}^m\r)_+\r]
\label{eq:qlax1}
}
for the $q$-pseudodifference operator
\eq{
\mathfrak{L} \triangleq D_q+\sum_{j\geq 0}u_j(x) D_q^{-j}.
}
were proposed by E.~Frenkel in \cite{MR1383952} as a $q$-analogue of the KP hierarchy. In
particular, the natural $q$-analogue of the Gelfand--Dickey ($n$-KdV) hierarchy
\eq{
\mathfrak{L}^{n+1} = D_q^{n+1}+\sum_{j\geq 1}^{n} \tau_j(x) D_q^j,
\label{eq:qlax2}
}
give rise to a completely integrable bi-Hamiltonian system . Rewriting the
$q$-difference Lax equations \cref{eq:qlax1,eq:qlax2} as ordinary Lax equations for
a discrete operator $L$ \cite{MR1626871}, the system \cref{eq:qlax1} can be recast in the
form of a reduction of the 2D-Toda flows under the constraint
\eq{
(L^{n+1})_-=0.
}
This corresponds to the RR2T of bidigree $(a,b)=(n+1,0)$.

\end{examplen}

\begin{examplen}[The bi-graded Toda hierarchy]

The bi-graded Toda lattice hierarchy of \cite{MR2246697} also enjoys a
representation as a (generalized) RR2T. By \cref{eq:equivlax,eq:Lbigrad}, the Lax operator for $(N,M)$ bi-graded
Toda 
\eq{
L = \Lambda^N+u_{N-1}\Lambda^{N-1}+\dots + u_{-M}\Lambda^{-M}
}
indeed corresponds to the Lax operator $L_1^{N+M}$ for the $-M$-generalized RR2T of bidegree
$(N+M,0)$. Notice that in this formulation we can only recover as reductions of the 2D-Toda flows only the standard flows and not the extended or logarithmic ones. 

\end{examplen}

\subsection{Rational reductions and the factorization problem}

It is illuminating to consider the form of the constraint leading to the RR2T
at the level of dressing operators. By \cref{rmk:dual}, the dual Lax operators
$\widehat{L}_i$ satisfy the 2D-Toda Lax equations,
\cref{eq:dual2dtlax}. Introducing the corresponding 2D-Toda dressing operators
$\widehat{S}_i$ as in \cref{eq:dressop,eq:sato}, which satisfy the Sato equations
\eq{ \label{eq:sato2}
\partial_{s_r^{(i)}} \widehat{S}_1 = - (\widehat{L}_i^r)_- \widehat{S}_1, \qquad 
\partial_{s_r^{(i)}} \widehat{S}_2 = - (\widehat{L}_i^r)_- \widehat{S}_2 ,
}
 the RR2T of bidegree $(a,b)$ can
be translated into the pair of constraints
\begin{subequations}\label{eq:dressred}
\begin{align} 
&S_1 \Lambda^a \widehat{S}_1^{-1} = S_2 \widehat{S}_2^{-1} \triangleq A, \label{eq:dressreda}\\
&S_1 \widehat{S}_1^{-1} = S_2 \Lambda^{-b} \widehat{S}_2^{-1} \triangleq B. \label{eq:dressredb}
\end{align}
\end{subequations}
%

\begin{prop}
The constraints given by Eq.~\eqref{eq:dressred} are preserved by the Sato equations for $S_i$, $\hat{S}_i$, hence define a reduction of 2D-Toda at the level of dressing operators that corresponds to the rational reduction of bi-degree $(a,b)$.
\end{prop}
\begin{proof}
Notice that this case the operators $A$, $B$ arise naturally as a combination of  the dressing operators of two copies of the 2D-Toda hierarchy. 
Clearly ~\eqref{eq:dressred} implies that the operators $A$, $B$ are of the form~\eqref{eq:Aop},~\eqref{eq:Bop}. The corresponding Lax operators $L_i$, $\hat{L}_i$, defined through~\eqref{eq:lax}, factorize as in~\eqref{eq:Lbigrad},~\eqref{eq:Lbigrad2}, i.e.
\begin{equation}
L_1^a = S_1 \Lambda^a S_1^{-1} = S_1 \Lambda^a \hat{S}_1^{-1} \cdot\hat{S}_1 S_1^{-1} = A B^{-1} , \text{ etc...}
\end{equation}
and Sato equations induce the flows~\eqref{eq:Aeq},~\eqref{eq:Beq}. It follows that the constraints~\eqref{eq:dressred} are preserved by the Sato equations. 
\end{proof}

As the simplest non-trivial rational reduction of the 2D-Toda hierarchy gives rise to the Ablowitz--Ladik hierarchy \cite{Brini:2011ff}, which is in turn related to a factorization problem of a T\"oplitz moment matrix,  it is natural to ask whether the generic rational reduction may be interpreted in the same way. 

\begin{defin}
We say that $\mu\in\cA$ is a \emph{block T\"oplitz operator} of bi-degree $(a,b)$ if 
\eq{
\label{eq:bltop}
\Lambda^a \mu \Lambda^{-b} = \mu .
}
\end{defin}
Equivalently, its matrix entries satisfy $\mu_{i+a,j+b}=\mu_{ij}$, which reduces to the usual T\"oplitz condition when
$a=b=1$. Clearly the property of being block T\"oplitz of bi-degree $(a,b)$ is
preserved by the time evolution as in \cref{eq:mues}.\\
Let now $(\mu_{ij})_{i,j\in\bbZ}$ be a block T\"oplitz matrix of bi-degree
$(a,b)$ depending on the times $s_r^{(i)}$ as in \cref{eq:mues} and such that
the factorization problems 
\begin{subequations}
\begin{align}
&\mu = S_1^{-1} S_2, \label{eq:fac1} \\
&\mu \Lambda^{-b} = \widehat{S}_1^{-1} \widehat{S}_2, \label{eq:fac2}
\end{align}
\end{subequations}
admit solutions for $S_i$, $\widehat{S}_i$ of the form of \cref{eq:dressop}. We have
the following
\begin{prop} 
\label{prop:propS}
The dressing matrices $S_i$, $\widehat{S}_i$ satisfy the Sato equations with
the constraints in \cref{eq:dressred}. The corresponding Lax operators
(\cref{eq:lax}) give a solution of the RR2T of bidegree $(a,b)$.
\end{prop}
\begin{proof}
By substituting
\cref{eq:fac1} into \cref{eq:fac2} we get
\eq{
S_1^{-1} S_2 \Lambda^{-b} = \widehat{S}_1^{-1} \widehat{S}_2 .
}
Left-multiplication by $S_1$ and right-multiplication by $\widehat{S}_2^{-1}$ give
\cref{eq:dressredb}. By the block T\"oplitz property, \cref{eq:bltop}, we can
rewrite \cref{eq:fac2} as
\eq{
\Lambda^{-a} \mu = \widehat{S}_1^{-1} \widehat{S}_2.
}
Performing the same substitution as before and rearranging the terms we obtain
\cref{eq:dressreda}.
\end{proof}

\subsection[Semi-infinite block T\"oplitz matrices and bi-orthogonal
  polynomials]{Semi-infinite block T\"oplitz matrices and bi-orthogonal
  polynomials on the unit circle}
\label{sec:biorth}
All statements of the previous sections can be transferred almost 
verbatim to the so-called semi-infinite case, given by
the algebra $\cA^{\frac{\infty}{2}}=\{(a_{ij}\in\bbC)_{i,j\in\bbZ_{\geq 0}}\}$
of complex semi-infinite matrices. In this case $\Lambda$ and $\Lambda^{-1}$ denote the semi-infinite matrices
\eq{
(\Lambda)_{ij} := \delta_{i+1,j}, \quad
(\Lambda^{-1})_{ij} := (\Lambda^T)_{ij} = \delta_{i,j+1}
}
Here, with an abuse of notation, we denote by $\Lambda^{-1}$ the transpose of
$\Lambda$, which is in fact only a right inverse of $\Lambda$. We have
\eq{
\Lambda^{-1} \Lambda= 1 - E_{11}
}
where $(E_{11})_{ij}=\delta_{i,0} \delta_{j,0}$. \\

\subsubsection{The factorization problem for 2D-Toda and bi-orthogonal polynomials}
In the semi-infinite case and for generic initial data for the 2D-Toda flows, a sufficient condition for the existence of the
factorization of \cref{eq:factor} is given by Gauss' elimination: if all the
leading principal minors of $\mu \in \cA^{\frac{\infty}{2}}$ are non-zero, this leads
to the LU decomposition of \cref{eq:factor}. The factorization problem can
then be interpreted as the construction of bi-orthogonal polynomials with
respect to the bilinear form $\bra,\ket_\mu$ associated to $\mu$.
More precisely, let $\mu\in \cA^{\frac{\infty}{2}}$ and let $\bra,\ket_\mu$ be
the $\bbC$-bilinear form on $\bbC[z]$ defined by
\eq{
\bra z^i , z^j \ket_\mu = \mu_{ij} .
}
Let $p_j^{(i)}(z)$, $i=1,2$, $j\geq0$ be monic polynomials in $\bbC[z]$ of
degree $j$. The factorization problem for $\mu$ is equivalent to the
requirement that $p_j^{(i)}(z)$ form a bi-orthogonal basis in $\bbC[z]$ w.r.t
$\bra,\ket_\mu$ i.e.
\eq{
 \label{eq:bior}
\bra p_i^{(1)} , p_j^{(2)} \ket_\mu = \delta_{ij} h_i .
}
Indeed,  the coefficients of the bi-orthogonal polynomials are related to the
matrices $S_1$, $S_2$ by 
\begin{subequations} \label{eq:orp}
\begin{align}
&p_i^{(1)}(z) = \sum_{k=0}^i (S_1)_{ik} z^k ,\label{eq:orp1}\\
&p_i^{(2)}(z) = h_i \sum_{k=0}^i (S_2^{-1})_{ki} z^k . \label{eq:orp2}
\end{align}
\end{subequations}
The bi-orthogonality property in \cref{eq:bior} turns into
\eq{
S_1 \mu S_2^{-1} h = h
}
i.e. the factorization of the moment matrix, \cref{eq:factor}. Denote now by
$p^{(i)}$ (resp. $\widehat{p}^{(i)}$) the semi-infinite
vector having $p_j^{(i)}$ (resp. $\widehat{p}_j^{(i)}$) as its $j^{\rm th}$
entry. By \cref{eq:factor,eq:lax}, the Lax operators $L_i$ act on bi-orthogonal polynomials as
\begin{align}
&L_1 p^{(1)}(z) = z p^{(1)}(z) , \\
&h L_2^T h^{-1} p^{(2)}(z) = z p^{(2)}(z).
\end{align}

\subsubsection{Semi-infinite block T\"oplitz matrices}
Let us now turn to the study of the $(a,b)$ RR2T in the semi-infinite case,
or, equivalently, to the factorization problem of semi-infinite block
T\"oplitz matrices.  We start by defining two sets of bi-orthogonal polynomials
associated with the $\bbC$-bilinear forms
\begin{align}
&\bra z^i, z^j\ket_\mu = \mu_{ij}, \\
&\bra z^i , z^j \ket_{\widehat\mu} = \widehat\mu_{ij}=\mu_{i,j+b},
\end{align}
where $\widehat{\mu} = \mu \Lambda^{-b}$. Both $\mu$ and $\widehat\mu$ satisfy the
T\"oplitz property, which translates, at the level of bilinear forms, into
\begin{equation}
\bra z^a f(z) , z^b g(b) \ket_\bullet = \bra f(z) , g(z) \ket_\bullet
\end{equation}
for any $f,g \in \bbC[z]$. The monic polynomials $p_j^{(i)}$ and
$\widehat{p}_j^{(i)}$ satisfy the bi-orthogonality conditions
\begin{subequations}
\begin{align}
&\bra p_i^{(1)} , p_j^{(2)} \ket_\mu = \delta_{ij} h_i , \\
&\bra \widehat{p}_i^{(1)} , \widehat{p}_j^{(2)} \ket_{\widehat\mu} = \delta_{ij} \widehat{h}_i .
\end{align}
\end{subequations}
The corresponding dressing matrices $S_i$, $\widehat{S}_i$ are defined through
\cref{eq:orp}; such matrices satisfy the factorization problems of
\cref{eq:fac1,eq:fac2}. If we assume that the moment matrix $\mu$ depends on the times $s_r^{(i)}$ as
in \cref{eq:mues}, then, according to \cref{prop:propS}, $(S_i,\widehat{S}_i)$ give a solution of the $(a,b)$-graded RR2T.

\begin{prop}
The bi-orthogonal polynomials $p_j^{(i)}$ and the dual bi-orthogonal polynomials $\widehat{p}_j^{(i)}$ are related by the following identities
\begin{subequations}
\begin{align}
A \widehat{p}^{(1)} &= z^a p^{(1)} , \label{eq:rel1}\\
\widehat{h} A^T h^{-1} p^{(2)} &= \widehat{p}^{(2)} ,\\
B \widehat{p}^{(1)} &= p^{(1)} ,\\
\widehat{h} B^T h^{-1} p^{(2)} &= z^b \widehat{p}^{(2)} .
\end{align}
\end{subequations}
\end{prop}
\begin{proof}
Let us prove the first relation. Applying $A$ to \cref{eq:orp1} we get 
\begin{equation}
(A \widehat{p}^{(1)} )_i= \sum_{k\geq0} (A \widehat{S}_1)_{ik} z^k,
\end{equation}
where we have used the fact that the sum in \cref{eq:orp1} can be extended to $\infty $ due to the triangular structure of $S_1$. 
 The first part of \cref{eq:dressreda} gives
\begin{equation}
A \widehat{S}_1 = S_1 \Lambda^a ,
\end{equation}
which substituted above gives \cref{eq:rel1}. The remaining relations are proved in a similar way.
\end{proof}

As a straightforward consequence we obtain recursion relations for the bi-orthogonal polynomials $p_j^{(2)}$ and $\widehat{p}_j^{(1)}$.
\begin{cor}
The bi-orthogonal polynomials $p_j^{(2)}$, $\widehat{p}_j^{(1)}$ satisfy the  relations
\begin{subequations} \label{eq:rectot}
\begin{align}
A \widehat{p}^{(1)} &= z^a B \widehat{p}^{(1)} ,\label{eq:rec3} \\
B^T h^{-1} p^{(2)} &= z^b A^T h^{-1} p^{(2)} . \label{eq:rec2} 
\end{align}
\end{subequations}
\end{cor}

\begin{rmk}
For $a=b=1$ we get from \cref{eq:rec3}
\begin{equation}
\widehat{p}_{i+1}^{(1)} + \alpha_0(i) \widehat{p}_{i}^{(1)} =
z ( \widehat{p}_{i}^{(1)} + \beta_1(i) \widehat{p}_{i-1}^{(1)} ),
\label{eq:rec11a}
\end{equation}
and from \cref{eq:rec2}
\begin{equation}
p_i^{(2)} h_i^{-1} + p_{i+1}^{(2)} h_{i+1}^{-1}  \beta_1(i+1) =
z( p_{i-1}^{(2)} h_{i-1}^{-1} + p_{i}^{(2)} h_i^{-1} \alpha_0(i) ) 
\label{eq:rec11b}
\end{equation}
with $i\geq 0$, and assuming $p_j^{(i)} = \widehat{p}_j^{(i)}=0$ when $j<0$.
Notice that in the general $(a,b)$ case the recursions in \cref{eq:rectot} involve $a+b+2$ terms. 
\end{rmk}

\begin{rmk}
For the Ablowitz--Ladik lattice, $(a,b)=(1,1)$,
the moment matrix can be seen to arise
from the scalar product on functions on the unit circle,
\eq{
\bra f,g \ket_\mu = \frac{1}{2\pi\ri}\int_{S^1} f(z) g(z^{-1}) \re^{\sum_{i>0}
\l(  s_i^{(1)} z^i-s_i^{(2)} z^{-1}\r)}\frac{\rd z}{z}.
\label{eq:biorthcirc}
}
Correspondingly, the associated 2D-Toda $\tau$-function is the partition of
the unitary matrix model,
\eq{
Z_{U(N)}= \prod_{i=0}^{n-1}h_n,
}
and the recursion relations of \cref{eq:rec11a,eq:rec11b} imply the three-term 
recursion relations of \cite{Hisakado:1996di,Kharchev:1996kh} for the unitary ensemble. The general $(a,b)$
case corresponds to 
complex integrals of the form
\eq{
\bra f, g \ket_\mu = \frac{1}{2\pi\ri}\int_{S^1} f(z^{b}) g(z^{-{a}}) \re^{\sum_{i>0}
\l(  s_i^{(1)} z^i-s_i^{(2)} z^{-1}\r)}\frac{\rd z}{z}.
}
Notice that the bilinear form on $\bbC[z]$ thus defined is not symmetric
anymore as soon as $a \neq b$, and the unitary matrix model interpretation is
correspondingly less obvious.
\end{rmk}

\subsection{Hamiltonian structure}
\label{sec:hamstruct}

Since the 2D-Toda hierarchy admits a triplet of compatible Poisson structures
\cite{GC}, a natural question arises as to whether the RR2T flows admit a Hamiltonian formulation.
Unlike the case of the extended bi-graded Toda hierarchy, 
the generic RR2T is not given by an affine submanifold in field space, and
correspondingly the Dirac reduction of the parent Poisson structures is not
straightforward. Remarkably, however, at least one Poisson structure can always be reduced
to the locus defined by the factorization of the Lax operator as in
\cref{eq:Lbigrad}. The key to this is a degeneration property of the
corresponding Poisson tensor, as we now turn to illustrate. \\

It is well-known that the 2D-Toda hierarchy can be formulated in terms of two Lax operators of the form
\eq{
\label{eq:2dtlaxopnm}
\bar{L}_1 = \Lambda^a + \sum_{j\geq -a+1} \bar{u}_j^{(1)}  \Lambda^{-j}, \quad 
\bar{L}_2 = \sum_{j\geq -b} \bar{u}_{j}^{(2)}  \Lambda^{j},
}
for two fixed integers $a, b\geq1$. They are related to the Lax operators
defined in~\eqref{eq:2dtlaxop} by $\bar{L}_1 = L_1^a$ and $\bar{L}_2 =
L_2^b$. In the rest of this subsection we will always use the formulation in terms of the Lax operators~\eqref{eq:2dtlaxopnm}, and, to keep the notations simple, we will drop the bars and denote them by $L_1$ and $L_2$. \\

Denote by $(\dot L_1, \dot L_2)$ an element in the tangent space
$T\mathcal{A}^{\rm 2DT}=\{(\dot L_1=\sum_{j\leq a-1}\dot
u_{j}^{(1)}\Lambda^{j},\dot L_2=\sum_{j\geq -b}\dot u_{j}^{(2)} \Lambda^{j}
)\}$ of the 2D-Toda phase space and introduce the bilinear pairing
\eq{
\bra(\dot L_1,\dot L_2),(X,Y) \ket = \tr(\dot L_1 X + \dot L_2 Y)
}
to induce differential forms in $T^*\mathcal{A}^{\rm 2DT}$ from operators $(X,Y)$ of the
form $X=\sum_{k>n}  x_k \Lambda^k $ and $Y=\sum_{k<m} y_k \Lambda^k$ for some
$n,m\in\mathbb{Z}$. Similarly, we denote by $(\dot A,\dot B)$ an element of
the  tangent space $T\mathcal{A}^{\rm RR}=\{(\dot A=\dot \alpha_{a-1}
\Lambda^{a-1}+\ldots+\dot\alpha_0,\dot B=\dot \beta_1 \Lambda^{-1}+\ldots + \dot
\beta_b \Lambda^{-b}) \}$ to the phase space of rational reductions
$\mathcal{A}^{\rm RR}$. The same bilinear pairing described above produces a
differential form on $\mathcal{A}^{\rm RR}$ starting this time from an
operator $(X,Y)$ of the more general form $X=\sum_{k\in \mathbb{Z}}  x_k
\Lambda^k $ and $Y=\sum_{k\in \mathbb{Z}} y_k \Lambda^k$. \\


It was shown in \cite{GC} that, for $a=b=1$, $\mathcal{A}^{\rm 2DT}$ can be endowed with
three compatible Poisson structures with respect to which the 2D-Toda flows
are Hamiltonian. The construction of~\cite{GC} can be easily extended to the general $a,b\geq1$ case. What was referred to  in \cite{GC} as the ``second'' Poisson tensor, in particular, reads as follows. When applied on a differential form corresponding via the pairing to the operator $(X_1,X_2)$, it gives the following vector 
\ea{
P(\bra\cdot,(X_1,X_2)\ket) =& \frac{1}{2} [L_1,(L_1 X_1 + X_1
  L_1)_{-}- (L_2 X_2 +X_2 L_2)_{-}] \nn \\
+& \frac{1}{2} [L_1,(\Lambda^a+1)(\Lambda^a-1)^{-1}
  \Res([L_1,X_1]+[L_2,X_2])]\nn \\
-& \frac{1}{2} L_1([L_1,X_1]+[L_2,X_2])_{\leq 0}- \frac{1}{2}
([L_1,X_1]+[L_2,X_2])_{\leq 0}L_1,\nn \\
 & \frac{1}{2} [L_2,(L_2X_2 + X_2 L_2)_{+}- (L_1 X_1 +X_1 L_1)_{+}] \nn \\
+& \frac{1}{2} [L_2,(\Lambda^a+1)(\Lambda^a-1)^{-1} \Res
  ([L_1,X_1]+[L_2,X_2])]\nn \\
-& \frac{1}{2}  L_2([L_1,X_1]+[L_2,X_2])_{> 0} - \frac{1}{2} ([L_1,X_1]+[L_2,X_2])_{>0}L_2 .
\label{poisson2dtoda}
}
This Poisson structure degenerates on the submanifold of $\mathcal{A}^{\rm 2DT}$ given by the image of $\mathcal{A}^{\rm RR}$,  as shown by the following Lemma, hence it yields, simply by restriction, a well-defined Poisson structure on such submanifold.
\begin{lem}
For $L_1 = AB^{-1}$, $L_2 = B A^{-1}$, we have that
\begin{equation}
P(\bra\cdot,(X_1,X_2)\ket) = ( (\dot{A} - A B^{-1} \dot{B}) B^{-1}, (\dot{B} - B A^{-1} \dot{A} ) A^{-1} ) \nn
\end{equation}
where $(\dot{A}, \dot{B}) \in T \mathcal{A}^{\rm RR}$ is given by 
\begin{align}
\dot{A} = &\left( ( X_2 B A^{-1} - A B^{-1} X_1 )_- + ((\Lambda^{-a} -1)^{-1} \zeta) \right) A - \nn\\
&-A \left( (A^{-1} X_2 B - B^{-1} X_1 A)_- +(1-\Lambda^a)^{-1} \zeta \right)  , \nn \\
\dot{B} = &\left( (B A^{-1} X_2 - X_1 A B^{-1} )_- + ((1-\Lambda^a)^{-1} \zeta ) \right) B   \nn \\
& - B \left( (A^{-1}X_2 B - B^{-1} X_1 A )_- + ((1-\Lambda^a)^{-1} \zeta)\right) , \nn
\end{align}
and
\begin{equation}
\zeta = \Res ( [L_1,X_1] +[L_2, X_2]) . \nn
\end{equation}
\end{lem}
In other words the vector given by the image by the Poisson tensor of the differential form $\bra\cdot,(X_1,X_2)\ket$ is equal to the push-forward of a  vector in $T\mathcal{A}^{\rm RR}$, i.e., it is tangent to $\mathcal{A}^{\rm RR}$.

For any functional $f=f(L_1,L_2)$ on $\mathcal{A}^{\rm 2DT}$ we denote by
$(\frac{\delta f}{\delta L_1 },\frac{\delta f}{\delta L_2} )$ a pair of
operators such that we can express the derivative of $f$ along $(\dot
L_1,\dot L_2)$ as
\eq{
\dot f = \bra \l(\frac{\delta f}{\delta L_1 },\frac{\delta f}{\delta L_2}\r),
\l(\dot L_1,  \dot L_2 \r) \ket .
}
In other words the vector $(\frac{\delta f}{\delta L_1 },\frac{\delta f}{\delta L_2} )$ is a preimage of the differential $\rd f$ with respect to the bilinear pairing above. The Poisson bracket of two functionals $f$, $g$ on $\mathcal{A}^{2DT}$ is 
\begin{equation} \label{eq:gtoda}
\{ f, g \} = \bra \l(\frac{\delta f}{\delta L_1 },\frac{\delta f}{\delta L_2}\r),
P\l(\bra\cdot, \l(\frac{\delta g}{\delta L_1 },\frac{\delta g}{\delta L_2}\r)   \ket\r)
 \ket .
\end{equation}
From the Lemma and skew-symmetry, it follows that $\{f,g\}$, when restricted on $\mathcal{A}^{RR}$, does not depend on the choice of functional $f$ (resp. $g$) on $\mathcal{A}^{2DT}$ as long as it restricts to the same $f_{|\mathcal{A^{RR}}}$ (resp. $g_{|\mathcal{A^{RR}}}$). In other words $\mathcal{A}^{RR}$ is a Poisson submanifold of $\mathcal{A}^{2DT}$.

The explicit form of RR2T Poisson brackets for the coefficients $\alpha_0,\ldots,\a_{a-1},\b_1,\ldots,\b_b$ of $A$ and $B$ can be computed starting from the 2D-Toda (second) Poisson bracket for the first $a$ and $b$ coefficients $u^{(1)}_0,\ldots u^{(1)}_{a-1},u^{(2)}_{-1},$ $\ldots,u^{(2)}_{b}$ of $L_1$ and $L_2$ respectively\footnote{See \cite{GC} for explicit formulas.} and applying the change of coordinates induced by the equations $L_1^a=AB^{-1}$, $L_2^b=BA^{-1}$. 

In case $(a,b)=(1,1)$, where $A=\Lambda+\alpha$ and $B=1+\beta \Lambda^{-1}$, one readily computes
\begin{equation}
\begin{split}
\{\a(n),\a(m)\} &= 0\\
\{ \log \a(n) , \log \b(m) \} &= \delta (n-m+1) - \delta(n-m)\\
\{ \log \b(n) , \log \b(m) \} &= \delta (n-m+1) - \delta(n-m-1)
\end{split}
\label{eq:pbavm}
\end{equation}
which coincides with the Poisson structure introduced by Adler--van~Moerbeke
\cite{MR1794352} for the Ablowitz--Ladik hierarchy. 

%

Since the
2D-Toda flows are Hamiltonian w.r.t \cref{poisson2dtoda}, with Hamiltonian
functions given by 
\begin{equation}
H_i^{(j)} = -\frac{1}{i}\tr L_j^i, \qquad j=1,2 ,
\label{eq:todaham}
\end{equation}
the Ablowitz--Ladik flows are Hamiltonian with respect to \cref{eq:pbavm},
with the same Hamiltonian functions.

\subsection{Long-wave limit and semi-classical Lax formalism}
\label{sec:dless}

Starting from the 2D-Toda lattice hierarchy of \cref{sec:2dtoda},
a continuous integrable system of $2+1$ evolutionary PDEs can be constructed by 
interpolation. For a fixed real parameter
$\epsilon>0$ -- the ``lattice spacing'' --  introduce dependent variables
$U_j^{(i)}(x)$ such that $U_j^{(i)}(\epsilon n) = (u_j^{(i)})_n$, and
accordingly define a shift operator $\Lambda_\epsilon=\re^{\epsilon \de_x}$ by one unit of
lattice spacing. 
Replacing the unit shift $\Lambda_1$ by the
$\epsilon$-shift $\re^{\epsilon \de_x}$ and rescaling the
time variables by $t_r^{(i)}\triangleq \epsilon s_r^{(i)}$ gives a system of
evolutionary partial differential equations in the time variables $t_r^{(i)}$
in the form
\ea{
\de_{t^{(p)}_r} U_j^{(i)}(x) = & \sum_{g\geq 0} \epsilon^{2g}\cP^{[g],p,r}_{i,j}(U,U_x, \dots, U^{(2g)})\nn \\
& = \sum_{k,l}\cP^{[0],p,r}_{k,l,i,j}(U) \de_x U_k^{(l)} +  
\cO\l(\epsilon^2\r)
\label{eq:dispexp}
}
where $\cP^{[g],p,r}_{i,j}(U,U_x,
\dots, U^{(2g)})$ is an element of the vector space $\cI_g$ of differential polynomials in $U(x)$ homogeneous
of degree $2g+1$ with respect to the independent variable $x$. Following \cite{MR2108440}, we will call this the {\it interpolated 2D-Toda
  lattice}. \\

We will be particularly interested in the quasi-linear limit of the
interpolated 2D-Toda lattice, where the
dispersion parameter $\epsilon$ is set to zero. As noticed in
\cite{Takasaki:1991zs}, 
the dispersionless limit $\epsilon\to 0$ of
\cref{eq:dispexp} can be formulated  as the quasi-classical (Ehrenfest) limit of
the Lax equations \cref{eq:2dtlax}, as follows. Write $\lambda_i(z)\triangleq
\sigma_\Lambda(L_i) \in \bbC((z))$ for the total symbol in the variable $z\in
\bbC$ of the difference operators $L_i$ in \cref{eq:2dtlaxop},
\eq{
\lambda_1(z) = z+ \sum_{j\geq0} U_j^{(1)}  z^{-j}, \quad \lambda_2(z) = \sum_{j\geq -1} U_{j}^{(2)}  z^{j}.
}
Furthermore, define the {\it Orlov functions}
\eq{
\cB^{(1)}_n(z) \triangleq [(\lambda_1)^n]_+, \quad  \cB^{(2)}_n(z) \triangleq [(\lambda_2)^{n}]_-,
}
where $[f]_\pm$ denotes the projection to the analytic/purely principal part
of $f\in\bbC((z))$, and for $f,g \in \bbC((x,z))$ define the Poisson bracket
\eq{
\{f,g\}_{\rm Lax} = z\l(\frac{\de f}{\de x}\frac{\de g}{\de z}-\frac{\de g}{\de
  x}\frac{\de f}{\de z}\r). 
}
Then the {\it semiclassical Lax equations}
\eq{
\label{eq:laxsato}
\frac{\de \lambda_i}{\de t_r^{(j)}} \triangleq \{\cB^{(j)}_r, \lambda_i\}_{\rm
Lax},
}
where the time-derivatives are understood to be taken at fixed $z$, induce the
dispersionless limit of the interpolated 2D-Toda flows of \cref{eq:dispexp} on the coefficients $U_k^{(l)}$ of $\lambda_l$,
\eq{
\de_{t^{(p)}_r} U_j^{(i)}(x) =  \sum_{k,l} \cP^{[0],p,r}_{k,l,i,j}(U)
\de_x U_k^{(l)}.
\label{eq:dToda}
}
Consistency of the dispersionless Lax equations \cref{eq:laxsato} requires the
existence of a potential function $\cF$ of the long-wave time variables
$t_r^{(j)}$ such that
\eq{
  \cB_n^{(i)}(z(\lambda_j))  =\delta_{ij}\lambda_j^{s_j n} 
+\delta_{j2}\frac{\de^2 \cF}{\de t^{(1)}_0 \de t^{(i)}_n}
-\sum_{m>0}\frac{\de^2 \cF}{\de t^{(i)}_n
      \de t^{(j)}_m}\frac{1}{m \lambda_j^{s_i m}}
}
where $s_i=(-1)^{i+1}$. By the general dToda theory \cite{Takasaki:1991zs}, the potential $\cF$ yields the eikonal limit of the
logarithm of the long-wave limit of the 2D-Toda $\tau$-function,
\eq{
\cF = \log \tau_{\rm dToda}.
}

\subsection{Rational reductions and Frobenius manifolds}


The integration of the consistency conditions for $\cF$ has a natural
formulation in the language of Frobenius manifolds \cite{MR2753798}. An even
more poignant picture emerges in the case of RR2T: by
\cite{Dubrovin:1994hc,2012arXiv1210.2312R} the dispersionless
limit (henceforth denoted as {\it dRR2T}) coincides with the Principal Hierarchy of the Frobenius manifold defined
on a genus zero double Hurwitz space, as we now turn to show.

\subsubsection{Flat structures and the Principal Hierarchy}

We start by giving the following

\begin{defin} Let $M$  be a complex manifold, $\mathrm{dim}_\bbC M=n$. A
  holomorphic \emph{Frobenius structure} $\cM=(M,\eta, \cdot)$ on $M$ is the datum of a
  holomorphic symmetric $(0, 2)$-tensor $\eta$, which is non-degenerate and
  with flat Levi-Civita connection $\nabla$, and a commutative, associative fiberwise product
  law with unit $X \cdot Y$ on vector fields $X,Y \in
  \cX(M)$, which is tensorial and satisfies \\

\begin{description}
\item[Compatibility]
\eq{
\eta(X\cdot Y, Z)=\eta(X, Y\cdot Z)\qquad \text{for all vector fields}\;\; X,
Y, Z;
}
\item[Flatness]
the pencil of affine connections
\eq{
\label{eq:defEC}
\nabla^{(\zeta)}_X Y\triangleq \nabla_X Y+\zeta X\cdot Y\qquad \zeta\in\mathbb{C}
}
is identically flat $\forall$~$\zeta\in \bbC$.
\end{description}

\end{defin}

Following terminology introduced in \cite{2012arXiv1210.2312R}, extra
flat structures on $\cM$ will be characterized according to the following

\begin{defin}
Let $\cM=(M, \eta, \cdot)$ be a holomorphic Frobenius manifold structure on
$M$, and let $e\in\cX(M)$ be the unit of the $\cdot$-product. We will say that $\cM$ is
\ben
\item \emph{semi-simple} if the product structure
  $\cdot|_{p}$ on $T_pM$ has no nilpotent elements for a generic $p\in M$;
\item of \emph{dual-type} if $\exists~d \in \bbZ$ such that $\forall~f\in \cO_M$,
\eq{
\label{eq:dualtype}
\nabla \rd f=0 \;\;\Rightarrow\;\;\l(\partial_e+\frac{d-1}{2}\r) f= c_f
}
for some constant $c_f\in \bbC$.
\item
\emph{conformal} if $\nabla e = 0$ 
and $\exists~E\in\cX(M)$ such that
$\nabla E \in \Gamma(\mathrm{End}(TM))$ is diagonalizable and horizontal
w.r.t. $\nabla$ and the pencil of affine connections \cref{eq:defEC} extends
to a flat meromorphic connection $\nabla^{(\zeta)}$ on $M\times \bbP_\zeta^1$ via
\ea{
\label{eq:defECint}
\nabla^{(\zeta)}\frac{\partial}{\partial \zeta}=& 0 \\
\nabla^{(\zeta)}_{\partial/\partial \zeta} X= & \frac{\partial}{\partial \zeta}X+E\cdot
X-\frac{1}{\zeta}\widehat{\mu}X
\label{eq:defECz}
}
where $\widehat{\mu}$ is the traceless part of $-\nabla E$;
\item
\emph{tri-hamiltonian} if it is conformal, $n$ is even and $\widehat{\mu}$ has
only two eigenvalues $\pm d/2$ with multiplicity $n/2$, where
$d=2(1-\tr(\nabla E))$.
\een
\end{defin}

A Frobenius manifold structure $\cM$ on $M$ embodies the existence of a
Hamiltonian hierarchy of quasi-linear commuting flows on its loop space \cite{Dubrovin:1994hc}.  
Let $\mathsf{t}=\{\tau_{(\zeta)}^\alpha \in \cO_M\}_{\alpha=1}^n$ be the datum of a marked
system of flat coordinates for $\nabla^{(\zeta)}$ depending holomorphically on $\zeta$
around $\zeta=0$: this is determined up to a $\bbC[[\zeta]]$-valued
affine transformation in general, a freedom which reduces to a complex affine
transformation when $\cM$ is conformal by virtue of \cref{eq:defECz}. Write 
$h_{\alpha,p}\triangleq \eta_{\a\b} ([\zeta^p] \tau_{(\zeta)}^\alpha) \in \cO_M$ for the $p^{\rm
  th}$-Taylor coefficient of $\eta_{\a,\b}\tau_{(\zeta)}^\alpha$ at $\zeta=0$. In terms
of the flat metric $\eta$, we define \cite{Dubrovin:1994hc} a hydrodynamic Poisson structure $\{,\}_\eta$ on
the loop space $\LL_\cM=\mathrm{Maps}(S^1,M)$ as
\begin{equation}\label{eq:dlessbracket}
\l\{\tau_{(0)}^\alpha(X), \tau_{(0)}^\beta(Y)\r\}_\eta = \eta^{\a\b}\delta'(X-Y),
\end{equation}
where $X,Y\in S^1$ are coordinates on the base of the loop space, as well as an infinite set of quasi-linear Hamiltonian flows via
\eq{
\frac{\de \tau^\b}{\de t^{\a,p} }\triangleq \l\{\tau^\b, H_{\a,p}\r\}_\eta=
\de_X \de^\b h_{\a,p}.
\label{eq:ph}
}
These flows generate a commuting family of Hamiltonian conservation laws
\cite{Dubrovin:1994hc}, which is complete as long as $\cM$ is semi-simple
\cite{MR796577}.

\begin{defin}
The hierarchy of hydrodynamic type \cref{eq:ph} will be called the Principal
Hierarchy associated to  $(\cM, \mathsf{t})$.
\label{def:ph}
\end{defin}

\subsubsection{Frobenius dual-type structures for the RR2T}

Let $a,b \in \bbZ^2_+$ and $m \in \bbZ$. In this section we will construct a Frobenius dual-type
structure \cite{2012arXiv1210.2312R} on the space of symbols of the Lax
operator $L_1^{a+m}=L_2^{-b-m}$ of the generalized RR2T of \cref{rmk:equivlax}.

\begin{defin}
\label{defn:hurwitz}
Let $v, q_{-a+1},\ldots, q_{b-1} \in \bbC$, $a,b \in \bbZ^+$  and
$\nu\in\mathbb{C}^*$. We define $\mathcal{H}_{a,b,\nu}$ to be the space of multivalued functions on $\mathbb{P}^1$ of the form
\eq{
\lambda(z)= \re^vz^{\nu+b}\frac{\prod_{k=0}^{a-1}(z-\re^{q_{-k}})}{\prod_{l=0}^{b-1}(z-\re^{-q_l})}.
\label{eq:lambda}
}
\end{defin}

\begin{rmk}
Writing
\eq{
z_k =
\l\{
\bary{ccc}
0 & \mathrm{for} & k=1 \\
\re^{q_{2-k}} & \mathrm{for} & k=2, \dots, a+1 \\
\re^{-q_{k+2-a}} & \mathrm{for} & k=a+2, \dots, a+b+1 \\
\infty & \mathrm{for} & k=a+b+2 \\
\eary
\r.
}
the meromorphic function
$z^{-\nu}\lambda(z)$ has, for generic values of the parameters, a zero of order $b$ at $z_1$, simple zeroes at
$z_{k+2}$, $k=0,\ldots,a-1$, a pole of order $a$ at
$z_{a+b+2}\triangleq \infty$, and simple poles at $z_{a+2+k}\triangleq e^{-q_k},
k=0,\ldots,b-1$. When $\nu=m\in \bbZ_0$, this function is the total symbol of
the $(a+m)^{\rm th}$ power of the Lax operator $L_1$ (or equivalently, the
$(b+m)^{\rm th}$ inverse power of $L_2$) of the $m$-generalized RR2T
 of bidegree $(a,b)$, up to a trivial rescaling of the argument $z$. The space $\HH_{a,b,m}$ is then a genus zero double
Hurwitz space: a moduli space of rational curves with a marked
meromorphic function $\lambda:\bbP^1\to \bbP^1$ having specified ramification
profile $\kappa\in \bbZ^{a+b+2}$ at zero and
infinity. In our case, the latter reads
\eq{
\kappa=(b+m, \underbrace{1, \dots, 1}_a, \underbrace{-1, \dots, -1}_b, -a-m).
}
\\

\end{rmk}

For $\nu=m\in\bbZ$, we define on the $(a+b)$-dimensional complex manifold $\mathcal{H}_{a,b,\nu}$ a
triplet $(\eta^{(1)}, \eta^{(2)},\eta^{(3)})$, where $\eta^{(i)}\in \Gamma(\mathrm{Sym}^2 T^*\HH_{a,b,m})$, $\det
\eta^{(i)}\neq 0$,  by the Landau--Ginzburg formulas
\ea{
\label{eq:predualmetric}
\eta^{(1)} (X, Y)= &
\sum_{i=1}^{a+b+2}\Res_{z_i}\frac{X(\lambda)Y(\lambda)}{\rd\lambda}\l(\frac{\rd z}{z}\r)^2, \\
\label{eq:dualmetric}
\eta^{(2)}(X, Y)= &
\sum_{i=1}^{a+b+2}\Res_{z_i}\frac{X(\log\lambda)Y(\log\lambda)}{\rd\log\lambda}\l(\frac{\rd
  z}{z}\r)^2,
\\
\label{eq:trimetric}
\eta^{(3)}(X, Y)= &
\sum_{i=1}^{a+b+2}\Res_{z_i}\frac{X(\lambda^{-1})Y(\lambda^{-1})}{\rd\lambda^{-1}}\l(\frac{\rd
  z}{z}\r)^2, \\
} 
for $X,Y \in \cX(\HH_{a,b,\nu})$. We further equip $T_\lambda\mathcal{H}_{a,b,m}$
with a triplet $(\bullet, \star, \ast)$ of commutative, associative products defined by
\ea{
\label{eq:predualprod}
\eta^{(1)}(X\bullet Y,
Z)= &
\sum_{i=1}^{a+b+2}\Res_{z_i}\frac{X(\lambda)Y(\lambda)Z(\lambda)}{\rd\lambda}\l(\frac{\rd
  z}{z}\r)^2, \\
\label{eq:dualprod}
\eta^{(2)}(X\star Y,
Z)= &
\sum_{i=1}^{a+b+2}\Res_{z_i}\frac{X(\log\lambda)Y(\log\lambda)Z(\log\lambda)}{\rd\log\lambda}\l(\frac{\rd
  z}{z}\r)^2,\\
\label{eq:triprod}
\eta^{(3)}(X\ast Y,
Z)= &
\sum_{i=1}^{a+b+2}\Res_{z_i}\frac{X(\lambda^{-1})Y(\lambda^{-1})Z(\lambda^{-1})}{\rd\lambda^{-1}}\l(\frac{\rd
  z}{z}\r)^2,
}
depending holomorphically on the base-point $\lambda\in\HH_{a,b,m}$. When
$\nu\notin \bbZ$, \cref{eq:predualmetric,eq:predualprod,eq:trimetric,eq:triprod}
are ill-defined, but the definition \cref{eq:dualmetric,eq:dualprod} of the
metric and product $(\eta^{(2)},\star)$ carries through unscathed. The main result of this section is the
following
\begin{thm}
\label{thm:frob}
Let $a,b\in \bbZ^+$, $\nu \in \bbC$. Then the following statements hold:
\ben[i)]
\item \cref{eq:dualmetric,eq:dualprod} define on $\HH_{a,b,\nu}$ a semi-simple Frobenius
 structure of dual-type $\cM^{(2)}_{a,b,\nu}=(\HH_{a,b,\nu}, \eta^{(2)}, \star)$ of charge one.
\item Let $\nu=m \in \bbZ$ and suppose that both $b+m, -a-m$ are either
  equal to one or negative. Then \cref{eq:predualmetric,eq:predualprod} define
  a conformal Frobenius structure $\cM^{(1)}_{a,b,m}=(\HH_{a,b,m}, \eta^{(1)}, \bullet)$ of charge one on $\HH_{a,b,m}$. The unit 
of this structure is flat iff $m\neq 1-b$ and $m\neq -a-1$, and we have that
\eq{
\cM^{(2)}_{a,b,m} = \DD(\cM^{(1)}_{a,b,m})
\label{eq:duality}
}
where $\DD$ is Dubrovin's duality morphism of Frobenius structures \cite{MR2070050}. 
\item Let $b=a$ and $\nu=1-a$. Then
  \cref{eq:predualmetric,eq:predualprod,eq:dualmetric,eq:dualprod,eq:trimetric,eq:triprod}
  define a tri-hamiltonian Frobenius structure on $\HH_{a,a,1-a}$.
\een
\end{thm}

\begin{proof}
\cref{thm:frob} is essentially a verbatim translation of Theorem~2 in
\cite{2012arXiv1210.2312R} to the setting of RR2T. We sketch the main points
of the proof below.
For Point (i), flatness of the residue pairing $\eta^{(2)}$ follows from checking, through a direct
computation of (\ref{eq:dualmetric}), that the coordinates $v, q_{-a+1}, \dots, q_{b-1}$ form in fact a flat coordinate frame for
$\eta^{(2)}$. Further, by the explicit form of \cref{eq:dualprod,eq:dualmetric}, the
 $\star$-product structure is clearly compatible with the metric $\eta^{(2)}$ in the sense that the two
form a Frobenius algebra on the tangent spaces of $\mathcal{H}_{a,b,\nu}$; it is
immediate to check that the algebra is unital, the identity consisting in the flat vector field
$e=\partial_{v}$. Moreover the $a+b$ critical values of $\log\lambda$,
\eq{
u_i\triangleq \log\lambda(y_i), \quad y_i \in \mathbb{P}^1 \hbox{ s.t. } \lambda'(y_i)=0, i=1,\ldots,a+b
}
are a set of local coordinates on $\mathcal{H}_{a,b,\nu}\setminus \Delta_{a,b,\nu}$,
where the discriminant 
$\Delta_{a,b,\nu}\triangleq \{\lambda \in \mathcal{H}_{a,b,\nu} | u_i \neq u_j \forall
i\neq j\}$. In these coordinates, the product and the metric take the form
\ea{
\partial_{u_i}\star \partial_{u_j} & = \delta_{ij} \partial_{u_i}, \nn \\
\eta^{(2)}(\partial_{u_i},\partial_{u_j}) & =\eta^{(2)}_{ii}(u) \delta_{ij}
}
for functions $\eta^{(2)}_{ii}(u) \in \cO(\HH_{a,b,\nu}\setminus\Delta_{a,b,\nu})$, possibly
singular on $\Delta_{a,b,\nu}$. Moreover, thanks to the flatness of $\eta^{(2)}$ and its compatibility with the product, we can write
$$\eta^{(2)}_{ii}(u) = \eta^{(2)}(\partial_{u_i},\partial_{u_i}) =  \eta^{(2)}(e,\partial_{u_i})$$
and, by the flatness of $e$ we get $\eta^{(2)}_{ii}=\partial_{u_i}t_1(u)$, where $\rd t_1(u)=\eta^{(2)}(e,\cdot)$. This means that $\eta^{(2)}$ is an Egoroff metric which implies (see for instance \cite{MR1461570}) that $\nabla_X \eta^{(2)}(Y\star Z,K)$ is symmetric in all four vector fields $X,Y,Z,K$.\\
The above proves that
\cref{eq:dualmetric,eq:dualprod} endow $\mathcal{H}_{a,b,\nu}$ with a semi-simple Frobenius
dual-type structure, which has charge one by the flatness of the unit vector
field.\\
As for Point (ii), notice that when $\nu =m\in \bbZ$, $\lambda$ is single-valued and $\HH_{a,b,\nu}$ is a genus zero double
Hurwitz space. Under the further condition that the zeroes of $\lambda$ be
simple, $\HH_{a,b,m}$ becomes a Hurwitz space in a standard sense, with the only
proviso that the divisor where $\lambda$ has multiple zeroes is removed. Then
under the conditions of Point (ii) the existence of a conformal Frobenius
manifold structure is a direct corollary of
\cite[Theorem~5.1]{Dubrovin:1994hc} for $m\neq1-a,1-b$; when $m=1-a$ or
$1-b$, the proof of the above theorem goes through almost unscathed except for
the covariant constancy of the unit vector field, which fails to be satisfied
in these cases. Furthermore,
\cref{eq:duality} follows from a standard argument (see
\cite[Proposition~5.1]{MR2462355}), which together with Point (i) above proves semi-simplicity and the charge one
condition. Finally, Point (iii) is an immediate consequence of Point (ii)
together with \cite[Theorem~2]{2012arXiv1210.2312R}.

\end{proof}

Under the conditions of Point (ii), the statement of \cref{thm:frob} implies
that the metrics $\eta^{(1)}$ and $\eta^{(2)}$ form a
flat pencil, which is exact if and only if $m\neq 1-a, 1-b$: $\eta^{(2)}$ is
the (inverse) of the intersection form on $\cM^{(1)}_{a,b,m}$. 
Moreover, when $\lambda$ has only simple zeroes and poles this is enhanced to a a
triple of compatible flat metrics $\eta^{(1)}, \eta^{(2)}, \eta^{(3)}$. And
finally, if the unit of the first structure is flat, the resulting Frobenius
structure is tri-hamiltonian.\\
By comparing the formulas for the flat coordinates for $\eta^{(2)}$ and $\eta^{(1)}$ one
easily sees when the pencil $(\eta^{(2)})^{-1}-\epsilon(\eta^{(1)})^{-1}$ is resonant, namely, when
$\eta^{(1)}$ and $\eta^{(2)}$ have common flat coordinates. This happens if and only if
$\lambda$ has more than one pole; there is one common flat coordinate for each
pole after the first. \\

As an immediate consequence of \cref{thm:frob}, the semi-classical limit of
the RR2T, \cref{eq:dToda,eq:laxsato}, has a neat description in terms of the
Principal Hierarchy of $\cM^{(i)}_{a,b,\nu}$, $i=1,2$. 
\begin{cor}
\label{cor:toda}
The following statements hold true:
\ben
\item for any $(a,b)\in \bbZ_+^2$, $m \in \bbZ$ and $\mathsf{t}\in \mathrm{Aff}_{a+b}(\bbC[[z]])$,  the Principal Hierarchy
of $(\cM^{(2)}_{a,b,m}, \mathsf{t})$ is a complete system of commuting Hamiltonian
conservation laws of the $m$-generalized dRR2T of bidegree $(a,b)$;
\item Let $-a-m<0$, $b+m<0$ as in Point (ii) of \cref{thm:frob}, and fix $\mathsf{t}\in \mathrm{Aff}_{a+b}(\bbC)$
  such that 

\ea{
\label{eq:htoda1}
 h_{\a,p}=&-\Res_{z=\infty}\frac{\lambda^{\frac{\a}{m+a}+p}}{\l(\frac{\a}{m+a}\r)_{1+p}}\frac{\rd
    z}{z}, \quad \a=1,\dots, m+a,\\
\label{eq:htoda2}
h_{\a+m+a,p}=&-\Res_{z=0}\frac{\lambda^{\frac{\a}{-b-m}+p}}{\l(\frac{\a}{-b-m}\r)_{1+p}}\frac{\rd
    z}{z}, \quad \a=1,\dots,-m-b-1,
}

where $(x)_n\triangleq \frac{\Gamma(x+n)}{\Gamma(x)}$.  Then the Hamiltonian flows
  of the Principal Hierarchy (\cref{eq:ph}) of $\cM^{(1)}_{a,b,m}$ associated
   to $h_{\a,p}$, $\a=1,\dots,a-b-1$,  coincide with
  the semiclassical Lax flows, \cref{eq:laxsato}, for the $m$-generalized dRR2T of
 bidegree $(a,b)$ upon identifying 
\ea{
\label{eq:resc1}
t_{\a,p}\to & \frac{\l(\frac{\a}{m+a}\r)_{1+p}}{\a+p(m+a)}t^{(1)}_{\a+p(m+a)}, \quad
\a_1,\dots,m+a, \\
t_{\a+m+a,p} \to & \frac{\l(\frac{\a}{-b-m}\r)_{1+p}}{\a-p(m+b)}t^{(2)}_{\a-p(b+m)}, \quad \a=1,\dots,-m-b-1.
\label{eq:resc2}
}
\een
\end{cor}

\begin{proof}
Point (2) of the Corollary is an immediate application of Proposition~6.3 and
Theorem~6.5 in \cite{Dubrovin:1994hc}. Notice in particular that 
Proposition~6.3 warrants the existence of a flat coordinate system
$\mathsf{t}$ for the deformed connection on $\HH_{a,b,m}\times \bbC$
(\cref{eq:defECint,eq:defECz}) which is compatible with
\cref{eq:htoda1,eq:htoda2}; the scaling factors in \cref{eq:resc1,eq:resc2}
are required for consistency with the definition of the semiclassical Lax flows\footnote{See
e.g. \cite[Section~1]{MR2452424}.}.
To see why Point~(1) holds, consider the
Taylor expansion in the variable $\zeta$ of the
deformed flatness equations $\cM^{(1)}_{a,b,m}$ in the tangent directions to $\HH_{a,b,m}$, \cref{eq:defEC}.
Then, from \cref{eq:dualmetric,eq:dualprod,eq:predualmetric,eq:predualprod}, writing the
$p^{\rm th}$ Taylor coefficient in flat coordinates for the intersection form
$\eta^{(2)}$ yields the deformed flatness equations of
the dual Frobenius structure $\cM^{(2)}_{a,b,m}$ with $\zeta=p$ \cite{MR2070050,Brini:2011ff}. The first statement then
follows immediately.
\end{proof}

\begin{rmk}
The dispersionless limit of the Poisson structure for RR2T obtained as a reduction of the second Poisson bracket of the 2D-Toda hierarchy \cite{GC} corresponds to the Poisson structure associated to the metric $\eta^{(2)}$ via \cref{eq:dlessbracket}, as one can promptly check by computing the Poisson brackets for the coefficients of $A$ and $B$, as described in \cref{sec:hamstruct}, and taking their quasi-classical limit.
\end{rmk}

\begin{rmk}
Under the conditions of Point (2) of \cref{cor:toda},
it should be stressed that the dToda Hamiltonian flows, \cref{eq:laxsato}, are generated
by a strict subset of the flat coordinates of the deformed connection,
\cref{eq:defECint,eq:defECz}. The remaining flows, which by semi-simplicity of
$\cM^{(1)}_{a,b,m}$ make the Principal Hierarchy a complete family of
conservation laws \cite{MR796577}, are a genuine extension of the dRR2T,
analogous to the extension of the  ordinary 1D-Toda hierarchy
\cite{MR2108440}. On the other hand, as soon as the conditions of Point (2) are not matched,
it can readily be checked in examples that the metric in \cref{eq:predualmetric} is
typically curved if either of $b+m$ or $-a-m$ is greater than one.
The conditions on the range of $m$ leaves only two possibilities for $m\geq
0$: $b=0,\; m=1$ or $b=1,\; m=0$. The case $m<0$ displays instead a wealth of
flat structures: as long as $-a-1\leq m\leq 1-b$ and $m\neq -a,-b$ the metric
$\eta^{(1)}$ in \cref{eq:predualmetric} is flat. Equivalently, for any fixed bidegree
$(a,b)$ there exist $a+b+1$ generalized RR2T (see \cref{rmk:equivlax}) such
that their semi-classical limit has a dispersionless bi-hamiltonian structure of
Dubrovin--Novikov type. This structure is exact when both $b+m$ and $-a-m$ are
negative, and tri-hamiltonian if $a=b, \;m=-a+1$.\\
\end{rmk}

\begin{rmk}[Flat coordinates of $\eta^{(1)}$]
\label{rmk:flat}
Flat coordinates for the first Frobenius structure on $\HH_{a,b,m}$ can be constructed using standard methods from
\cite{Dubrovin:1994hc,2012arXiv1210.2312R}. For definiteness, consider the
case when $b=1$ and $m=0$. By applying the change of variables $z \mapsto \re^{-q_0}(z+1)$ we obtain
\eq{
\lambda = \re^{v-aq_0} (z+1)^a \left(1-\frac{e^{2q_0} - 1}{z}\right)
\prod_{k=1}^{a-1} \left( 1-\frac{e^{q_{-k}+q_0}}{z+1}\right), \hspace{1cm}
\phi= \frac{\rd z}{z+1}.
}
We denote by $z=z(\lambda,q)$ a local inverse of the function $\lambda(z,q)$ and, from the equation $\partial_q (\lambda(z(\lambda,q),q)) = 0 $, we obtain the``thermodynamic identity'' $\partial_q \lambda = - (\partial_z \lambda)(\partial_q z)$, from which we can rewrite the residue formula \cref{eq:predualmetric} as
\eq{
\eta(X, Y)=\sum_{i=1}^{a+3}\Res_{z_i} X(\log(z+1)) Y(\log(z+1)) \rd\lambda
}
Now notice that we can expand the local solutions $\log(z(\lambda,q)+1)$ in the following way as series of $\lambda$:
\eq{
\begin{array}{l l}
\log(z+1)=\frac{1}{a}\left[ \log\lambda - (v-aq_0) - \sum_{k=1}^a\tau_k \frac{1}{\lambda^{k/a}} \right] + \mathcal{O}\left(\frac{1}{\lambda^{1+1/a}}\right), &  z\to\infty \\
\log(z+1)=\mathcal{O}\left(\frac{1}{\lambda}\right),& z\to0 \\
\log(z+1)=\log \lambda + c_0 + \mathcal{O}(\lambda), & z\to -1\\
\log(z+1)= c_j + \mathcal{O}(\lambda), & z\to e^{q_{-j}+q_0}-1
\end{array}
}
This shows that the only contribution to the sum in \cref{eq:predualmetric} comes from $z=\infty$ and that the coefficients
\begin{align*}
&\tau_0 = v-aq_0\\
&\tau_k = \frac{a}{k}\ \Res_{\lambda^{1/a}=\infty}\left[ \lambda^{k/a} \frac{\partial}{\partial \lambda^{1/a}} \log(z+1) \rd \lambda^{1/a}\right] = \frac{a}{k}\ \Res_{z=\infty} \frac{\lambda^{k/a}}{z+1}\rd z ,\qquad k=1,\ldots,a
\end{align*}
are flat coordinates for $\eta^{(1)}$.
\end{rmk}

\begin{examplen}[Bi-hamiltonian structure of $q$-deformed dispersionless $2$-KdV]

Let us consider the dispersionless limit of the $q$-deformed Gelfand--Dickey
hierarchy of \cref{exa:qdef} for $n=2$. The symbol of the Lax operator reads
\eq{
\lambda(z)= z^3+a z^2 + b z + c
}
and a quick inspection of the semi-classical Lax equations
reveals that $c$ is invariant under the flows of \cref{eq:laxsato}. When $c=0$,
the hierarchy then manifestly reduces to the generalized dRR2T of bidegree $(a,b)=(2,0)$ with
$\nu=m=1$, $v=0$. \\

By \cref{thm:frob}, the space of coefficients $\HH_{2,0,1}$ is endowed
with a conformal Frobenius manifold structure
$\cM^{(1)}_{2,0,1}=(\HH_{2,0,1},\eta^{(1)},\bullet)$ of charge one. The discussion
of \cref{rmk:flat} shows that flat coordinates for the metric $\eta^{(1)}$ are
given by
\eq{
t_1= -\frac{a}{3}, \quad t_2 = b-\frac{a^2}{6}.
}
In this chart, $\eta^{(1)}$ takes the off-diagonal form
$\eta_{ij}^{(1)}=\delta_{i+j,2}$, and the algebra structure on
$\cM^{(1)}_{2,0,1}$ is induced by the polynomial prepotential 
\eq{
F^{(1)}(t_1,t_2)= \frac{12}{5} t_1^6-t_2 t_1^4+\frac{1}{4} t_2^2 t_1^2-\frac{t_2^3}{144}.
}
As far as the dual-type Frobenius structure
$\cM^{(2)}_{2,0,1}=(\HH_{2,0,1},\eta^{(2)},\star)$ is concerned, from the proof of Point~(i) of
\cref{thm:frob} we know that the zeroes $(\re^{q_0}$, $\re^{q_{-1}})$ of $\lambda$ are
exponentiated flat coordinates of $\eta^{(2)}$. Then the Miura transformation 
\ea{
t_1= & \frac{1}{3} \left(\re^{q_0}+\re^{q_{-1}}\right),\\
t_2= & \frac{1}{6} \left(4 \re^{q_0+q_{-1}}-\re^{2 q_0}-\re^{2 q_{-1}}\right).
}
and \cref{eq:dualmetric} yield $\eta_{ij}^{(2)}=\frac{3+(-1)^{i+j}}{2}$ in
the chart $(q_0,q_{-1})$. Finally, the $\star$-product is given by
\cref{eq:dualprod} by the dual prepotential
\eq{
F^{(2)}(q_0,q_{-1})=\frac{5 q_0^3}{6}+\frac{1}{2} q_{-1} q_0^2+q_{-1}^2 q_0+\frac{2 q_{-1}^3}{3}-\text{Li}_3\left(\re^{q_{-1}-q_0}\right)
}
where $\Li_3(x)=\sum_{n>0}\frac{x^n}{n^3}$ is the polylogarithm
function of order 3. 
\end{examplen}

\begin{examplen}[Tri-hamiltonian structure of dispersionless Ablowitz--Ladik]

Let now $a=b=1$, $m=0$. This case corresponds to the dispersionless limit of 
the Ablowitz--Ladik hierarchy of  \cref{exa:al}. For this case, the Frobenius
manifold structures $\cM^{(1)}_{1,1,0}$ and $\cM^{(2)}_{1,1,0}$ on $\HH_{1,1,0}$ were
constructed in \cite{Brini:2011ff}; we will review and expand on that in
light of the general result of \cref{thm:frob}. In this case, the symbol of the Lax operator reads
\eq{
\lambda(z)=\re^{v} z\frac{z-\re^{q_0}}{z-\re^{-q_0}}.
}
By the proof of Point~(i) of \cref{thm:frob} we know that $(v,q_0)$ are flat
co-ordinates for the metric $\eta^{(2)}$ defined by
\cref{eq:dualmetric}. Furthermore, Point~(ii) of \cref{thm:frob} implies that the metric $\eta^{(1)}$ is flat in this
case. By the discussion of \cref{rmk:flat}, flat co-ordinates for $\eta^{(1)}$
are given by
\ea{
v= &\frac{1}{2} \l(\log  \left(t_1+\re^{t_2}\right)+t_2\right),\\
q_0 = & \frac{1}{2} \l(\log  \left(t_1+\re^{t_2}\right)-t_2\r).
}
Notice that $t_2v-q_0$ is a flat coordinate for both $\eta^{(1)}$ and
$\eta^{(2)}$, and the flat pencil is resonant in this case. The Frobenius
potentials in the respective flat frames are
\ea{
F^{(1)}(t_1,t_2)= & \frac{1}{2} t_2 t_1^2+\re^{t_2} s_1+\frac{1}{2} s_1^2 \log \left(s_1\right)\\
F^{(2)}(v,q_0)=& v^2 q_0+ 2 v q_0^2 +\frac{7
  q_0^3}{3} + \text{Li}_3\left(\re^{2 q_0}\right)
}

A further consequence of \cref{thm:frob} is the existence of a third
compatible flat metric $\eta^{(3)}$, along with the corresponding Frobenius
manifold structure $\cM^{(3)}_{1,1,0}$. Introducing a local chart
$(s_1,s_2)$ via 
\ea{
v= & -\frac{1}{2} \left(s_2+3 \log \left(\re^{-s_2}-s_1\right)\right), \\
q_0 = & \frac{1}{2} \left(s_2+\log \left(\re^{-s_2}-s_1\right)\right).
}
gives a flat co-ordinate system for $\eta^{(3)}$ as defined in
\cref{eq:trimetric}, as indeed 
$\eta^{(3)}_{ij}=\delta_{i+j,3}$; the pencil $(\eta^{(3)})^{-1}-\epsilon
(\eta^{(2)})^{-1}$ is again resonant, since $s_2=v+3 q_0$. It follows from \cref{eq:triprod} that the
third product structure is induced by the prepotential
\eq{
F^{(3)}(s_1,s_2)= \frac{1}{2} s_2 s_1^2-\re^{-s_2} s_1-\frac{1}{2} s_1^2
\log \left(s_1\right),
}
which shows that the first and the third Frobenius structures are isomorphic,
\eq{
\cM_{1,1,0}^{(1)}\simeq\cM_{1,1,0}^{(3)}.
} 
Such isomorphism is non-trivial, in that $\eta^{(1)}$ and $\eta^{(3)}$ do not
share a common flat system and the associated Frobenius structures are not
related by an affine change of flat co-ordinates.
\end{examplen}

\section{Equivariant mirror symmetry of toric trees}
\label{sec:GW}

Let $X$ be a smooth quasi-projective variety over $\bbC$ with vanishing odd cohomologies, $T$ an algebraic
torus action on $X$ with projective fixed loci $i_j:X_j^T\hookrightarrow X$,
$j=1,\dots,r\in \bbN$. If $X$ is projective, the equivariant Gromov--Witten invariants of $(X,T)$
\cite{MR1408320} are defined as
\eq{
\bra 
\phi_{\a_1}
\dots  
\phi_{\a_n}
\ket_{g,n,\b}^{X,T} \triangleq
\int_{[\overline{\cM}_{g,n}(X,\b)]_T^{\rm vir}} \prod_{i=1}^n
\ev^*_i(\phi_{\a_i})
\in H_T({\rm pt}),
\label{eq:gw}
}
where $\overline{\cM}_{g,n}(X,\beta)$ is the stable compactification \cite{Kontsevich:1994na} of the
moduli space of degree $\beta\in H_2(X,\bbZ)$ morphisms from $n$-pointed,
genus $g$ curves to $X$, $[\overline{\cM}_{g,n}(X,\b)]_T^{\rm vir}$ is the
$T$-equivariant virtual fundamental
class of $\overline{\cM}_{g,n}(X,\b)$, 
$\phi_{\a_i}\in
H_T(X)$ are arbitrary equivariant cohomology classes of $X$, and
$\ev_i:\overline{\cM}_{g,n}(X,\b)\to X$ is the evaluation map at the $i^{\rm th}$
marked point.
\cref{eq:gw} still makes sense if $X$ is
non-compact as long as $X^T_i$ is for all $i$; in that case, we define invariants  by
their localization to the fixed locus by the Graber--Pandharipande
virtual localization formula \cite{MR1666787,
  Kontsevich:1994na}. For $T$-equivariant cohomology classes $\phi_1,\phi_2 \in
H_T(X)$, write $\eta$ for the non-degenerate inner product
\eq{
\eta(\phi_1,\phi_2)\triangleq\sum_{j=1}^r \int_{X_j^T}\frac{i_j^*(\phi_1\cup \phi_2)}{\re (N_{X/X^T})}.
\label{eq:etagw}
}
We will denote by the
same symbol the flat non-degenerate pairing on $T(H_T(X))$ obtained from \cref{eq:etagw} by identifying $T_\tau
H_T(X) \simeq H_T(X)$ $\forall~\tau\in H_T(X)$. For vector vields $\varphi_i\in \cX(H_T(X))$, $i=1,2$, the genus zero equivariant Gromov--Witten invariants
\cref{eq:gw} define further a product structure $\varphi_1 \circ
\varphi_2$ on the tangent fiber at $\tau$ 
through
\eq{
\eta(\varphi_1, \varphi_2 \circ \varphi_3) \triangleq \sum_{n\geq 0}\sum_{\beta\in
  H_2(X,\bbZ)} \bra  \phi_1, \phi_2, \phi_3, \tau^{\otimes n} \ket_{0,n+3,\b}^{X,T} 
}
which is commutative, associative, and compatible with $\eta$
\cite{MR1408320}. The corresponding Frobenius manifold structure
$QH_T(X)\triangleq (H_T(X), \eta, \circ)$ on $H_T(X)$ is the {\it $T$-equivariant quantum
cohomology of $X$}. \\

Let $\mu_i=c_1(\cO_{BT_i}(1))$ be the hyperplane class on the classifying space
$BT_i$ of the $i^{\rm th}$-factor of $T=(\bbC^{\star})^l$, and write
$\bbK\triangleq \bbC(\mu_1, \dots, \mu_n)$ for the field of fractions of $H^\bullet(BT)$.
Then $QH_T(X)$ is a finite dimensional
dual-type Frobenius manifold over $\bbK$ of charge one: it has a flat identity by the Fundamental Class
Axiom of Gromov--Witten theory, and it is generally non-conformal as a
consequence of the non-trivial grading of the ground field $\bbK$. The purpose
of this section is to exhibit an isomorphism of such Frobenius dual-type
structures with the second Frobenius structure on $\HH_{a,b,\mu}$ of
\cref{thm:frob} for a suitable family of targets. When
$X$ is the total space of the bundle $\cO_{\bbP^1}(-1)\oplus
\cO_{\bbP^1}(-1)$ and $T\simeq \bbC^*$ is the one-torus action that covers the
trivial action on the base and scales the fibers with opposite weights, it was already
shown in \cite{Brini:2010ap,Brini:2011ff} that $QH_T(X)\simeq
\cM^{(2)}_{1,1,0}$. Moreover, it was proved in \cite{Brini:2013zsa} that
$\cM^{(2)}_{a,0,\nu}$ is isomorphic to the $T$-equivariant orbifold cohomology of
the $A_{a-1}$-surface singularity, where $T\simeq \bbC^*$ acting with generic
weights specified by $\nu$. We will see how this correspondence with
Gromov--Witten theory generalizes for arbitrary $(a,b,\nu)$.

\subsection{Toric data}

Let $\cS_{a,b} = \{v_{i}\in \bbZ^3\}_{i=1}^{a+b+2}$ be the set of
three-dimensional integer vectors
\eq{
v_i = \l\{\bary{cc}
(0,a+1-i,1) & i=1,\dots, a+1, \\
(1,a+2-i,1) & i=a+2, \dots, a+b+2.
\eary\r.
}
$\cS_{a,b}$ is the skeleton of the fan of a toric variety, given by the 
cone over a triangulation of the rays $v_i$ (\cref{fig:fanorb,fig:fanres}). We can
construct it as a GIT quotient $\bbC^{a+b+2}/\!\!/(\bbC^\star)^{a+b-1}$ \cite{MR1677117} by considering the
exact sequence
\eq{\begin{xy}
(10,0)*+{0}="z"; (30,0)*+{\bbZ^{a+b-1}}="y";(60,0)*+{\bbZ^{a+b+2}}="a"; (85,0)*+{\bbZ^3}="b"; (100,0)*+{0}="c";
{\ar^{}  "z";  "y"};
{\ar^{M}  "y";"a"};
{\ar^{N}  "a";"b"};
{\ar^{} "b";"c"};
\end{xy},
\label{eq:GIT}
}
where
\ea{
M^T = &\l(
\bary{cccccccccc}
1 & -2 & 1 & 0 & 0 & 0 & \dots & 0 & \dots & 0 \\
0 & 1 & -2 & 1 & 0 & 0 & \dots & 0 & \dots & 0 \\
& \vdots & \ddots & & & \vdots & & \vdots \\
0 & \dots  & 1 & -2 & 1 & 0 & 0 & 0 & \dots & 0 \\
0 & \dots  & 0 & 1 & -1 & -1 & 1 & 0 & \dots & 0 \\
0 & \dots & 0 & 0 & 0 & 1 & -2 & 1  & \dots & 0 \\
& \vdots & & & & \ddots & \vdots \\
0 & \dots & 0 & \dots & 0 & \dots & 0 & 1 & -2 & 1
\eary
\r),
\label{eq:MT} \\
N = & \l(
\bary{ccccccccc}
0 & 0 & 0 & \dots & 0 & 1 & 1 & \dots & 1  \\
0 & 1 & 2 & \dots & a & 0 & -1 & \dots & -b  \\
1 & 1 & 1 & \dots & 1 & 1 & 1 &  \dots &  1  \\
\eary
\r).
\label{eq:N}
}
\begin{figure}[t]
\vspace{.5cm}
\begin{minipage}[t]{0.49\linewidth}
\centering
\includegraphics{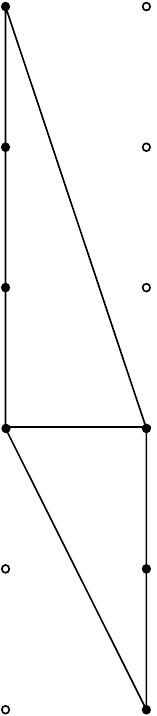}
\caption{The toric diagram of the orbifold $X_{a,b}
$ 
for $a=3$, $b=2$.}
\label{fig:fanorb}
\end{minipage} 
\begin{minipage}[t]{0.49\linewidth}
\centering
\includegraphics{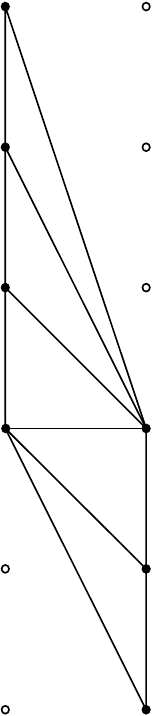}
\caption{The toric diagram of the minimal resolution $Y_{a,b}$ for
  $a=3$, $b=2$.}
\label{fig:fanres}
\end{minipage} 
\end{figure}
A triangulation of the fan corresponds to a choice of chamber in the GIT
problem, as in \cref{fig:fanorb,fig:fanres}. The picture in \cref{fig:fanorb} corresponds to the orbifold chamber in the
secondary fan of \cref{eq:MT,eq:N}; we will denote by $X_{a,b}$ the resulting singular variety. It is obtained by deleting the unstable
locus 
\eq{
X_{a,b}^{\rm us} \triangleq V\l(\prod_{i=2}^{a-1} x_i \prod_{j=2}^{b-1} x_{a+j}\r)
}
in $\bbC^{a+b+2}$ and quotienting by the $(\bbC^\star)^{a+b-1}$ action with weights specified by
$M$ in \cref{eq:MT}.
The picture in \cref{fig:fanres}
corresponds instead to the smooth (large volume) chamber: we remove the
Zariski-closed set
$Y_{a,b}^{\rm us}$ defined by 
\eq{
Y_{a,b}^{\rm us} \triangleq V\l(\prod_{j>i+1, j \neq a+1, a+2}\bra
    x_i, x_j    \ket  \prod_{j=1}^{a-1}\bra
 x_{a+1}, x_j    \ket \prod_{j=a+4}^{a+b+2}\bra
 x_{a+2}, x_j    \ket\r)
}
and then quotient by the $(\bbC^\star)^{a+b-1}$ action with weights specified by
$M$ in \cref{eq:MT}. The resulting variety, which we will denote by $Y_{a,b}$,
is a smooth quasi-projective Calabi--Yau threefold, and the variation of GIT
given by moving from \cref{fig:fanorb} to \cref{fig:fanres} is a crepant
resolution of the singularities of $X_{a,b}$. \\

\subsubsection{$T$-equivariant cohomology}

The resolution $Y_{a,b}$ can be visualized as a tree of two chains $\{L_i\}_{i=1}^{a-1}$ and $\{
L_i\}_{i=a+1}^{a+b-1}$ of $\bbP^1$ with normal bundle
$\cO+\cO(-2)$, which are then connected along a $(-1,-1)$ curve $L_a$. We will
refer to the resulting geometry as a {\it toric tree}, to reflect the shape of
the corresponding web diagram  (\cref{fig:toricweb}). Explicitly, we have
\eq{
L_i \triangleq \l\{\bary{lc}
 V(x_{i+1}, x_{a+2}) & i<a, \\
V(x_{a+1}, x_{a+2}) & i=a, \\
V(x_{a+1}, x_{i+2}) & i>a.
\eary\r.
}
The fundamental cycles $[L_j]\in
H_2(Y_{a,b},\bbZ)$ of the links of the chain are a system of generators for
$H_2(Y_{a,b},\bbZ) \simeq \bbZ^{a+b-1}$. Define $\omega_j\in
H^2(Y_{a,b},\bbZ)$ to be their cohomology duals, and $\cO(\omega_j)$ the
corresponding line bundles; by definition, they restrict to $\cO(1)$ on $L_j$,
and to the trivial bundle on $L_i$, $i\neq j$. Consider now the following $T\simeq (\bbC^\star)^2$-action on $\bbC^{a+b+2}$:
\eq{
(x_i; \sigma_1,\sigma_2) \to \l\{\bary{lc} 
\sigma_1^{-1} x_a & i=a, \\
\sigma_2 \sigma_1 x_{a+1} & i=a+1, \\
\sigma_2^{-1} x_{a+2} & i = a+2, \\
x_i & {\rm else}.
\eary\r.
\label{eq:Tab}
}
This descends to an effective torus action on $X_{a,b}$, which preserves
$K_{Y_{a,b}} \simeq \cO_{Y_{a,b}}$. Let $\{p_i\}_{i=1}^{a+b}$ denote the fixed points of the
torus action, so that $p_i$ and $p_{i+1}$ correspond to the poles of each $\bbP^1$ in the
chain. Turning on a torus action as in \cref{eq:Tab} we obtain an action on the
bundles over the links of the chain, linearized as in \cref{eq:Tab}; their equivariant first Chern classes provide lifts of $\omega_j$
to $T$-equivariant cohomology, which we will denote by the same symbol
$\omega_j \in H_T(Y_{a,b})$.

\begin{figure}[t]
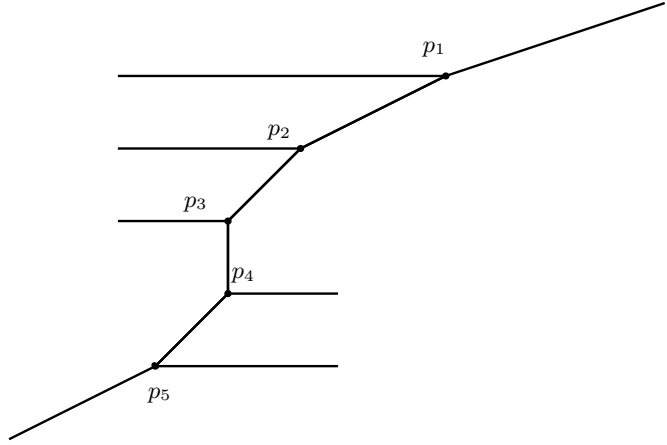
\centering
\caption{The toric web diagram of $Y_{a,b}$ for $a=3$, $b=2$.}
\label{fig:toricweb}
\end{figure}
\subsection{Mirror symmetry}

Denote $\mu_i\triangleq c_1(\cO_{BT_i}(1))$ where $\bbC^*\simeq T_i\hookrightarrow T$ are
the two cartesian projections of the two-torus $T$ acting on $Y_{a,b}$. We
have the following

\begin{thm}
\label{thm:mirror}
Let $(a,b,\nu)$ be as in \cref{defn:hurwitz}. Then
\eq{
QH_T(Y_{a,b}) \simeq \cM^{(2)}_{a,b,\nu}
\label{eq:mirror}
}
upon identifying $\nu=\mu_1/\mu_2$.
\end{thm}

\begin{proof}

The proof is given by explicit calculation of both sides of
\cref{eq:mirror}. For the r.h.s., we will use the fact that in positive degree
all genus zero Gromov--Witten invariants can
be computed by a combined use of the deformation invariance of GW invariants and the
Aspinwall--Morrison formula \cite{Aspinwall:1991ce,MR1832332,MR1421397}. The result
is \cite{MR1832332, Karp:2005vq}
\eq{
\bra \omega_{i_1} \dots
  \omega_{i_n} \ket_{0,n,\beta}^{Y_{a,b},T} = \l\{
\bary{cc}
\frac{1}{d^3} & \hbox{if $i_j=a$ for some $j$~}, \beta=d\l([L_a]+\sum_{i=k_a}^{a-1} [L_i] +
\sum_{j=a+k_b}^{a+b+1} [L_j]\r), \\ & k_\bullet=\min(\{i_j\},\bullet)\\
-\frac{1}{d^3} & \hbox{if $k_+=\max(\{i_j\})<a$ or $k_-=\min(\{i_j\})>a$ }, \\
& \beta=d(L_{k-} +\dots L_{k+}), \\
0 & \mathrm{else}.
\eary
\right.
\label{eq:posdeg}
}
When $\beta=0$ and $n=3$, Gromov--Witten invariants are defined as the equivariant triple intersection
numbers of $Y_{a,b}$, which can be computed explicitly by localization to the $T$-fixed
points from \cref{eq:Tab}. Explicitly, the restrictions of the K\"ahler
classes to the fixed loci read
\eq{
\omega_i|_{p_j} = 
\l\{
\bary{lcl}
(a-i)\mu_2 + \mu_1 &  {\rm for} & j \leq i \leq a-1, \\
0 & {\rm for} & i\leq a-1, j>i, \\
0 &  {\rm for} & i\geq a, j\leq i, \\
(a-i)\mu_2 - \mu_1 & {\rm for} & j > i\geq a. \\
\eary
\r.
}
and the moving part contribution to the Euler class is computed as
\eq{
\re_T(TM)\big|_{p_i} = 
\l\{
\bary{lcl}
-\mu_2\l((a-i)\mu_2 + \mu_1\r)\l(\mu_1+\mu_2(a-i+1)\r) & {\rm for} & i\leq a, \\
\mu_2  \l(\mu_1+(i-a-1)\mu_2\r)\l((i-a)\mu_2+\mu_1\r) & {\rm for} & i\geq a+1.
\eary
\r.
\label{eq:tw}
}
Then, denoting $s_{i,c} \triangleq \sum_{k=i}^c (\re_T(TM))^{-1}|_{p_k}$, we get
\eq{
s_{i,c}=\l\{\bary{lcl}
\frac{i-c-1}{\mu_2 ((a-c) \mu_2+\mu_1)
   ((a+1-i) \mu_2+\mu_1)} & \mathrm{for} & i<c\leq a, \\
\frac{c-i+1}{\mu_2 ((a-c) \mu_2-\mu_1)
   ((a+1-i) \mu_2-\mu_1)} & \mathrm{for} & a<i\leq c.
\eary\r.
}
and therefore,
\ea{
\bra \mathbf{1}^3\ket^{Y_{a,b},T}_{0,3,0} =& s_{1,a+b} = s_{1,a-b}, \nn \\
=& \frac{b-a}{\mu_2 (b \mu_2+\mu_1)
   (a \mu_2+\mu_1)}
}
Furthermore, for $i\leq j\leq k< a$:
\ea{
\bra \mathbf{1}^2, \omega_i \ket^{Y_{a,b},T}_{0,3,0} =& \l((a-i)\mu_2+\mu_1\r)s_{1,i}  \nn \\
=& -\frac{i}{\mu_2 (a \mu_2  + \mu_1)}, \\
\bra \mathbf{1}, \omega_i, \omega_j \ket^{Y_{a,b},T}_{0,3,0} =&
\frac{-i \l((a-j)\mu_2+\mu_1\r)}{(a \mu_2  + \mu_1)},\\
\bra \omega_i, \omega_j, \omega_k \ket^{Y_{a,b},T}_{0,3,0} =&
 \frac{i  \l((a-j)\mu_2+\mu_1\r)  \l((a-k)\mu_2+\mu_1\r)}{\mu_2  (a \mu_2+\mu_1)},
}
and for $i\geq j \geq k\geq a$
\ea{
\bra \mathbf{1}^2, \omega_i \ket^{Y_{a,b},T}_{0,3,0} =&
\l(-(i-a)\mu_2-\mu_1\r)s_{i+1,a+b},  \nn \\
=& \frac{i-a-b}{\mu_2 (b \mu_2+\mu_1)}, \\
\bra \mathbf{1}, \omega_i, \omega_j \ket^{Y_{a,b},T}_{0,3,0} =&
\frac{\l(i-a-b\r) \l((a-j)\mu_2-\mu_1\r) }{\mu_2 (b
  \mu_2+\mu_1) }, \\
\bra \omega_i, \omega_j, \omega_k \ket^{Y_{a,b},T}_{0,3,0} =&
\frac{\l(i-a-b\r) \l((a-j)\mu_2-\mu_1\r) \l((a-k)\mu_2-\mu_1\r) }{\mu_2 (b \mu_2+\mu_1)}. 
\label{eq:lasttriple}
}
Writing $\tau=\tau_0 \mathbf{1}+\sum_{i=1}^{a+b-1}\tau_i \omega_i$ for $\tau\in
H_T(Y_{a,b})$, where we set $\omega_0 \triangleq \mathbf{1}_Y$,
\crefrange{eq:posdeg}{eq:lasttriple} imply that the generating function
$F_{\rm GW}^{Y_{a,b},T}$ of the genus zero Gromov--Witten
invariants of $Y_{a,b}$ takes the form 
\ea{
F_{\rm GW}^{Y_{a,b},T}(\tau) \triangleq & \sum_{n,\b}\bra \frac{\tau^{\otimes n}}{n!}\ket^{Y_{a,b},T}_{0,n,\b}, \nn \\
=& \sum_{i,j,k} \bra \omega_i, \omega_j, \omega_k
\ket_{0,3,0}^{Y_{a,b},T} \tau_i\tau_j\tau_k + 
\sum_{l=0}^{a-1}\sum_{k=0}^{b-1}\text{Li}_3\left(\re^{\tau_a+\tau_{a-1}+\dots+\tau_{a-l}+\tau_{a+b-1}+\dots+\tau_{a+b-k}}\right)\nn
\\
						&-\sum_{k\leq l
  =1}^{a-1}\text{Li}_3\left(\re^{\tau_{k}+\dots+\tau_{l}}\right)-\sum_{k\leq l
  =a+1}^{a+b-1}\text{Li}_3\left(\re^{\tau_{k}+\dots+\tau_{l}}\right).
\label{eq:FGW}
}
As far as the r.h.s. of \cref{eq:mirror} is concerned, the prepotential of
$\cM^{(2)}_{a,b,\nu}$ can be computed analytically in closed form from
\cref{eq:dualmetric,eq:dualprod}. A rather tedious, but completely
straightforward residue calculation shows that the prepotentials coincide
\eq{
F^{\cM^{(2)}_{a,b,\nu}}(v,q_{-i},q_j)=F_{\rm GW}^{Y_{a,b},T}(\tau)
}
upon identifying flat coordinates as
\ea{
\label{eq:newcoords}
v =& \frac{\tau_0}{\mu_2}+(a-\nu)\frac{\tau_a}{2}+\sum_{j=1}^{a-1}j\tau_{j}, & \\
q_{-k} =& -\left( \frac{\tau_a}{2}+\tau_{a-1}+\dots+\tau_{a-k}\right) & k= 0,\ldots,a-1,\\
q_l =& -\left( \frac{\tau_a}{2}+\tau_{a+1}+\dots +\tau_{a+l}\right) & l=1,\ldots,b-1.
}
\end{proof}

\cref{thm:mirror} prompts the following immediate generalization of the
 conjectural correspondence of
\cite{Brini:2010ap} for the Ablowitz--Ladik hierarchy.

\begin{conj}
The full descendent all-genus Gromov--Witten potential of $(Y_{a,b},T)$ for
$\mu_1=m \mu_2$ is the logarithm of a
$\tau$-function of the $m$-generalized RR2T of bidegree $(a,b)$. 
\label{conj:gwis}
\end{conj}

In other words, the parameter $m$ in \cref{rmk:equivlax} corresponds to a
choice of weights of a resonant subtorus $\bbC^*\simeq T' \subset T$. Its
proof up to genus one will be the subject of Section~\ref{sec:genus1}.

\begin{rmk}
When $b=0$, the GIT quotient in \cref{eq:GIT} yields $Y^{a+1,0}\simeq \bbC
\times \cA_{a}$, where $\cA_a$ is the canonical resolution of the $A_a$
surface singularity. \cref{conj:gwis} then suggests that a suitable
$\tau$-function of the $q$-deformed $a$-KdV hierarchy should yield 
the total GW potential of $\bbC \times \cA_{a}$. This has interesting
implications already for the case $a=1$ and $\cA_{0}=\bbC^2$, where it would imply that the
$\tau$-function of the scalar hierarchy highlighted in \cite{Brini:2011ij} to be
underlying the generating functions of triple Hodge integrals on $\cM_{g,n}$
should be a $\tau$-function of the $q$-deformed KdV hierarchy of \cite{MR1383952}.
\end{rmk}

\subsection{Twisted periods and the Dubrovin connection}

Information on the genus zero gravitational invariants of $Y_{a,b}$ is encoded
into the the pencil of affine connections of \cref{eq:defEC}, or the {\it
  Dubrovin connection} on $QH_T(Y_{a,b})$. An immediate spin-off of \cref{thm:mirror} is an explicit characterization of
its space of solutions. \\

Let $\nu=m\in \bbZ$ and $\pi:\cU_{a,b,m} \to \HH_{a,b,m}$ be the universal
curve over the genus zero double Hurwitz space $\HH_{a,b,m}$. For $\lambda \in \HH_{a,b,m}$ we write $C_\lambda$ for the fiber
of $\pi$ at $\lambda$ and $C_{[\lambda]} \triangleq C_\lambda \setminus
\{\re^{q_0},\re^{-q_0},\{\re^{\sgn(k) q_k}\}_{k\neq 0=1-a}^{b-1}\}$. Let now $p:\tilde C_{[\lambda]} \to C_{[\lambda]}$ be
the universal covering map and, for  $\zeta\in \bbC$, fix a choice of principal branch for
$\lambda^{\zeta} = \exp(\zeta \log\lambda)$  as
\eq{
\lambda^{\zeta}(z) = z^{\zeta (m+b)}\prod_{i=1-a}^{0}|z-q_i|^{\zeta} \re^{\ri
  \zeta \arg_{i,+}(z)}  \prod_{j=0}^{b-1}|z-q_j^{-1}|^{-\zeta}\re^{-\ri \zeta \arg_{j,-}(z)}
\label{eq:lambdazeta}
}
where $\arg_{i,\pm}(z)\in [0,2\pi)$ is the angle formed by
  $z-\re^{\pm q_i}$ with $\Im(z)=0$. On the complex line $L_\lambda$
  parametrized by $\lambda^{\zeta}$, we have a
monodromy representation $\rho_\lambda : \pi_1(C_{[\lambda]})\to L_\lambda \simeq \bbC$ 
defined by local coefficients 
$l_{q_i}$ around $\re^{q_i}$ resulting in
multiplication by $\mathfrak{q}_i:=\rho_\lambda(l_{q_i}) = \re^{2\pi \ri
   \zeta \sigma_i}$, where $\sigma_i=(i+m+b+1)$ or $(i+m+b+a-1)$ for $i>0$ or
$i<0$ respectively, and we set $\mathfrak{q}_\pm=\mathfrak{q}_{0^\pm}$. Then the sheaf of sections of $\tilde C_{[\lambda]}
\times_{\pi_1(C_{[\lambda]})} L_\lambda \to C_{[\lambda]}$ defines a locally constant sheaf
$L_\lambda$ on $C_{[\lambda]}$, and we  denote by $H_\bullet(C_{[\lambda]}, L_\lambda)$ (resp.  $H^\bullet(C_{[\lambda]}, L_\lambda)$) the homology (resp.~cohomology)
groups of $C_{[\lambda]}$ twisted by the set of local coefficients determined by
$\mathfrak{q}_i$. Integrating $\lambda^{\zeta}\phi \in H^1(C_{[\lambda]},
L_\lambda)$  over $\gamma \in H_1(C_{[\lambda]},
L_\lambda)$  defines the {\it twisted period mapping}
\eq{
\bary{ccccc}
\Pi_\lambda & : & H_1(C_{[\lambda]},L_\lambda) & \to & \cO(\HH_{a,b,m}), \\
& & \gamma & \to & \int_\gamma \lambda^{\zeta}~\rd \log y.
\eary
\label{eq:periodmap}
}
\\

Let now $\mathrm{Sol}_{a,b,\nu,\zeta}$ be the $(a+b)$-dimensional $\bbC(\zeta)$-vector space of
horizontal sections the Dubrovin connection,
\eq{
\mathrm{Sol}_{\lambda} = \{s \in \cX(\HH_{a,b,m}), \nabla^{(\eta^{(2)},\zeta)}s=0 \}.
\label{eq:QDELG}
}
As for the ordinary periods of $\cM^{(1)}_{a,b,m}$ \cite{Dubrovin:1998fe}, twisted
periods are an affine basis for the space of flat coordinates of the deformed flat connection on
$\cM^{(2)}_{a,b,m}$.
\begin{prop}[\cite{MR2070050}]
\label{thm:tp}
The gradients with respect to $\eta^{(2)}$ of the twisted periods of
\cref{eq:periodmap} generate over $\bbC(\zeta)$ the solution space of the horizontality condition for the
Dubrovin connection, \cref{eq:defEC}, on $QH_T(Y_{a,b})$, 
\eq{
\mathrm{Sol}_{a,b,m,\zeta} = \mathrm{span}_{\bbC(\zeta)}
\{\nabla^{\eta(2)} \Pi_{\lambda}(\gamma) \}_{\gamma \in H_1(C_{[\lambda]},
L_\lambda)}.
}
\end{prop}

\begin{rmk}
Except for the double Hurwitz space interpretation, all of the above
generalizes trivially to the case when $\nu \in \bbC$.
\end{rmk}

\subsubsection{The twisted period mapping for $\cM^{(2)}_{a,b,\nu}$}

For generic monodromy weights, the homology with local coefficients $L_\lambda$
coincides with the integral homology of the Riemannian covering \cite{MR1930577} of $C_{[\lambda]}$,
\eq{
H^\bullet(C_{[\lambda]}, L_\lambda) \simeq  H^\bullet(\tilde
C_{[\lambda]}/[\pi_1(C_{[\lambda]}),\pi_1(C_{[\lambda]})]), \bbZ).
}
A basis of $H_1(C_{[\lambda]},  L_\lambda)$ can then be presented in the form
of compact loops $\gamma_k=[l_{0},l_{\re^{q_k\sgn(k)}}]$,
$\gamma_\pm=[l_{0},l_{\re^{\pm q_0}}]$ given by the commutator of simple
  oriented loops around zero and each of the punctures of
  $C_{[\lambda]}$. Then, the twisted periods
\ea{
\label{eq:eulerint}
\Omega_\pm  \triangleq &  \frac{\Pi_{\lambda}(\gamma_k)}{(1-\re^{2\pi\ri
    \zeta\nu})(1-\re^{\mp 2\pi\ri \zeta})},  \\
\Omega_k  \triangleq  & \frac{\Pi_{\lambda}(\gamma_k)}{(1-\re^{2\pi\ri
    \zeta\nu})(1-\re^{-\sgn(k)2\pi\ri \zeta})}, k \neq 0,
\label{eq:eulerint2}
}
give a $\bbC(\zeta,\nu)$-basis of $\mathrm{Sol}_{a,b,\nu}$ \cite{MR1424469,
  MR1930577,Brini:2013zsa}. In turn, the period integrals of \cref{eq:eulerint,eq:eulerint2}
are hypergeometric functions in exponentiated flat variables for
$\eta^{(2)}$. 
\begin{prop}
\label{prop:twistlaur}
The twisted periods of $\cM_{a,b,\nu}^{(2)}$ are given by
\ea{
\Omega_\pm 
= & \frac{\Gamma(\xi)\Gamma(1\pm \zeta)}{\Gamma(1+\xi\pm
  \zeta))}
\re^{\zeta (v+2 q_0)}\re^{\pm \xi q_0} \prod_{j=1-a, j\neq 0}^{b-1} \re^{\zeta q_j}
 \nn \\
\times &
\Phi^{[a-\theta(\pm 1),b-\theta(\pm 1)]}\l(\xi,\zeta,-\zeta,1+\xi\pm \zeta);
\{\re^{\pm q_0-q_{i}}\}_{i=1-a}^{-1}, \re^{\pm 2q_0},\{\re^{\pm q_0+q_{i}}\}_{i=1}^{b-1}\r) \label{eq:pilaur1} \\
\Omega_k 
= & \frac{\Gamma(\xi)\Gamma(1+\sgn(k)\zeta)}{\Gamma(1+\xi+\sgn(k)\zeta)} \re^{\zeta
  (v+2q_0)} \re^{\xi \sgn(k) q_k} \prod_{j=1-a, j\neq
  0}^{b-1} \re^{\zeta q_j}
 \nn \\
\times &
\Phi^{[a-\theta(k),b-\theta(k)]}\l(\xi,\zeta,-\zeta,1+\xi+\sgn(k)\zeta);
\{\re^{\sgn(k) q_k-q_{i}}\}_{i\neq k=1-a}^{0},\{\re^{\sgn(k) q_k+q_{i}}\}_{i\neq k=0}^{b-1}\r)
\label{eq:pilaur2}
}
where $\theta(x)$ is Heaviside's step function and we defined
\eq{
\label{eq:Phi}
\Phi^{[M,N]}(a, b_1, b_2,c, w_1, \dots, w_{M+N}) \triangleq  F_D^{(M+N)}(a;
\overbrace{b_1, \dots, b_1}^{M~\mathrm{times}}, \overbrace{b_2, \dots,
  b_2}^{N~\mathrm{times}}; c; w_1, \dots, w_{M+N}),
}
and $\xi\triangleq \zeta(\nu+b)$.
\end{prop}
In \cref{eq:Phi}, $F_D^{(M)}(a; b_1, \dots, b_M; c; w_1, \dots, w_M)$ is the
Lauricella function of type $D$ \cite{MR0422713}:
\eq{
F_D^{(M)}(a; b_1, \dots, b_M; c; w_1, \dots, w_M) \triangleq \sum_{i_1, \dots, i_M}
\frac{(a)_{\sum_j i_j}}{(c)_{\sum_j i_j}}\prod_{j=1}^M \frac{(b_j)_{i_j} w_j^{i_j}}{i_j!}.
\label{eq:FD}
}
where we used the Pochhammer symbol $(x)_m\triangleq
\Gamma(x+m)/\Gamma(x)$. The proof is an immediate consequence of
\cref{eq:eulerint,eq:eulerint2} and the Euler integral representation of the Lauricella
function,
\eq{
F_D^{(M)}(a; b_1, \dots, b_M; c; w_1, \dots, w_M) =
\frac{\Gamma(c)}{\Gamma(a)\Gamma(c-a)} \int_0^1
z^{a-1}(1-z)^{c-a-1}\prod_{i=1}^M (1-w_i z)^{-b_i} \rd z.
\label{eq:eulerintFD}
}

\subsection{Dispersive deformation and elliptic Gromov--Witten invariants}
\label{sec:genus1}

In this section we study the dispersive deformation of the $m$-generalized
RR2T at order $\cO(\epsilon^2)$, and describe in detail the workflow of the
proof of \cref{conj:gwis} at the genus one approximation. In order to do so, we first offer a reformulation
of \cref{conj:gwis} in the language of the theory of formal loop spaces.\\

\subsubsection{\cref{conj:gwis} as a Miura equivalence of dispersive hierarchies.}

Recall that the Principal Hierarchy of \cref{def:ph} can be thought of as a triplet $(\cM,
\{,\}^{[0]}, \HH^{[0]})$ where $\cM$ is an $n$-dimensional complex Frobenius manifold, $\{, \}^{[0]} \triangleq \{, \}_{\eta}$ in
\cref{eq:dlessbracket} is a local Poisson structure on the loop space $\LL_\cM$, and
$\HH^{[0]}=(H_{\a,p})_{\a,p}$ is a family of local functionals
$H^{[0]}_{\a,p} = \int_{S^1} h^{[0]}_{\a,p} \rd x$, $h^{[0]}_{\a,p} \in
\cO_\cM$  for $\a=1, \dots, n$ and $p\in
\bbZ^+$, giving rise to commuting Hamiltonian vector fields on $\LL_\cM$ as in
\cref{eq:ph}. When $\cM=QH_T^\bullet(Y_{a,b})\simeq \cM^{(2)}_{a,b,\nu}$, the
isomorphism of \cref{thm:mirror} induces a Poisson morphism $\LL_{QH_T^\bullet(Y_{a,b})} \simeq
\LL_{\cM^{(2)}_{a,b,\nu}}$ such that the dispersionless Toda densities
$h_{\a,p}$ pull back to the expansion of the Hamiltonian densities of the
Principal Hierarchy of $QH_T^\bullet(Y_{a,b})$, proving \cref{conj:gwis} at
the genus zero approximation. \\

For the higher genus theory, we have two, a
priori inequivalent deformations of $\{, \}^{[0]}$ and $\HH^{[0]}$, depending on a formal parameter $\epsilon$. The first one
is the spatial interpolation of the Toda lattice of \cref{sec:dless}  applied to the 2D-Toda
Hamiltonians of \cref{eq:todaham} and to the
second 2D-Toda Poisson
structure reduced on the factorization locus
$\cA^{\rm RR}$ (\cref{sec:hamstruct}): we call this the {\it RR2T deformation}. The second is the Buryak--Posthuma--Shadrin
deformation of the
Poisson structure and Hamiltonians induced by Givental's formula for the
higher genus Gromov--Witten potential \cite{MR1901075,MR2998900}; we will
refer to this as the {\it GW deformation}. In either
case,  $(\{,\}^{[0]}, (H^{[0]}_{\a,p})_{\a,p})$ deforms as
\ea{
\{\tau^\a(X),\tau^\b
(Y)\}^{[0]}  \rightarrow &
\{\tau^\a(X),\tau^\b(Y)\}^{[\epsilon]} \nn \\
=& \{\tau^\a(X),\tau^\b
(Y)\}^{[0]} +\sum_{g=1}^\infty\epsilon^{g}
\sum_{s=0}^{g+1}  \mathscr{P}^{\a,\b}_{g,s}(\tau, \tau_X, \dots,
\tau^{(s)})\delta^{(s)}(X-Y) \nn \\
h_{\a,p}^{[0]}  \rightarrow & h_{\a,p}^{[\epsilon]} \nn \\ =& h_{\a,p}^{[0]}(\tau)+\sum_{g=1}^\infty h_{\a,p}^{[g]}(\tau, \tau_X, \dots,
\tau^{(g)}) 
\label{eq:hamdef}
}
where $h_{\a,p}^{[g]}$, $\mathscr{P}^{\a,\b}_{g,s}$ are polynomials in the jet variables
$\tau^{(i)}=\de_X^{i} \tau$ ($i>0$), graded homogeneous of degrees $g$ and
$g-s+1$ respectively; these vanish for the GW deformation when $g$ is odd. \\

Both deformations come with a canonical system of
coordinates for the jet space of $\cM$ - the  tau-symmetric coordinates $\tau^\alpha$ for
the Gromov--Witten deformation, and the coefficients $(\a,\b)$ of the Lax operators of
\cref{eq:Aop,eq:Bop,eq:Lbigrad} in the deformation by lattice
interpolation. \cref{conj:gwis} can then be stated as the
existence of an $\epsilon$-dependent Poisson morphism which matches the
Poisson structures and the Hamiltonian densities, up to total derivatives, of
the two deformations. Such morphism, if it
exists, should
take the form of an element of the {\it polynomial Miura group} of
transformations of the form \cite{dubrovin:2001}
\eq{
(\a, \b) \rightarrow \tau(\a, \b)+\sum_{g> 0}\epsilon^g \mathscr{F}_{[g]}(\a,\b,
\a_x, \beta_x, \dots, \a^{(g)}, \b^{(g)}).
\label{eq:miuragen}
}
The leading order in $\epsilon$ of the sought-for Miura transformation is just the change of variables to flat
coordinates given by \cref{eq:Aop,eq:Bop,eq:lambda,eq:newcoords}. We can then rephrase
\cref{conj:gwis} as follows:\\
{\bf \cref{conj:gwis} (reloaded).} {\it There exists a polynomial Miura transformation,
  \cref{eq:miuragen}, matching the GW and the $(a,b)$ RR2T
  deformations  to all orders of the dispersive expansion.}

\subsubsection{The genus one case - strategy of the proof}

On the RR2T side, we have all the ingredients that are needed to compute the dispersive deformation
of the Principal Hierarchy: all we have to do is to take the spatial interpolation of
\cref{poisson2dtoda,eq:todaham}. On the other hand, closed-form
expressions for the Gromov--Witten dispersive deformation of the Poisson
bracket and the Hamiltonians from Givental's formula require control to all
orders of the
steepest-descent asymptotics of the oscillating integrals of $\cM$, which is
typically out of computational
reach\footnote{An alternative approach, which would lead to a proof of
  \cref{conj:gwis} sidestepping the issue of the
  Hamiltonian structure, would be to derive the Hirota bilinear equations for
the RR2T directly from Givental's formula - an approach successfully pioneered
by Milanov and Tseng \cite{MR2309158,MR2433616} for the extented bigraded Toda hierarchy. Unfortunately, the fact that we are dealing with the
dual Frobenius structure hampers a straightforward generalization to the case
at hand.}. However, a workaround to this problem exists in genus one,
corresponding to the
$\cO(\epsilon^2)$ approximation. In this case, the rational Miura
transformation \cite{Dubrovin:1997bv}
\eq{
\tau_\a(x) \to \tau_\a(x)+\frac{\epsilon^2}{24} \frac{\de^2}{\de_x
  \de_{t^{\a,0}}} \l(\log\det M + G(\tau)\r),
\label{eq:g1qm}
}
where
\eq{
M_{\a,\b}=c_{\a\b\gamma}\tau^\gamma_x, \quad c_{\a\b\gamma}=\de^3_{\a\b\gamma} F(\tau),
}
deforms the Principal Hierarchy associated to quantum cohomology to the
$\cO(\epsilon^2)$ truncation of the full higher genus hierarchy; here
$F$ and $G$ denote respectively the genus 0 and 1 primary Gromov--Witten
potential. Dubrovin--Zhang show \cite{Dubrovin:1997bv,dubrovin:2001} that the associated tau function satisfies the
genus one topological recursion relations, and it restricts (by construction)
on the small phase space to the primary Gromov--Witten potential.\\

As all the ingredients in \cref{eq:g1qm} are explicitly known by localization
in our case, the proof of \cref{conj:gwis} at order $\cO(\epsilon^2)$
becomes practically feasible. Our strategy to prove it can be structured in the following four
steps.\\

\begin{description}
\item[Step 1] 
Compute the deformation of the Poisson structure and the Hamiltonian densities
on the phase space of the
Principal Hierarchy from the quasi-Miura transformation, \cref{eq:g1qm}.
\item[Step 2] 
Compute the reduction of the second Poisson structure for the 2D-Toda lattice on the phase
space of the $(a,b)$ RR2T, from \cref{poisson2dtoda}, and the associated Toda
Hamiltonian densities, from \cref{eq:todaham}. Interpolate and expand in the
lattice spacing to $\cO(\epsilon^2)$.
\item[Step 3] 
Find a family of Miura transformations matching the deformed Poisson tensors
of Steps 1 and 2.
\item[Step 4]
Find a Miura group element such that the Hamiltonian densities agree after pull-back, up to total derivatives.
\end{description}

This method of proof can be automatized for given $(a,b)$ and verified
symbolically; a parametric statement in $(a,b)$ hinge on performing Step~1 (in
particular the computation of \cref{poisson2dtoda} on the factorization locus)
parametrically in these two variables. The relevant computer code is
available upon request.

\begin{rmk}
A priori there is no guarantee that a Miura group element satisfying Steps~3-4
exists. However, solutions to Step~3 are guaranteed to exist by the vanishing
of the loop space Poisson cohomology in degree 1 and 2, as soon as $\cM$ has
trivial topology \cite{MR1885831,MR2105635,dubrovin:2001}: in this case there are Miura group elements
$(\mathscr{F}_{\rm RR2T}, \mathscr{F}_{\rm GW})$ such that the deformed Poisson brackets
are trivialized to their $\epsilon=0$ limit,
\eq{
\mathscr{F}_{\rm RR2T}^* \{, \}^{[\epsilon]}_{\rm RR2T} =  \{, \}^{[0]} = \mathscr{F}_{\rm GW}^* \{, \}^{[\epsilon]}_{\rm GW},
}
to all orders in $\epsilon$. Furthermore, such Miura group elements are far from unique: for any
formal $\epsilon$-series $K$ with values in graded-homogeneous differential polynomials, 
\beq
K=\sum_{g\geq 0}\epsilon^g K_{[g]}(\tau, \dots, \tau^{(g)}), \quad \deg K_{[g]}=g
\eeq
composing $\mathscr{F}_{\rm RR2T}^*$, $\mathscr{F}_{\rm GW}^*$ from the left with the time-$\epsilon$
canonical transformation,
\eq{
\tau^\a \to \tau^\a+\sum_{g > 0} \frac{\epsilon^g}{g!} \overbrace{\l\{K,\l\{K,
{\dots}, \l\{K,\tau^\a\r\}^{[0]}\r\}^{[0]}\r\}^{[0]}}^{g~{\rm times}}
\label{eq:Kcan}
}
leaves $\{,\}^{[0]}$ invariant to all orders in $\epsilon$. Proving Step~4 amounts then to show that there
exists (at least) one such $K$ to $\cO(\epsilon^2)$ such that the Toda-deformed Hamiltonians pull
back to the GW-deformed ones under composition.
\end{rmk}

\begin{rmk}
In fact, when it comes to Step~4 it is sufficient to show that the two
deformations agree on a {\it single} Hamiltonian $\bar H^{[\epsilon]}$. Once this is done, the
involutivity condition with the perturbed Hamiltonian,
\eq{
\{ \bar H^{[\epsilon]}, H_{\a,p}^{[\epsilon]} \}=0,
}
admits, order by order in $\epsilon$, a unique solution for the dispersive
deformation of the Hamiltonian densities in \cref{eq:hamdef} \cite{MR2462355}. The simplest
choice is to pick $\bar H^{[0]}$ to be the dispersionless limit of the Toda
Hamiltonian given by $\tr L_1$, i.e.,
\beq
\bar H^{[0]}=\int_{S^1}\Res_{z=0}\lambda(z)\frac{\rd z}{z},
\eeq
with the RR2T and GW perturbations computed from \cref{eq:Aop,eq:Bop} and
\cref{eq:g1qm} respectively.
\label{rmk:trl}
\end{rmk}

\begin{rmk}
A further simplification in the computations comes from the fact that it is
sufficient to prove \cref{conj:gwis} for the genus one deformation of the
Principal Hierarchy with $G=0$; switching $G$ - the elliptic GW potential - to
an arbitrary function on the small phase space amounts to composing the result with an explicit,
polynomial Miura group element. This simplifies considerably the proof of \cref{conj:gwis}.
\label{rmk:G=0}
\end{rmk}

\subsubsection{Step 1}

We start first with the following technical 
\begin{lem}[\cite{Dubrovin:1997bv}]
The genus $1$ topological deformation of the principal hierarchy associated to a semi-simple Frobenius manifold with potential $F$, flat coordinates $\tau^1,\ldots,\tau^N$ and flat metric $\eta$ is Miura-equivalent, up to higher genera, to the following deformation of the Poisson structure:
\eq{
\label{g1DZPoisson}
\begin{split}
\{\tau^\alpha(x),\tau^\beta(y)\}^{[\epsilon]}_{\rm GW}=&\eta^{\alpha\beta} \delta'(x-y) + \frac{\epsilon^2}{24} \left(c^{\alpha \beta \mu}_\mu(\tau(x)) + c^{\alpha \beta \mu}_\mu(\tau(y)) \right) \delta'''(x-y)\\
&-\frac{\epsilon^2}{24}\left(\partial_x\left(c^{\alpha \beta \mu}_\mu(\tau(x))\right) + \partial_y\left(c^{\alpha \beta \mu}_\mu(\tau(y))\right) \right) \delta'(x-y)+O(\epsilon^4) 
\end{split}
}
and Hamiltonian densities:
\eq{
\label{g1DZHamiltonian}
\begin{split}
h_{\beta,p}^{\rm GW} = & h_{\beta,p}^{[0]}  \\ &+\frac{\epsilon^2}{24} \left( \frac{\partial h_{\beta,p-1}^{[0]}}{\partial u^\zeta} \left(c^\zeta_{\nu\gamma}c^{\mu\nu}_{\alpha\mu}-c^\zeta_{\mu\nu\alpha}c^{\mu\nu}_{\gamma}\right)- \frac{\partial h_{\beta,p-2}^{[0]}}{\partial u^\zeta} c^\zeta_{\delta\sigma} c^{\sigma\mu}_\mu c^\delta_{\alpha\gamma}\right) \tau^\alpha_x \tau^\gamma_x +O(\epsilon^4)\\
\end{split}
}
where $c_{\alpha\beta\gamma}$ and $c_{\alpha\beta\gamma\delta}$ denote the third and fourth derivatives of $F$, respectively, and the indices are raised and lowered by $\eta$.
\end{lem}
As per \cref{rmk:G=0}, the Miura-equivalence appearing in the above theorem is a change of coordinates of the form
\begin{equation}\label{Miura}
\tilde\tau^\alpha = \tau^\alpha 
+\epsilon^2 \left(A^{\a}_{\mu\nu}(\tau)\tau^\mu_x\tau^\nu_x+B^{\a}_\mu(\tau)\tau^\mu_{xx}\right) + \cO(\epsilon^4)
\end{equation}
which can be explicitly computed in terms of the $G$-function of the Frobenius
manifold. 

\begin{rmk}
\cref{g1DZHamiltonian} expresses the dispersive deformation of the $p^{\rm
  th}$-Taylor coefficient of the canonically-normalized flat sections of the
Dubrovin connection for $Y_{a,b}$. However, by \cref{rmk:trl}, we will be
mainly interested in deforming the dToda flow generated by the residue of
the Lax symbol at infinity: since we are dealing with the second structure
$\cM_{a,b,\nu}$ on $\cM$, this is equivalent to the twisted period around a
Pochhammer loop encirling 1 and $\infty$,
\cref{eq:periodmap}, with the parameter $\zeta$ in \cref{eq:lambdazeta} set
equal to one. This little twist in the story amounts to resum
$\zeta^p h_{\b,p}^{\rm GW}$ w.r.t. $p$ in \cref{g1DZHamiltonian}, and then
evaluating the result at $\zeta=1$, which gives
\eq{
\label{g1DZtrl}
\begin{split}
\overline{h}^{\rm GW} = & \overline{h}^{[0]} +\frac{\epsilon^2}{24} 
\frac{\partial \overline{h}^{[0]}}{\partial \tau^\rho}
\left(c^\rho_{\nu\gamma}c^{\mu\nu}_{\alpha\mu}-c^\rho_{\mu\nu\alpha}c^{\mu\nu}_{\gamma}-
c^\rho_{\delta\sigma} c^{\sigma\mu}_\mu c^\delta_{\alpha\gamma}\right)
\tau^\alpha_x \tau^\gamma_x +O(\epsilon^4) \\
\end{split}
}
\end{rmk}

\begin{exa}[$(a,b,m)=(1,1,0)$]
This is the case of the Ablowitz--Ladik hierarchy. Here, \cref{eq:FGW,g1DZPoisson} together imply that the deformation of
the Poisson bracket is trivial at $\cO(\epsilon^2)$,
\eq{
\l\{ \tau^\a(x),\tau^\beta(y)\r\}^{[\epsilon]}_{\rm GW}=-\mu_2^{-2}\delta^{\a+\b,1}\delta'(x-y)+\cO(\epsilon^4),
}
whereas the first Hamiltonian density gets corrected as
\ea{
\bar h^{\rm GW}= & 
\re^{-\tau_0/\mu_2}(1-\re^{\tau_1})+
\frac{\epsilon^2 \re^{\tau_1-\frac{\tau_0}{\mu _2}}}{24 \mu _2 \left(\re^{\tau_1}-1\right)}
\Bigg[2 \mu _2
  \left(\left(\re^{\tau_1}-1\right) \left(\tau_0'(x)\right)^2+\re^{\tau_1}
  \left(\tau_1'(x)\right)^2\right) \nn \\ -& \left(4 \left(\re^{\tau_1}-1\right) \tau_0'(x)-\tau_1'(x)\right) \tau_1'(x)\Bigg]
+\cO(\epsilon^4)\nn \\
}
\end{exa}

\begin{exa}[$(a,b)=(1,2,0)$]
In this case the dispersionless Poisson bracket does get corrected from
\cref{g1DZPoisson}. Setting $\mu_2=1$ for notational simplicity, we find
\eq{
\{\tau^\a(x), \tau^\b(y)\}^{[\epsilon]}_{\rm GW}=\{\tau^\a(x), \tau^\b(y)\}^{[0]}+\epsilon^2 \cT(\tau,
\tau_x, \tau_{xx}) 
\l\{
\bary{rl}
0 & \a=0~{\rm or}~\b=0, \\
1 & \a=\b=1, \\
-2 & (\a,\b)=(2,1), (1,2),\\
4 & \a=\b=2,
\eary
\r.
}
with
\ea{
\cT(\tau, \tau_x, \tau_{xx})= & \frac{\re^{\tau _2(x)}}{12 \left(\re^{\tau
    _2(x)}-1\right)^4} \Bigg[
\left(\re^{\tau _2(x)}-1\right) \left( 2\delta^{(3)}(x-y)
+3  \left(\re^{\tau _2(x)}+1\right) \tau _2'(x) \delta
''(x-y) \right) \nn \\
+ & 
\left(\left(4 \re^{\tau_2(x)}+\re^{2 \tau_2(x)}+1\right) \tau
_2'(x){}^2-\left(\re^{2 \tau_2(x)}-1\right) \tau_2''(x)\right) \delta '(x-y).
}
The first dToda density reads here
\eq{
\bar h^{[0]} = \re^{\tau _0(x)} \left(\re^{\tau _1(x)}+\re^{\tau _1(x)+\tau
  _2(x)}-1\right),
}
and its full $\cO(\epsilon^2)$ GW-deformation can be read off from \cref{g1DZtrl}.
\end{exa}

\subsubsection{Step 2}

This step consists of a straightforward application of the
$\epsilon$-interpolation to \cref{poisson2dtoda,eq:todaham}. For the sake of
readability, we exemplify it in the two instances considered above.

\begin{exa}[$(a,b,m)=(1,1,0)$]
As opposed to the GW-deformation, the RR2T-deformed Poisson bracket receives
in this case corrections to all (even and odd) orders in $\epsilon$, as is
apparent from \cref{eq:pbavm}. The continuous interpolation leads to
\ea{
\{\a(x),\a(y)\}^{[\epsilon]}_{\rm RR2T} =& 0,\\
\{ \log \a(x) , \log \b(y) \}^{[\epsilon]}_{\rm RR2T} =& \epsilon^{-1}\l(\delta
(x-y+\epsilon) - \delta(x-y)\r) \nn \\ =& \delta'(x-y)+\frac{\epsilon}{2}\delta''(x-y)+\frac{\epsilon^2}{6}+\cO(\epsilon^3),\\
\{ \log \b(x) , \log \b(y) \}^{[\epsilon]}_{\rm RR2T} =& \epsilon^{-1}\l(\delta
(x-y+\epsilon) - \delta(x-y-\epsilon)\r)\nn \\ =& 2\delta'(x-y)+\frac{\epsilon^2}{3}\delta'''(x-y)+\cO(\epsilon^4).
}
In the same vein, \cref{eq:todaham} for $i=1$ gives
\ea{
\bar h^{\rm RR2T}=\a(x)-\b(x+\epsilon)= & \a(x)-\b(x)+\epsilon
\b'(x)-\frac{\epsilon^2}{2}\b''(x)+\cO(\epsilon^3), \nn \\
= & \a(x)-\b(x)+ \hbox{(total derivative)}.
}
\end{exa}

\begin{exa}[$(a,b,m)=(1,2,0)$]

\cref{poisson2dtoda} computes the full-dispersive Poisson bracket on the factorization locus as
\eq{
\begin{split}
\{\alpha_1(x),\alpha_1(y)\}^{[\epsilon]}_{\rm RR2T} &= 0\\
\epsilon \{\alpha_1(x),\b_1(y)\}^{[\epsilon]}_{\rm RR2T} &=  \b_1(y) \left(\alpha _1(y-\epsilon ) \delta (x-y+\epsilon )-\alpha _1(y) \delta (x-y)\right)\\
\epsilon \{\alpha_1(x),\b_2(y)\}^{[\epsilon]}_{\rm RR2T} & = \b_2(y) \left(\alpha _1(y-2 \epsilon ) \delta (x-y+2 \epsilon )-\alpha
   _1(y) \delta (x-y)\right) \\
\epsilon \{\b_1(x),\b_1(y)\}^{[\epsilon]}_{\rm RR2T} &=  \left(\b_2(y+\epsilon )-\b_1(y) \b_1(y+\epsilon )\right) \delta
   (x-y-\epsilon ) \\ &+\left(\b_1(y) \b_1(y-\epsilon )-\b_2(y)\right) \delta (x-y+\epsilon ) \\
\epsilon \{\b_1(x),\b_2(y)\}^{[\epsilon]}_{\rm RR2T} &= \b_2(y) \left(\b_1(y-2 \epsilon ) \delta (x-y+2 \epsilon )-\alpha
   _2(y+\epsilon ) \delta (x-y-\epsilon )\right) \\
\epsilon \{\b_2(x),\b_2(y)\}^{[\epsilon]}_{\rm RR2T} &=\b_2(y) \bigg[(-\b_2(y+2 \epsilon )
\delta (x-y-2 \epsilon )-\b_2(y+\epsilon ) \delta (x-y-\epsilon ) \\ &+ \alpha
   _3(y-\epsilon ) \delta (x-y+\epsilon )+\b_2(y-2 \epsilon ) \delta (x-y+2 \epsilon )\bigg] 
\end{split}
}
It should be noticed that the Poisson bracket is not logarithmically constant
in these coordinates. The full-dispersive deformation is given by
Taylor-expanding the r.h.s. in $\epsilon$. As before, the full-dispersive first
Hamiltonian is here given as
\ea{
\bar h^{\rm RR2T}=\a_1(x)-\b_1(x+\epsilon)= & \a_1(x)-\b_1(x)+ \hbox{(total derivative)}.
}
\end{exa}

\subsubsection{Step 3}

The next step is to match the Poisson structures $\{,\}^{[\epsilon]}_{\rm GW}$ and
$\{,\}^{[\epsilon]}_{\rm RR2T}$ computed in Steps~1-2. We will do this by explicitly
computing the trivializing polynomial Miura transformation that transforms
them back to their undeformed expression. We start from the GW-deformation.
\begin{lem}
The Miura transformation
\begin{equation}
\tau^\alpha \mapsto \tau^\alpha - \frac{\epsilon^2}{24} \left(\partial_x^2 c^{\alpha\mu}_\mu \right) + O(\epsilon^4)
\end{equation}
transforms the Poisson bracket (\ref{g1DZPoisson}) to
\begin{equation}\label{standardPoisson}
\{\tau^\alpha(x),\tau^\beta(y)\}=\eta^{\alpha\beta} \delta'(x-y) +O(\epsilon^4)
\end{equation}
and the Hamiltonian densities (\ref{g1DZHamiltonian}) to
\begin{equation}\label{g1DZHamiltoniant}
h_{\beta,p} =  h_{\beta,p}^{[0]}  -\frac{\epsilon^2}{24} \left( \frac{\partial h_{\beta,p-1}^{[0]}}{\partial \tau^\zeta} c^\zeta_{\mu\nu\alpha}c^{\mu\nu}_{\gamma}+ \frac{\partial h_{\beta,p-2}^{[0]}}{\partial \tau^\zeta} c^\zeta_{\delta\sigma} c^{\sigma\mu}_\mu c^\delta_{\alpha\gamma}\right) \tau^\alpha_x \tau^\gamma_x +O(\epsilon^4)
\end{equation}
\end{lem}
\begin{proof}
The proof is an immediate consequence of the formula $\tilde{P}^{\alpha\beta} = (L^*)^\alpha_\mu \circ P^{\mu\nu} \circ L^\beta_\nu$ for the transformation of the differential operator $P^{\alpha\beta}$ associated to the Poisson bracket, where $(L^*)^\alpha_\mu = \sum_{s\geq 0} \frac{\partial \tilde\tau^\alpha}{\partial \tau^\mu_s} \partial_x^s$ and $L^\beta_\nu = \sum_{s\geq 0} (-\partial_x)^s \circ\frac{\partial \tilde\tau^\beta}{\partial \tau^\nu_s}$, with $\tau^\alpha_s=\partial_x^s \tau^\alpha$. For the Hamiltonians one simply evaluates the functions at the shifted values and performs Taylor's expansion.
\end{proof}

One by-product of the Lemma is that the expression for the deformed Hamiltonian
densities simplifies as well in this Miura-deformed coordinates. On the RR2T
side, we act in the same way - by plugging an arbitrary Miura transformation
that trivializes $\{,\}^{[\epsilon]}_{\rm RR2T}$ to $\cO(\epsilon^2)$ and solving the
ensuing overconstrained differential system.

\begin{exa}[$(a,b,m)=(1,2,0)$]
In this case, a trivialization of the RR2T-deformed Poisson bracket reads,
at $\cO(\epsilon^2)$,
\ea{
\a_1(x)\to & \alpha _1(x)-\frac{1}{2} \epsilon  \alpha _1'(x)+\epsilon ^2 \left(\frac{5}{24} \alpha _1''(x)-\frac{\alpha _1'(x){}^2}{12 \alpha _1(x)}\right)+\cO\left(\epsilon
   ^3\right), \\
 \beta_1(x)\to & \beta_1(x)+\frac{\epsilon ^2}{24 \left(\beta _1(x){}^2-4 \beta
   _2(x)\right){}^2 \beta _2(x){}^2}
\Bigg[2 \beta _1(x){}^5 \beta _2'(x){}^2-\beta _2(x){}^2 \beta _1(x){}^3 \beta
  _1'(x){}^2+\nn \\
-& 14 \beta _2(x) \beta _1(x){}^3 \beta
   _2'(x){}^2+16 \beta _2(x){}^2 \beta _1(x){}^2 \beta _1'(x) \beta _2'(x)-20
\beta _2(x){}^3 \beta _1(x) \beta _1'(x){}^2 \nn \\
+& 32 \beta _2(x){}^3 \beta _1'(x) \beta
   _2'(x)-2 \beta _2(x) \beta _1(x){}^5 \beta _2''(x)+\beta _2(x){}^2 \beta
_1(x){}^4 \beta _1''(x)\nn \\ 
+& 10 \beta _2(x){}^2 \beta _1(x){}^3 \beta _2''(x)+4 \beta_2(x){}^3 \beta
_1(x){}^2 \beta _1''(x)-8 \beta _2(x){}^3 \beta _1(x) \beta _2''(x)
\nn \\
-& 32 \beta _2(x){}^4 \beta _1''(x)\Bigg]+\cO\left(\epsilon ^3\right), \\
\beta_2(x)\to &\beta _2(x)+\frac{1}{2} \epsilon  \beta _2'(x)+\epsilon ^2
\left(\frac{\beta _2'(x){}^2}{4 \beta _2(x)}-\frac{1}{8} \beta
_2''(x)\right)+\cO\left(\epsilon ^3\right),
}
so that in the new variables we have
\ea{
\{\a_1(x),\a_1(y)\}^{[0]}=& 0, \nn \\
\{\a_1(x),\b_1(y)\}^{[0]}= & \alpha _1(x) \delta (x-y) \beta _1'(x)+\alpha _1(x)
\beta _1(x) \delta '(x-y), \nn \\
\{\a_1(x),\b_2(y)\}^{[0]}= & 2 \alpha _1(x) \delta (x-y) \beta _2'(x)+2 \alpha _1(x)
\beta _2(x) \delta '(x-y), \nn \\
\{\b_1(x),\b_1(y)\}^{[0]}= &
\l(2 \beta _1(x)  \beta _1'(x)- \beta _2'(x)\r)\delta(x-y)+2\l( \beta
_1(x){}^2 -\beta_2(x)\r)\delta '(x-y),\nn \\
\{\b_1(x),\b_2(y)\}^{[0]}=&
3 \beta _1(x) \delta (x-y) \beta _2'(x)+3 \beta _1(x) \beta _2(x) \delta
'(x-y),\nn \\
\{\b_2(x),\b_2(y)\}^{[0]}=&
6 \beta _2(x) \delta (x-y) \beta _2'(x)+6 \beta _2(x){}^2 \delta '(x-y).
\label{eq:triv12}
}
Relating now $(\a, \b)$ to $q$ as in \cref{eq:lambda} and composing with
\cref{eq:newcoords} to go to $\tau$-variables returns $\{,\}^{[0]}=\{,\}_\eta$, the dispersionless
Poisson bracket in flat coordinates for the metric $\eta^{(2)}$ of
$\cM_{1,2,0}^{(2)}$. The general Miura group element trivializing $\{,\}^{\rm
  RR2T}$ is obtained by composing \cref{eq:triv12} with \cref{eq:Kcan}, for an 
arbitrary $K$.
\end{exa}

\subsubsection{Step 4}

All is left to do at this stage is to find a canonical generator $K$ such that
$\bar h^{\rm GW}$ matches with $\bar
h^{\rm RR2T}$ in the resulting trivializing coordinate system, up to total
derivatives. The quickest way to do it is as follows: choose $K$ such that 
the transformed $\bar h^{\rm RR2T}$ has 1) no linear term in $\epsilon$ and 2)
no linear terms in $\tau^\a_{xx}$ at $\cO(\epsilon^2)$: this amounts to the
solution of two inhomogeneous linear systems of rank $(a+b)$. Then compose
$\bar h^{\rm RR2T}$ in the resulting coordinate system with a further canonical transformation generated by a differential
polynomial $\tilde K$, with vanishing linear term in $\epsilon$. Now imposing
that the difference of the transformed $\bar h^{\rm RR2T}$ with $h^{\rm GW}$
is a total derivative is equivalent to a rank $\binom{a+b+1}{2}$ linear
system in the derivatives of the components of $\tilde K$; checking
compatibility of the solution then concludes the proof.

\begin{exa}[$(a,b,m)=(1,2,0)$]
Let us see this explicitly at work in the case when $(a,b,m)=(1,2,0)$. The
GW- and RR2T-deformed Hamiltonian density in the coordinates for which the Poisson is in
standard form, \cref{g1DZHamiltoniant}, read here
\ea{
\bar h^{\rm GW}= & \re^{\tau _0} \left(\re^{\tau _1}+\re^{\tau _1+\tau _2}-1\right)-\frac{\epsilon ^2}{24 \left(\re^{\tau _1}-1\right)
 \left(\re^{\tau _2}-1\right){}^2 \left(\re^{\tau _1+\tau
   _2}-1\right)}\nn \\
\times &  \bigg[ \re^{\tau _0+\tau _1+\tau_2} 
\left(\re^{\tau _1}-1\right) \left(3 \re^{\tau _2}
- \re^{2
   \tau _2}+5 \re^{\tau _1+\tau _2}-5 \re^{\tau _1+2 \tau _2}+2 \re^{\tau _1+3
  \tau _2}-4\right) \left(\tau _2'\right){}^2 \nn \\
-& 2 \left(\re^{\tau _2}-1\right) \left(\re^{\tau _1+\tau _2}-1\right)
   \tau _2''\Big)
      +2 \left(\re^{\tau _1}-1\right) \left(\re^{\tau _2}-1\right){}^2
      \left(\re^{\tau _2}+1\right) \left(\re^{\tau _1+\tau _2}-1\right)
      \left(\tau _0'\right){}^2 \nn \\ +& 4
   \left(\re^{\tau _1}-1\right) \left(\re^{\tau _2}-1\right) \left(\re^{\tau _1+\tau _2}-1\right) \left(\left(\re^{2 \tau _2}-1\right) \tau _1'
   +\re^{\tau _2} \left(\re^{\tau _2}-2\right) \tau
   _2'\right) \tau _0'\nn \\ +& \left(\re^{\tau _2}-1\right){}^2 \left(-2 \re^{\tau _1}+\re^{\tau _2}-2 \re^{\tau _1+\tau _2}+2 \re^{2 \left(\tau _1+\tau _2\right)}+2 \re^{2 \tau _1+\tau _2}-2 \re^{\tau
   _1+2 \tau _2}+1\right) \left(\tau _1'\right){}^2  \nn \\ +& 2 \re^{\tau _2} \left(\re^{\tau _1}-1\right) \left(\re^{\tau _2}-1\right) \left(-\re^{\tau _2}-4 \re^{\tau _1+\tau _2}+2 \re^{\tau _1+2 \tau
   _2}+3\right) \tau _1' \tau _2'\bigg]+\cO\left(\epsilon
   ^4\right)\\
\bar h^{\rm Toda} =& 
\re^{\tau _0} \left(\re^{\tau _1}+\re^{\tau _1+\tau _2}-1\right)+\frac{\epsilon}{2} \re^{\tau _0}   \left(\left(3 \re^{\tau _1}+3 \re^{\tau _1+\tau _2}-2\right) \tau _0'+3
   \re^{\tau _1} \left(\left(\re^{\tau _2}+1\right) \tau _1'+\re^{\tau _2}
   \tau _2'\right)\right)\nn \\
+& \frac{\epsilon ^2 \re^{\tau _0}}{24 \left(\re^{\tau _2}-1\right){}^2}  \bigg[\left(\re^{\tau _2}-1\right){}^2\left(3 \left(9 \re^{\tau _1}+9 \re^{\tau
   _1+\tau _2}-4\right)  \left(\tau _0'\right){}^2+27 \re^{\tau _1} \left(\re^{\tau _2}+1\right)  \left(\tau
   _1'\right){}^2\r) \nn \\
+& 54 \re^{\tau _1+\tau _2} \left(\re^{\tau _2}-1\right){}^2 \tau _1' \tau _2'+54 \re^{\tau _1} \left(\re^{\tau _2}-1\right){}^2 \tau _0' \left(\left(\re^{\tau
   _2}+1\right) \tau _1'+\re^{\tau _2} \tau _2'\right)+30 \re^{\tau _1+\tau
  _2} \left(\tau _2'\right){}^2 \nn \\ 
-& 51 \re^{\tau _1+2 \tau _2} \left(\tau _2'\right){}^2+27 \re^{\tau
   _1+3 \tau _2} \left(\tau _2'\right){}^2+10 \big(3 \re^{\tau _1} +2 \re^{\tau _2}- \re^{2 \tau _2} -3 \re^{\tau _1+\tau _2} -3 \re^{\tau
   _1+2 \tau _2} \nn \\ 
+& 3 \re^{\tau _1+3 \tau _2} -1\big) \tau_0''
+10 \big(3 \re^{\tau _1} \tau _1''-3 \re^{\tau _1+\tau _2} \tau _1''-3
\re^{\tau _1+2 \tau _2} \tau_1''+3 \re^{\tau _1+3 \tau _2}\big) \tau _1'' \nn
\\ +& \big( 2\re^{\tau _1} +30 \re^{\tau _1+\tau _2}-60 \re^{\tau _1+2 \tau _2} +28 \re^{\tau _1+3 \tau _2}\big) \tau_2''\bigg]
+\cO\left(\epsilon^3\right)
}
Let us first get rid of the linear term in $\epsilon$ in $\bar h^{\rm
  RR2T}$, as well as of the terms linear in the second derivatives. This is accomplished by an $\cO(\epsilon^2)$ transformation
generated by
\ea{
K_1= &\left(-\tau _0-\tau _1-\frac{\tau _2}{2}\right)  -\frac{\epsilon e^{-\tau _2} }{24 \left(e^{\tau
   _2}-1\right)} 
\bigg[ \left(-26 e^{\tau _2}+26 e^{2 \tau _2} +e^{3 \tau _2} -1\r)\tau
  _0'\nn \\ -& 10\l( e^{\tau _2} - e^{2 \tau _2}\r)\tau_1' +\l(-8 e^{\tau _2} +2 e^{2 \tau _2} -e^{3 \tau _2} -1\r)\tau_2'\bigg]+\cO\left(\epsilon ^2\right)
}
Composing this with a canonical transformation generated by
\eq{
K_2 = \epsilon\l(K_2^{(0)}(\tau) \tau^0_x+ K_2^{(1)}(\tau) \tau^1_x+ K_2^{(2)}(\tau) \tau^2_x\r)
}
such that $h^{\rm RR2T}-h^{\rm GW} = \de_x f$ gives a system of six linear
equations in
the $\tau$-derivatives of $K_2^{(i)}$, which is solved by 
\ea{
\de_{\tau^0}K_2^{(1)} = & \de_{\tau^1}K_2^{(0)}-\frac{e^{\tau _1}+e^{\tau _1+\tau _2}-2}{24
   \left(e^{\tau _1}-1\right) \left(e^{\tau _1+\tau _2}-1\right)},\nn \\
\de_{\tau^0}K_2^{(2)} = & \de_{\tau^2}K_2^{(0)}-\frac{e^{-\tau _2} \left(-e^{\tau _2}+e^{2 \tau _2}-e^{3 \tau _2}+e^{\tau _1+\tau _2}+e^{\tau _1+4 \tau _2}-1\right)}{24 \left(e^{\tau
   _2}-1\right) \left(e^{\tau _1+\tau _2}-1\right)},\nn \\
\de_{\tau^1}K_2^{(2)} = & \de_{\tau^2}K_2^{(1)}-\frac{e^{\tau _1+\tau _2} \left(e^{\tau _1}+e^{\tau _1+\tau _2}-2\right)}{24 \left(e^{\tau _1}-1\right) \left(e^{\tau _2}-1\right) \left(e^{\tau_1+\tau _2}-1\right)},
}
which is immediately shown to be compatible. 

\end{exa}

\subsection{Further applications}

\cref{prop:twistlaur} has a number of applications for the study of the
gravitational quantum cohomology of $Y_{a,b}$ as well as of its higher
genus Gromov--Witten theory. When $b=0$, these were explored in detail in
\cite{Brini:2013zsa}: we highlight below the main features of their
generalization to arbitrary $(a,b)$.

\subsubsection{Twisted periods and the $J$-function}

A distinguished basis of flat co-ordinates
for the Dubrovin connection is given by the generating function of genus zero one-point
descendent invariants of $Y_{a,b}$, or {\it $J$-function} \cite{MR1653024},
\eq{
J^\alpha_{Y_{a,b}}(\tau,\zeta) \triangleq \frac{\delta^{\alpha,0}}{\zeta}+\tau^\alpha+\zeta \sum_{n\geq
  0}\sum_{\beta\in\bbZ}\frac{1}{n!}\bra \frac{\omega_\a}{1-\zeta \psi}, \tau^{\otimes
  n}\ket_{0,n+1,\beta}^{Y_{a,b}},
\label{eq:Jfun}
}
where $\psi$ is a cotangent line class and the denominator is a formal
geometric series expansion in $\zeta \psi$. Write $i_j : p_j \hookrightarrow
Y_{a,b}$ for the embedding of the $j^{\rm th}$ fixed point into $Y_{a,b}$, and
define $u^{\rm cl}_j \mathbf{1}_j \triangleq i_j^*(\tau^\alpha\omega_\alpha) \in
H_T(\{p_j\})$. The coefficients $u^{\rm cl}_j(\tau)$ are linear functions
$u^{\rm cl}_j=\sum c_{j\alpha} \tau^\alpha$, and they are canonical
co-ordinates for the classical equivariant cohomology algebra of $Y_{a,b}$.
By the Divisor Axiom of Gromov--Witten theory, the coefficients $c_{j\alpha}=\omega_\alpha|_{p_j}$
are local exponents of \cref{eq:defEC} at the Fuchsian point $\mathrm{LR}=\{\tau_\a=-\infty\}$,
and the vector of the localized components $J^j_{Y_{a,b}} \mathbf{1}_j=i^*_j(J^\alpha_{Y_{a,b}}\omega_\a)$
of the $J$-function diagonalizes the monodromy around LR with weights
$c_{j\alpha}\in\bbC$, 
\eq{
J^j_{Y_{a,b}} \sim \zeta \re^{\zeta u_j^{\rm cl}}(1+\cO(\re^{\tau})).
\label{eq:JLR}
}
The asymptotic behavior \cref{eq:JLR} at LR characterizes uniquely the
localized components of the $J$-function as
a flat coordinate system for \cref{eq:defEC}. Knowledge of the monodromy
properties of the twisted periods of \cref{eq:pilaur1,eq:pilaur2} at LR is then
sufficient, by \cref{prop:twistlaur}, to give a closed form expression for the
$J$-functions as a hypergeometric function in exponentiated flat
variables. This can be achieved via an iterated use of the connection formula
at infinity for the Gauss function, as explained in
\cite[Appendix~C]{Brini:2013zsa}. In vector notation, the final result in our case is
\eq{
J_{Y_{a,b}} = (\cA^{(a)} \oplus \cA^{(b)}) \Pi
}
where
\ea{
\cA^{(a)}_{ij} =& \left\{\bary{cl} \re^{\pi\ri i \zeta}
\frac{\zeta \Gamma(1+\zeta\nu-(i+1) \zeta)}{\Gamma(1-\zeta) \Gamma(\zeta\nu-i \zeta)} & i=-j, \\
\re^{-\ri \pi  (\zeta\nu-\zeta (2 j+1))} \frac{\zeta \sin (\pi  \zeta)  \Gamma (1-\zeta\nu+\zeta i) \Gamma (1+\zeta\nu-\zeta (i+1))}{\pi  \Gamma (1-\zeta)}
 & i>-j,  \\
0 & i<-j.
  \eary\right.
\label{eq:matrAa} \\
\cA^{(b)}_{ij} =& \left\{\bary{cl} \re^{-\pi\ri j \zeta}
\frac{\zeta \Gamma(1+\xi+(a+j+1) \zeta)}{\Gamma(1+\zeta) \Gamma(\xi+(j+a) \zeta)} & i=j+a, \\
\re^{-\ri \pi  (\xi+\zeta (2j+1+a))} \frac{\zeta \sin (\pi  \zeta)  \Gamma (1-\xi-\zeta (j+a)) \Gamma (1+\xi+\zeta (a+j+1))}{\pi  \Gamma (1+\zeta)}
 & j>i+a,  \\
0 & j<i+a.
  \eary\right.
\label{eq:matrAb}
}
\\

In the same vein, let $H_{{\rm orb},T}(X_{a,b})$ denote the $T$-equivariant
Chen--Ruan cohomology of $[X_{a,b}]$. This has two torus fixed orbi-points
$[p^{(a)}]$ and $[p^{(b)}]$ - the
North and South pole of the base weighted projective line - with stackiness
$\bbZ_a$ and $\bbZ_b$ respectively. By the Atiyah--Bott isomorphism, $H_{{\rm
    orb},T}(X_{a,b})$ is then generated by the (Thom push-forwards) of
$\mathbf{1}_{\frac{i}{a}}$, $\mathbf{1}_{\frac{j}{b}}$, $\a=0,\dots, a-1$,
$\b=0,\dots, b-1$. Writing $x\triangleq \sum_{c\in \{a,b\},\a\in \bbZ_c} x^{\alpha,c} \mathbf{1}_{\frac{\a}{c}}$ for a point $x\in H_{{\rm  orb},T}(X_{a,b})$, the
orbifold $J$-function
\eq{
J^{\gamma,c}_{X_{a,b}}(x,\zeta) \triangleq \frac{\delta^{\alpha,0}}{\zeta}+x^\alpha+\zeta \sum_{n\geq
  0}\sum_{\beta\in\bbZ}\frac{1}{n!}\bra \frac{\mathbf{1}_{\frac{\g}{c}}}{1-\zeta \psi}, x^{\otimes
  n}\ket_{0,n+1,\beta}^{X_{a,b}},
\label{eq:Jfunorb}
}
gives a system of flat co-ordinates for the Dubrovin connection on
$T(H_{{\rm  orb},T}(X_{a,b}))$. As Frobenius manifolds, the quantum
cohomologies of $X_{a,b}$ and $Y_{a,b}$ are isomorphic
\cite{MR2510741},
with the undeformed flat coordinates related as \cite{MR2483931}
\ea{
x^{\alpha,a} = & \sum_{\g=0}^{a-1} \varepsilon_a^{\a\g}\tau^{a-1-\g}- \varepsilon_a^{\a},
\nn \\
x^{\beta,b} = & \sum_{\b=0}^{b-1} \varepsilon_b^{\b\g}\tau^{\g+a+1}-
\varepsilon_b^{\b},
}
where $\varepsilon_c \triangleq \re^{\frac{2\pi\ri}{c}}$, and by the Divisor Axiom the localized components of $J^{X_{a,b}}$ are the
unique set of flat coordinates of the Dubrovin connection such that
\eq{
J^{\gamma,c}_{X_{a,b}}(x,\zeta) \simeq \zeta  \re^{\zeta
  x^{0,c}}\l(1+\cO(x)\r).
}
This can be compared with the behavior of the twisted periods at $x=0$, where
the integrals appearing in \cref{eq:pilaur1,eq:pilaur2,eq:eulerint,eq:eulerint2} can be
explicitly evaluated in terms of the Euler Beta function. The result is
\eq{
\Pi = (\cB^{(a)}\oplus \cB^{(b)}) J_{\cX_{a,b}}
\label{eq:matrB1}
}
where
\ea{
\label{eq:matrBa}
\cB^{(a)}_{jk}= & \quad \varepsilon_a^{(j-n/2) \zeta\nu}  \frac{\varepsilon_a^{-jk}}{a}
\left\{\bary{cc} 
 -\varepsilon_a^{k/2} \frac{\Gamma\l(\frac{\zeta\nu+k}{a}\r)\Gamma(1-\zeta)}{\Gamma\l(\frac{\zeta\nu+k}{a}-\zeta\r)} &
\mathrm{\quad for} \quad 1\leq k\leq {a-1}
\\ \frac{\Gamma(\zeta\nu/a)\Gamma(1-\zeta)}{\zeta
  \Gamma(1+\zeta(1+\nu/a))} & \mathrm{\quad for} \quad
k=0 \eary\right.  \\
\cB^{(b)}_{jk}= & \varepsilon_b^{(j-n/2) (\xi+a)}  \frac{\varepsilon_b^{-jk}}{b}
\left\{\bary{cc} -\varepsilon_b^{k/2} 
\frac{\Gamma\l(\frac{\xi+a-k}{b}+1\r)\Gamma(1+\zeta)}{\Gamma\l(\frac{\xi+a-k}{b}-\zeta+1\r)} &
\mathrm{for} \quad 1\leq k\leq {b-1}
\\ \frac{\Gamma((a+\xi)/b)\Gamma(1+\zeta)}{\zeta
  \Gamma(1+(\xi+a)/b-\zeta))} & \mathrm{for} \quad
k=0 \eary\right.  
\label{eq:matrBb}
}
The composition $\cU\triangleq (\cA^{(a)}\cB^{(a)}\oplus \cA^{(b)}\cB^{(b)})$
gives the transition matrix from the vector-form of the orbifold $J$-function to the one of the resolution. Closed-form knowledge of $\cU$ has important applications to
the Crepant Resolution Conjecture \cite{MR2529944}, as well as to its generalization to open
Gromov--Witten theory \cite{Brini:2013zsa}; in particular, by the block
diagonal form of \cref{eq:matrAa,eq:matrAb,eq:matrBa,eq:matrBb}, the genus
zero results in \cite{Brini:2013zsa} generalize immediately to the case at
hand. Similarly, \cref{prop:twistlaur} makes it an exercise in book-keeping to
generalize to arbitrary $(a,b)$ the quantized Crepant Resolution Conjecture proven in \cite{Brini:2013zsa} for the case $b=0$.


\subsubsection{Pure braid group actions in quantum cohomology}

A further application of \cref{thm:mirror,prop:twistlaur} is a complete characterization of the
monodromy group of the Dubrovin connection. By \cref{thm:frob,thm:mirror}, the
open set $\cM^{(2), \rm reg}_{a,b,\nu} \triangleq \cM^{(2)}_{a,b,\nu}\setminus \Delta_{a,b,\nu}$
of regular points for the pencil of flat connections of \cref{eq:defEC} on
$\cM^{(2)}_{a,b,\nu}$ is the complement of the arrangement of hyperplanes $\{q_i=q_j\}_{i\neq
  j}$. Equivalently, it is isomorphic to the configuration space of $a+b$-distinct points in
$\bbP^{1}\setminus \{0,1,\infty\}$,
\eq{
\cM^{(2), \rm reg}_{a,b,\nu} \simeq M_{0,a+b+3}.
}
Any simple loop $\sigma$ in $\cM^{(2), \rm reg}_{a,b,\nu}$ then gives a monodromy action on $\mathrm{Sol}_\lambda$,
\eq{
M_\sigma : \pi_1(\cM^{(2),\rm reg}_{a,b,\nu}) \to \mathrm{Aut}(\mathrm{Sol}_{\lambda}),
\label{eq:msigma}
}
which is a representation of the colored braid group in $a+b+2$ strands, as
$\pi_1(M_{0,a+b+3})\simeq \mathrm{PB}_{a+b+2}$. Monodromy matrices in the
twisted period basis can be computed explicitly \cite{MR2962392}; the
resulting representation is the Gassner representation \cite{MR849651}, with weights
specified as in \cref{eq:lambda}. 

\begin{rmk}
By the $T$-equivariant version of Iritani's integral structures in quantum
cohomology, this pure braid group action carries through to an action on the
$T$-equivariant $K$-groups of $X_{a,b}$ and $Y_{a,b}$. Very recently,
pure braid group actions on the derived category of coherent sheaves were
constructed in \cite{Donovan:2013hra} for a family of toric Calabi--Yau
obtained from deformations of resolutions of type $A$ surface
singularities; when the variety is a threefold, their examples coincide
precisely with $Y_{a,b}$. It would be interesting to establish a clear link between our $D$-module construction and theirs.
\end{rmk}

\bigskip
{\footnotesize
\noindent\textit{Acknowledgments.}
We would like to thank Mattia Cafasso, John Gibbons, Chiu-Chu Melissa Liu and
Dusty Ross for discussions. We would also like to thank the American Institute of Mathematics for hosting
the workshop ``Integrable systems in Gromov--Witten and symplectic field theory'' in January 2012, during which this paper was started. A.~B. was partially supported
by a Marie Curie IEF under Project n.274345, and by an INdAM-GNFM Progetto
Giovani 2012 grant. P.~R. was partially supported by a Chaire
CNRS/Enseignement superieur 2012-2017 grant. S.R was partially supported by
the European Research Council under the ERC-FP7 Grant n.307074.
}

\bibliography{miabiblio}

\def\cprime{$'$} \def\cprime{$'$}
\begin{bibdiv}
\begin{biblist}

\bib{MR0377223}{article}{
      author={Ablowitz, M.~J.},
      author={Ladik, J.~F.},
       title={Nonlinear differential-difference equations},
        date={1975},
        ISSN={0022-2488},
     journal={J. Mathematical Phys.},
      volume={16},
       pages={598\ndash 603},
      review={\MR{MR0377223 (51 \#13396)}},
}

\bib{MR1626871}{article}{
      author={Adler, M.},
      author={Horozov, E.},
      author={van Moerbeke, P.},
       title={The solution to the {$q$}-{K}d{V} equation},
        date={1998},
        ISSN={0375-9601},
     journal={Phys. Lett. A},
      volume={242},
      number={3},
       pages={139\ndash 151},
         url={http://dx.doi.org/10.1016/S0375-9601(98)00082-6},
      review={\MR{1626871 (2000b:37069)}},
}

\bib{MR1794352}{article}{
      author={Adler, Mark},
      author={van Moerbeke, Pierre},
       title={Integrals over classical groups, random permutations, {T}oda and
  {T}oeplitz lattices},
        date={2001},
        ISSN={0010-3640},
     journal={Comm. Pure Appl. Math.},
      volume={54},
      number={2},
       pages={153\ndash 205},
      review={\MR{MR1794352 (2003i:34198)}},
}

\bib{Aratyn:1993ry}{article}{
      author={Aratyn, H.},
      author={Nissimov, E.},
      author={Pacheva, S.},
       title={{Construction of KP hierarchies in terms of finite number of
  fields and their Abelianization}},
        date={1993},
     journal={Phys.Lett.},
      volume={B314},
       pages={41\ndash 51},
      eprint={hep-th/9306035},
}

\bib{Aspinwall:1991ce}{article}{
      author={Aspinwall, Paul~S.},
      author={Morrison, David~R.},
       title={{Topological field theory and rational curves}},
        date={1993},
     journal={Commun. Math. Phys.},
      volume={151},
       pages={245\ndash 262},
      eprint={hep-th/9110048},
}

\bib{Bonora:1993fa}{article}{
      author={Bonora, L.},
      author={Xiong, C.S.},
       title={{The $(N^{\rm th}, M^{\rm th})$ KdV hierarchy and the associated
  W-algebra}},
        date={1994},
     journal={J.Math.Phys.},
      volume={35},
       pages={5781\ndash 5819},
      eprint={hep-th/9311070},
}

\bib{Brini:2011ij}{article}{
      author={Brini, Andrea},
       title={{Open topological strings and integrable hierarchies: Remodeling
  the A-model}},
        date={2012},
     journal={Commun.Math.Phys.},
      volume={312},
       pages={735\ndash 780},
      eprint={1102.0281},
}

\bib{Brini:2010ap}{article}{
      author={Brini, Andrea},
       title={{The local Gromov-Witten theory of $\Bbb CP\sp 1$ and integrable
  hierarchies}},
        date={2012},
     journal={Commun.Math.Phys.},
      volume={313},
       pages={571\ndash 605},
      eprint={1002.0582},
}

\bib{Brini:2011ff}{article}{
      author={Brini, Andrea},
      author={Carlet, Guido},
      author={Rossi, Paolo},
       title={{Integrable hierarchies and the mirror model of local
  $\mathbb{CP}^1$}},
        date={2012},
     journal={Physica D},
      volume={241},
       pages={2156\ndash 2167},
      eprint={1105.4508},
}

\bib{Brini:2013zsa}{article}{
      author={Brini, Andrea},
      author={Cavalieri, Renzo},
      author={Ross, Dustin},
       title={{Crepant resolutions and open strings}},
        date={2013},
      eprint={1309.4438},
}

\bib{MR979202}{article}{
      author={Bruschi, M.},
      author={Ragnisco, O.},
       title={Lax representation and complete integrability for the periodic
  relativistic {T}oda lattice},
        date={1989},
        ISSN={0375-9601},
     journal={Phys. Lett. A},
      volume={134},
      number={6},
       pages={365\ndash 370},
         url={http://dx.doi.org/10.1016/0375-9601(89)90736-6},
      review={\MR{979202 (90a:58066)}},
}

\bib{MR2483931}{article}{
      author={Bryan, Jim},
      author={Graber, Tom},
       title={The crepant resolution conjecture},
        date={2009},
     journal={Proc. Sympos. Pure Math.},
      volume={80},
       pages={23\ndash 42},
      review={\MR{2483931 (2009m:14083)}},
}

\bib{MR1832332}{article}{
      author={Bryan, Jim},
      author={Katz, Sheldon},
      author={Leung, Naichung~Conan},
       title={Multiple covers and the integrality conjecture for rational
  curves in {C}alabi-{Y}au threefolds},
        date={2001},
        ISSN={1056-3911},
     journal={J. Algebraic Geom.},
      volume={10},
      number={3},
       pages={549\ndash 568},
      review={\MR{1832332 (2002j:14047)}},
}

\bib{MR2998900}{article}{
      author={Buryak, Alexandr},
      author={Posthuma, Hessel},
      author={Shadrin, Sergey},
       title={A polynomial bracket for the {D}ubrovin-{Z}hang hierarchies},
        date={2012},
        ISSN={0022-040X},
     journal={J. Differential Geom.},
      volume={92},
      number={1},
       pages={153\ndash 185},
         url={http://projecteuclid.org/euclid.jdg/1352211225},
      review={\MR{2998900}},
}

\bib{MR2534519}{article}{
      author={Cafasso, Mattia},
       title={Matrix biorthogonal polynomials on the unit circle and
  non-abelian {A}blowitz-{L}adik hierarchy},
        date={2009},
        ISSN={1751-8113},
     journal={J. Phys. A},
      volume={42},
      number={36},
       pages={365211, 20},
      review={\MR{2534519}},
}

\bib{GC}{article}{
      author={Carlet, Guido},
       title={The {H}amiltonian structures of the two-dimensional {T}oda
  lattice and {$R$}-matrices},
        date={2005},
        ISSN={0377-9017},
     journal={Lett. Math. Phys.},
      volume={71},
      number={3},
       pages={209\ndash 226},
      review={\MR{2141468 (2006m:37091)}},
}

\bib{MR2246697}{article}{
      author={Carlet, Guido},
       title={The extended bigraded {T}oda hierarchy},
        date={2006},
        ISSN={0305-4470},
     journal={J. Phys. A},
      volume={39},
      number={30},
       pages={9411\ndash 9435},
         url={http://dx.doi.org/10.1088/0305-4470/39/30/003},
      review={\MR{2246697 (2007g:37066)}},
}

\bib{MR2753798}{article}{
      author={Carlet, Guido},
      author={Dubrovin, Boris},
      author={Mertens, Luca~Philippe},
       title={Infinite-dimensional {F}robenius manifolds for {$2+1$} integrable
  systems},
        date={2011},
        ISSN={0025-5831},
     journal={Math. Ann.},
      volume={349},
      number={1},
       pages={75\ndash 115},
         url={http://dx.doi.org/10.1007/s00208-010-0509-3},
      review={\MR{2753798 (2012b:53196)}},
}

\bib{MR2108440}{article}{
      author={Carlet, Guido},
      author={Dubrovin, Boris},
      author={Zhang, Youjin},
       title={The extended {T}oda hierarchy},
        date={2004},
        ISSN={1609-3321},
     journal={Mosc. Math. J.},
      volume={4},
      number={2},
       pages={313\ndash 332, 534},
      review={\MR{2108440 (2005k:37168)}},
}

\bib{MR2452424}{article}{
      author={Chang, Jen-Hsu},
       title={Remarks on the waterbag model of dispersionless {T}oda
  hierarchy},
        date={2008},
        ISSN={1402-9251},
     journal={J. Nonlinear Math. Phys.},
      volume={15},
      number={suppl. 3},
       pages={112\ndash 123},
         url={http://dx.doi.org/10.2991/jnmp.2008.15.s3.12},
      review={\MR{2452424 (2009m:35416)}},
}

\bib{MR2510741}{article}{
      author={Coates, Tom},
      author={Corti, Alessio},
      author={Iritani, Hiroshi},
      author={Tseng, Hsian-Hua},
       title={Computing genus-zero twisted {G}romov-{W}itten invariants},
        date={2009},
        ISSN={0012-7094},
     journal={Duke Math. J.},
      volume={147},
      number={3},
       pages={377\ndash 438},
         url={http://dx.doi.org/10.1215/00127094-2009-015},
      review={\MR{MR2510741}},
}

\bib{MR2529944}{article}{
      author={Coates, Tom},
      author={Iritani, Hiroshi},
      author={Tseng, Hsian-Hua},
       title={Wall-crossings in toric {G}romov-{W}itten theory. {I}. {C}repant
  examples},
        date={2009},
        ISSN={1465-3060},
     journal={Geom. Topol.},
      volume={13},
      number={5},
       pages={2675\ndash 2744},
         url={http://dx.doi.org/10.2140/gt.2009.13.2675},
      review={\MR{2529944 (2010i:53173)}},
}

\bib{MR1677117}{book}{
      author={Cox, David~A.},
      author={Katz, Sheldon},
       title={Mirror symmetry and algebraic geometry},
      series={Mathematical Surveys and Monographs},
   publisher={American Mathematical Society},
     address={Providence, RI},
        date={1999},
      volume={68},
        ISBN={0-8218-1059-6},
      review={\MR{MR1677117 (2000d:14048)}},
}

\bib{MR2105635}{article}{
      author={Degiovanni, Luca},
      author={Magri, Franco},
      author={Sciacca, Vincenzo},
       title={On deformation of {P}oisson manifolds of hydrodynamic type},
        date={2005},
        ISSN={0010-3616},
     journal={Comm. Math. Phys.},
      volume={253},
      number={1},
       pages={1\ndash 24},
         url={http://dx.doi.org/10.1007/s00220-004-1190-8},
      review={\MR{2105635 (2005k:53150)}},
}

\bib{MR849651}{article}{
      author={Deligne, P.},
      author={Mostow, G.~D.},
       title={Monodromy of hypergeometric functions and nonlattice integral
  monodromy},
        date={1986},
        ISSN={0073-8301},
     journal={Inst. Hautes \'Etudes Sci. Publ. Math.},
      number={63},
       pages={5\ndash 89},
         url={http://www.numdam.org/item?id=PMIHES_1986__63__5_0},
      review={\MR{849651 (88a:22023a)}},
}

\bib{Donovan:2013hra}{article}{
      author={Donovan, Will},
      author={Segal, Ed},
       title={{Mixed braid group actions from deformations of surface
  singularities}},
        date={2013},
      eprint={1310.7877},
}

\bib{Dubrovin:1994hc}{article}{
      author={Dubrovin, Boris},
       title={Geometry of {$2$}{D} topological field theories},
        date={1994},
     journal={{\rm in} ``Integrable systems and quantum groups'' ({M}ontecatini
  {T}erme, 1993), Lecture Notes in Math.},
      volume={1620},
       pages={120\ndash 348},
      eprint={hep-th/9407018},
      review={\MR{MR1397274 (97d:58038)}},
}

\bib{Dubrovin:1998fe}{incollection}{
      author={Dubrovin, Boris},
       title={{Painleve transcendents and two-dimensional topological field
  theory}},
        date={1999},
   booktitle={The {P}ainlev\'e property},
      series={CRM Ser. Math. Phys.},
   publisher={Springer},
     address={New York},
       pages={287\ndash 412},
      review={\MR{1713580 (2001h:53131)}},
}

\bib{MR2070050}{article}{
      author={Dubrovin, Boris},
       title={On almost duality for {F}robenius manifolds, in ``{G}eometry,
  topology, and mathematical physics"},
        date={2004},
     journal={Amer. Math. Soc. Transl. Ser. 2},
      volume={212},
       pages={75\ndash 132},
      review={\MR{2070050 (2005e:53142)}},
}

\bib{MR2462355}{article}{
      author={Dubrovin, Boris},
       title={On universality of critical behaviour in {H}amiltonian {PDE}s},
        date={2008},
     journal={{\rm in} ``Geometry, topology, and mathematical physics", Amer.
  Math. Soc. Transl. Ser. 2},
      volume={224},
       pages={59\ndash 109},
      review={\MR{MR2462355}},
}

\bib{Dubrovin:1997bv}{article}{
      author={Dubrovin, Boris},
      author={Zhang, Youjin},
       title={{Bihamiltonian hierarchies in 2D topological field theory at
  one-loop approximation}},
        date={1998},
     journal={Commun. Math. Phys.},
      volume={198},
       pages={311\ndash 361},
      eprint={hep-th/9712232},
}

\bib{dubrovin:2001}{article}{
      author={Dubrovin, Boris},
      author={Zhang, Youjin},
       title={Normal forms of hierarchies of integrable pdes, frobenius
  manifolds and gromov--witten invariants},
        date={2001},
      eprint={math/0108160},
}

\bib{Eguchi:1994in}{article}{
      author={Eguchi, Tohru},
      author={Yang, Sung-Kil},
       title={{The Topological $\mathbb{CP}^1$ model and the large $N$ matrix
  integral}},
        date={1994},
     journal={Mod. Phys. Lett.},
      volume={A9},
       pages={2893\ndash 2902},
      eprint={hep-th/9407134},
}

\bib{MR0422713}{book}{
      author={Exton, Harold},
       title={Multiple hypergeometric functions and applications},
   publisher={Ellis Horwood Ltd.},
     address={Chichester},
        date={1976},
        note={Foreword by L. J. Slater, Mathematics \& its Applications},
      review={\MR{0422713 (54 \#10699)}},
}

\bib{MR1383952}{article}{
      author={Frenkel, Edward},
       title={Deformations of the {K}d{V} hierarchy and related soliton
  equations},
        date={1996},
        ISSN={1073-7928},
     journal={Internat. Math. Res. Notices},
      number={2},
       pages={55\ndash 76},
         url={http://dx.doi.org/10.1155/S1073792896000062},
      review={\MR{1383952 (97d:58094)}},
}

\bib{Gerasimov:1990is}{article}{
      author={Gerasimov, A.},
      author={Marshakov, A.},
      author={Mironov, A.},
      author={Morozov, A.},
      author={Orlov, A.},
       title={{Matrix models of 2-D gravity and Toda theory}},
        date={1991},
     journal={Nucl.Phys.},
      volume={B357},
       pages={565\ndash 618},
}

\bib{MR1885831}{article}{
      author={Getzler, Ezra},
       title={A {D}arboux theorem for {H}amiltonian operators in the formal
  calculus of variations},
        date={2002},
        ISSN={0012-7094},
     journal={Duke Math. J.},
      volume={111},
      number={3},
       pages={535\ndash 560},
         url={http://dx.doi.org/10.1215/S0012-7094-02-11136-3},
      review={\MR{1885831 (2003e:32026)}},
}

\bib{MR1091265}{incollection}{
      author={Gibbons, John},
      author={Kupershmidt, Boris~A.},
       title={Relativistic analogs of basic integrable systems},
        date={1990},
   booktitle={Integrable and superintegrable systems},
   publisher={World Sci. Publ., Teaneck, NJ},
       pages={207\ndash 231},
      review={\MR{1091265 (92b:58089)}},
}

\bib{MR1408320}{article}{
      author={Givental, Alexander},
       title={Equivariant {G}romov-{W}itten invariants},
        date={1996},
        ISSN={1073-7928},
     journal={Internat. Math. Res. Notices},
      number={13},
       pages={613\ndash 663},
         url={http://dx.doi.org/10.1155/S1073792896000414},
      review={\MR{1408320 (97e:14015)}},
}

\bib{MR1653024}{article}{
      author={Givental, Alexander},
       title={A mirror theorem for toric complete intersections},
        date={1998},
      volume={160},
       pages={141\ndash 175},
      review={\MR{1653024 (2000a:14063)}},
}

\bib{MR1901075}{article}{
      author={Givental, Alexander~B.},
       title={Gromov-{W}itten invariants and quantization of quadratic
  {H}amiltonians},
        date={2001},
        ISSN={1609-3321},
     journal={Mosc. Math. J.},
      volume={1},
      number={4},
       pages={551\ndash 568, 645},
        note={Dedicated to the memory of I. G. Petrovskii on the occasion of
  his 100th anniversary},
      review={\MR{MR1901075 (2003j:53138)}},
}

\bib{MR1666787}{article}{
      author={Graber, T.},
      author={Pandharipande, R.},
       title={Localization of virtual classes},
        date={1999},
        ISSN={0020-9910},
     journal={Invent. Math.},
      volume={135},
      number={2},
       pages={487\ndash 518},
         url={http://dx.doi.org/10.1007/s002220050293},
      review={\MR{1666787 (2000h:14005)}},
}

\bib{Hisakado:1996di}{article}{
      author={Hisakado, Masato},
       title={{Unitary matrix models and Painleve III}},
        date={1996},
     journal={Mod.Phys.Lett.},
      volume={A11},
       pages={3001\ndash 3010},
      eprint={hep-th/9609214},
}

\bib{MR1461570}{incollection}{
      author={Hitchin, Nigel},
       title={Frobenius manifolds},
        date={1997},
   booktitle={Gauge theory and symplectic geometry ({M}ontreal, {PQ}, 1995)},
      series={NATO Adv. Sci. Inst. Ser. C Math. Phys. Sci.},
      volume={488},
   publisher={Kluwer Acad. Publ.},
     address={Dordrecht},
       pages={69\ndash 112},
        note={With notes by David Calderbank},
      review={\MR{1461570 (98f:58048)}},
}

\bib{johnson2008notes}{article}{
      author={Johnson, P},
      author={Pandharipande, R},
      author={Tseng, HH},
       title={Notes on local $\mathbb{P}^1$-orbifolds},
        date={2008},
      eprint={http://www.math.ethz.ch/$\sim$rahul/lPab.ps},
}

\bib{pauljohnsonthesis}{article}{
      author={Johnson, Paul},
       title={{Equivariant {G}romov--{W}itten theory of one dimensional
  stacks}},
        date={2009},
      eprint={0903.1068},
}

\bib{Karp:2005vq}{article}{
      author={Karp, Dagan},
      author={Liu, Chiu-Chu~Melissa},
      author={Marino, Marcos},
       title={{The Local Gromov-Witten invariants of configurations of rational
  curves}},
        date={2006},
        ISSN={1465-3060},
     journal={Geom. Topol.},
      volume={10},
       pages={115\ndash 168},
      eprint={math/0506488},
}

\bib{Kharchev:1996kh}{article}{
      author={Kharchev, S.},
      author={Mironov, A.},
      author={Zhedanov, A.},
       title={{Faces of relativistic Toda chain}},
        date={1997},
     journal={Int.J.Mod.Phys.},
      volume={A12},
       pages={2675\ndash 2724},
      eprint={hep-th/9606144},
}

\bib{Kontsevich:1994na}{article}{
      author={Kontsevich, Maxim},
       title={{Enumeration of rational curves via Torus actions, in { The
  moduli space of curves (Texel Island, 1994)}}},
        date={1994},
     journal={Progr. Math.},
      volume={129},
       pages={335\ndash 368},
      eprint={hep-th/9405035},
}

\bib{mikhailov1979integrability}{article}{
      author={Mikhailov, AV},
       title={Integrability of a two-dimensional generalization of the {T}oda
  chain},
        date={1979},
     journal={JETP Lett},
      volume={30},
      number={7},
       pages={414\ndash 418},
}

\bib{MR2309158}{article}{
      author={Milanov, Todor~E.},
       title={Hirota quadratic equations for the extended {T}oda hierarchy},
        date={2007},
        ISSN={0012-7094},
     journal={Duke Math. J.},
      volume={138},
      number={1},
       pages={161\ndash 178},
         url={http://dx.doi.org/10.1215/S0012-7094-07-13815-8},
      review={\MR{MR2309158 (2008a:14070)}},
}

\bib{MR2433616}{article}{
      author={Milanov, Todor~E.},
      author={Tseng, Hsian-Hua},
       title={The spaces of {L}aurent polynomials, {G}romov-{W}itten theory of
  {$\Bbb P\sp 1$}-orbifolds, and integrable hierarchies},
        date={2008},
        ISSN={0075-4102},
     journal={J. Reine Angew. Math.},
      volume={622},
       pages={189\ndash 235},
      review={\MR{MR2433616}},
}

\bib{MR2962392}{article}{
      author={Mimachi, Katsuhisa},
      author={Sasaki, Takeshi},
       title={Irreducibility and reducibility of {L}auricella's system of
  differential equations {$E\sb {\rm D}$} and the {J}ordan-{P}ochhammer
  differential equation {$E\sb {\rm JP}$}},
        date={2012},
        ISSN={1340-6116},
     journal={Kyushu J. Math.},
      volume={66},
      number={1},
       pages={61\ndash 87},
         url={http://dx.doi.org/10.2206/kyushujm.66.61},
      review={\MR{2962392}},
}

\bib{MineevWeinstein:2000ie}{article}{
      author={Mineev-Weinstein, Mark},
      author={Wiegmann, Paul~B.},
      author={Zabrodin, Anton},
       title={{Integrable structure of interface dynamics}},
        date={2000},
     journal={Phys.Rev.Lett.},
      volume={84},
       pages={5106\ndash 5109},
      eprint={nlin/0001007},
}

\bib{Nakatsu:2007dk}{article}{
      author={Nakatsu, Toshio},
      author={Takasaki, Kanehisa},
       title={{Melting crystal, quantum torus and Toda hierarchy}},
        date={2009},
     journal={Commun.Math.Phys.},
      volume={285},
       pages={445\ndash 468},
      eprint={0710.5339},
}

\bib{MR2199226}{article}{
      author={Okounkov, A.},
      author={Pandharipande, R.},
       title={The equivariant {G}romov-{W}itten theory of {${\bf P}\sp 1$}},
        date={2006},
        ISSN={0003-486X},
     journal={Ann. of Math. (2)},
      volume={163},
      number={2},
       pages={561\ndash 605},
      review={\MR{MR2199226 (2006j:14075)}},
}

\bib{Okounkov:2003sp}{article}{
      author={Okounkov, Andrei},
      author={Reshetikhin, Nikolai},
      author={Vafa, Cumrun},
       title={{Quantum Calabi-Yau and classical crystals}},
        date={2006},
     journal={Progr.Math.},
      volume={244},
       pages={597},
      eprint={hep-th/0309208},
}

\bib{pavlov2003tri}{article}{
      author={Pavlov, Maxim},
      author={Tsarev, Sergey},
       title={Tri-hamiltonian structures of {E}gorov systems of hydrodynamic
  type},
        date={2003},
     journal={Funct.~Anal.~Appl.},
      volume={37},
      number={1},
       pages={32\ndash 45},
}

\bib{MR2287835}{article}{
      author={Riley, Andrew},
      author={Strachan, Ian A.~B.},
       title={A note on the relationship between rational and trigonometric
  solutions of the {WDVV} equations},
        date={2007},
        ISSN={1402-9251},
     journal={J. Nonlinear Math. Phys.},
      volume={14},
      number={1},
       pages={82\ndash 94},
      eprint={nlin/0605005},
         url={http://dx.doi.org/10.2991/jnmp.2007.14.1.7},
      review={\MR{2287835 (2008c:53093)}},
}

\bib{2012arXiv1210.2312R}{article}{
      author={Romano, Stefano},
       title={{Frobenius structures on double {H}urwitz spaces}},
        date={2013},
     journal={Int.~Math.~Res.~Notices, {\tt doi:10.1093/imrn/rnt215}},
      eprint={1210.2312},
}

\bib{MR1993935}{book}{
      author={Suris, Yuri~B.},
       title={\textrm{The problem of integrable discretization: {H}amiltonian
  approach}},
      series={Progress in Mathematics},
   publisher={Birkh\"auser Verlag},
     address={Basel},
        date={2003},
      volume={219},
        ISBN={3-7643-6995-7},
      review={\MR{1993935 (2004h:37086)}},
}

\bib{MR0810626}{article}{
      author={Takasaki, Kanehisa},
       title={Initial value problem for the {T}oda lattice hierarchy, in
  "{G}roup representations and systems of differential equations" ({T}okyo,
  1982)},
        date={1984},
     journal={Adv. Stud. Pure Math.},
      volume={4},
       pages={139\ndash 163},
      review={\MR{810626 (88a:58099b)}},
}

\bib{Takasaki:2012ba}{article}{
      author={Takasaki, Kanehisa},
       title={{Old and new reductions of dispersionless Toda hierarchy}},
        date={2012},
     journal={SIGMA},
      volume={8},
       pages={102},
      eprint={1206.1151},
}

\bib{Takasaki:2013gja}{article}{
      author={Takasaki, Kanehisa},
       title={{A modified melting crystal model and the Ablowitz–Ladik
  hierarchy}},
        date={2013},
     journal={J.Phys.},
      volume={A46},
       pages={245202},
      eprint={1302.6129},
}

\bib{Takasaki:2013lqa}{article}{
      author={Takasaki, Kanehisa},
       title={{Generalized Ablowitz-Ladik hierarchy in topological string
  theory}},
        date={2013},
      eprint={1312.7184},
}

\bib{Takasaki:1991zs}{article}{
      author={Takasaki, Kanehisa},
      author={Takebe, Takashi},
       title={{S({D}iff)(2) Toda equation: Hierarchy, tau function and
  symmetries}},
        date={1991},
     journal={Lett.Math.Phys.},
      volume={23},
       pages={205\ndash 214},
      eprint={hep-th/9112042},
}

\bib{MR796577}{article}{
      author={Tsarev, S.~P.},
       title={Poisson brackets and one-dimensional {H}amiltonian systems of
  hydrodynamic type},
        date={1985},
        ISSN={0002-3264},
     journal={Dokl. Akad. Nauk SSSR},
      volume={282},
      number={3},
       pages={534\ndash 537},
      review={\MR{796577 (87b:58030)}},
}

\bib{MR810623}{article}{
      author={Ueno, Kimio},
      author={Takasaki, Kanehisa},
       title={Toda lattice hierarchy, in ``{G}roup representations and systems
  of differential equations" ({T}okyo, 1982)},
        date={1984},
     journal={Adv. Stud. Pure Math.},
      volume={4},
       pages={1\ndash 95},
      review={\MR{810623 (88a:58099a)}},
}

\bib{MR1930577}{book}{
      author={Vassiliev, V.~A.},
       title={Applied {P}icard-{L}efschetz theory},
      series={Mathematical Surveys and Monographs},
   publisher={American Mathematical Society},
     address={Providence, RI},
        date={2002},
      volume={97},
        ISBN={0-8218-2948-3},
      review={\MR{1930577 (2003k:32043)}},
}

\bib{MR1421397}{article}{
      author={Voisin, Claire},
       title={A mathematical proof of a formula of {A}spinwall and {M}orrison},
        date={1996},
        ISSN={0010-437X},
     journal={Compositio Math.},
      volume={104},
      number={2},
       pages={135\ndash 151},
         url={http://www.numdam.org/item?id=CM_1996__104_2_135_0},
      review={\MR{1421397 (98h:14069)}},
}

\bib{MR1424469}{book}{
      author={Whittaker, E.~T.},
      author={Watson, G.~N.},
       title={A course of modern analysis},
      series={Cambridge Mathematical Library},
   publisher={Cambridge University Press},
     address={Cambridge},
        date={1996},
        ISBN={0-521-58807-3},
        note={An introduction to the general theory of infinite processes and
  of analytic functions; with an account of the principal transcendental
  functions, Reprint of the fourth (1927) edition},
      review={\MR{1424469 (97k:01072)}},
}

\bib{MR1804297}{article}{
      author={Yu, Lei},
       title={Waterbag reductions of the dispersionless discrete {KP}
  hierarchy},
        date={2000},
        ISSN={0305-4470},
     journal={J. Phys. A},
      volume={33},
      number={45},
       pages={8127\ndash 8138},
         url={http://dx.doi.org/10.1088/0305-4470/33/45/309},
      review={\MR{1804297 (2001k:37120)}},
}

\end{biblist}
\end{bibdiv}

\end{document}